\newif \ifproofs
\proofstrue

\documentclass[11pt]{article}
\usepackage{amssymb}
\usepackage{mathrsfs,amsmath}
\usepackage{graphicx}
\usepackage{amsmath}
\usepackage{amsthm}
\usepackage{bbm}
\usepackage{dsfont}
\usepackage{listing}
\usepackage{hyperref}
\usepackage{color}
\usepackage[ruled, vlined,linesnumbered, noend]{algorithm2e}
\usepackage{comment}
\usepackage{enumitem}
\usepackage{booktabs}
\usepackage{tabulary}
\usepackage{soul}

\theoremstyle{definition}

\usepackage{bm}

\newif\ifarxiv   
\arxivfalse 

\ifarxiv
\else

\fi

\usepackage{xcolor}

\colorlet{blue}{black}

\colorlet{green}{black}

\colorlet{purple}{black}

\usepackage[natbibapa]{apacite}

\def\s{\bm{s}}

\def\n{n} 
\def\c{c} 
\def\k{k} 
\def\a{a} 
\def\m{m} 
\def\s{s} 
\def\S{\mathcal{S}} 
\def\B{B} 
\def\q{\bm{q}} 
\def\f{f} 
\def\l{l} 
\def\sigmaa{\boldsymbol{\sigma}} 
\def\ppi{\boldsymbol{\pi}} 
\def\0{\bm{0}}

\newcommand{\bigpar}[1]{\left( #1 \right)}
\newcommand{\bigbra}[1]{\left[ #1 \right]}
\newcommand{\bigbrace}[1]{\left\{ #1 \right\}}

\newcommand{\casewise}[1]{\left\{ #1 \right.}

\usepackage{graphicx}
\usepackage{csvsimple}
\usepackage{rotating}
\usepackage{hyperref}
\hypersetup{
    colorlinks=true,
    linkcolor=blue,
    citecolor=red,
    filecolor=magenta,      
    urlcolor=cyan
    }

\usepackage[nocomma,long]{optidef}
\usepackage{listings}
\usepackage{comment}
\usepackage{appendix}
\usepackage[ruled,vlined]{algorithm2e}
\usepackage{amsfonts} 

\usepackage{soul}
\usepackage[normalem]{ulem}

\usepackage{multicol}
\usepackage{enumitem}
\usepackage{amssymb}
\usepackage{bm}
\usepackage{booktabs}
\usepackage{tabulary}
\usepackage{hyperref}

\usepackage{color} 
\definecolor{mygreen}{RGB}{28,172,0} 
\definecolor{mylilas}{RGB}{170,55,241}
\DeclareFixedFont{\ttb}{T1}{txtt}{bx}{n}{12} 
\DeclareFixedFont{\ttm}{T1}{txtt}{m}{n}{12}  
\usepackage{bbm}
\usepackage{amsmath,amsthm}

\usepackage{subcaption}

\usepackage{hyperref}

\newtheorem{theorem}{Theorem}
\newtheorem{corollary}{Corollary}
\newtheorem{proposition}{Proposition}

\newtheorem{lemma}{Lemma}
\newtheorem{observation}{Observation}

\newtheorem{assumption}{Assumption}

\theoremstyle{definition}
\newtheorem{definition}{Definition}
\newtheorem{remark}{Remark}
\newtheorem{example}{Example}

\newenvironment{hproof}{%
  \proof}{\endproof}

\usepackage{tikz,pgfplots}
\pgfplotsset{compat=1.14}
\pgfplotsset{scaled y ticks=false}

\usepackage{color}
\definecolor{deepblue}{rgb}{0,0,0.5}
\definecolor{deepred}{rgb}{0.6,0,0}
\definecolor{deepgreen}{rgb}{0,0.5,0}

\usepackage{listings}

\allowdisplaybreaks

\usepackage{fancyhdr}
\usepackage{calc}
\fancyheadoffset[RE]{\marginparsep+\marginparwidth}

\usepackage[top=1.5in, bottom=1in, left=1in, right=1in]{geometry}

\graphicspath{ {fig/} }

\title{Combatting Misreporting Fraud in Government Benefits Programs}

\author{%
\begin{tabular}{@{}ccc@{}}
    Devansh Jalota$^{1}$\footnote{Devansh Jalota and Matthew Tsao contributed equally.} &
    Matthew Tsao$^{2}$\footnotemark[\value{footnote}] &
    Marco Pavone$^{3}$ \\
    {\scriptsize\texttt{djalota3@gatech.edu}} &
    {\scriptsize\texttt{mattxtsao@gmail.com}} &
    {\scriptsize\texttt{pavone@stanford.edu}}
\end{tabular}
\\[0.7em]
\parbox{0.92\textwidth}{\centering\footnotesize
$^1$ H. Milton Stewart School of Industrial and Systems Engineering,
Georgia Institute of Technology, Atlanta, GA, USA \\
$^2$ Lyft Inc., San Francisco, CA, USA \\
$^3$ Department of Aeronautics and Astronautics,
Stanford University, Stanford, CA, USA 
}
}

\begin{document}
\maketitle

\begin{abstract}
    Benefits programs are policy tools for providing targeted aid, such as food assistance or transit subsidies. However, these programs are often susceptible to \emph{misreporting fraud}, wherein users misreport private attributes to obtain excess aid. To address this problem, we introduce a two-stage audit game between an administrator and users in which the administrator can audit users at a cost and levy fines on misreporting users, with player payoff structures designed to reflect real-world benefits programs. 
    We show that the administrator's least-favorable Bayesian Stackelberg equilibria can be computed tractably via linear programming and derive prescriptive bounds on the administrator’s excess payments and users’ misreporting probability as functions of audit costs and fines. 
    We further establish that, with an appropriately chosen budget level, the reduction in excess payments in any least-favorable Bayesian Stackelberg equilibrium offsets the audit budget required to sustain the mechanism, yielding total costs no greater than those of a baseline without an auditing authority, which we corroborate through a case study of Washington D.C.’s federal transit benefits program.
\end{abstract}


\section{Introduction}
\vspace{-2pt}

\color{blue}
Government benefits (welfare) programs are significant policy tools for improving access to essential goods and services by directly subsidizing their cost for eligible users in the population. Examples include food assistance programs~\citep{food-stamp-2002} that enable low-income households to purchase food and federal transit benefits programs~\citep{kutz2007federal} that provide employees with a transit benefit to cover commuting costs. These programs differ in their design, ranging from direct cash transfers~\citep{cash-transfers-website} that provide users with unrestricted spending power to credit or voucher-based systems~\citep{food-stamp-2002,kutz2007federal} that constrain users' spending to designated goods or services. While these programs differ in their design, the unifying feature of such programs is to reduce the effective out-of-pocket expenditures of benefits program recipients by directly offsetting their cost of accessing essential resources.
\color{black}

Despite their appeal, benefits programs often fall short of their goals due to fraud~\citep{lange1979fraud}. This work focuses on \emph{misreporting fraud}, wherein users misreport their attributes to gain more aid than they are entitled to, which arises as the eligibility for benefits programs typically depends on unobservable private user information.
For instance, millions of ineligible users in India have ration cards~\citep{ration-card-millions} to obtain subsidized food, as ration card eligibility depends on concealable traits, e.g., whether a user owns a motorized vehicle. Further, in transit benefits programs, users often misreport commuting costs to obtain more transit credits, e.g., in Washington D.C. in 2006, at least \$17 million of fraudulent benefits were claimed~\citep{kutz2007federal}. 



Consequently, there is a need to design mechanisms to deter users' manipulation of private information. One such mechanism is that of conducting user audits that reveal the true type of each user. \color{blue}
However, despite their appeal, audit-based mechanisms face several critical challenges. From a computational perspective, solving for equilibria in audit games is typically intractable, often requiring the solution of NP-hard problems~\citep{estornell2021incentivizing} or multiple large-scale optimization problems~\citep{blocki2013audit}. Beyond computational intractability, auditing in benefits program contexts presents additional practical limitations. Audits are costly to perform, and unlike many audit settings, such as those involving financial institutions monitoring small cohorts of grantees, benefits programs operate at massive scales, often serving thousands or even millions of users; hence, administrators can only audit a small fraction of recipients. Moreover, in benefits programs, the cost of auditing is often comparable to the value of the disbursed benefits, necessitating a payoff structure that explicitly incorporates audit costs into the administrator’s objective, which departs from much of the auditing literature~\citep{estornell2023incentivizing}.

Motivated by these practical considerations, we study a natural audit game that mirrors the operational realities of benefits programs and leverage our problem’s structure to transform the original bi-level equilibrium computation problem into a linear program that can be solved in polynomial time in the problem parameters. Moreover, while auditing every user is infeasible, we show that even a modest audit budget can effectively deter misreporting. While we present our model and analysis in the language of benefits programs, our results may apply to other settings where misreporting fraud may be a concern, e.g., tax evasion~\citep{ALLINGHAM1972323}.

\color{black}

\textbf{Our Contributions:} This work introduces an audit mechanism to address misreporting fraud in government benefits programs. Within this audit mechanism, the administrator running the program can audit users at some cost and levy a fine against a user if the audit reveals that a user misreported their information (or type). The goal of this audit mechanism is to reduce the probability users misreport, thus decreasing the excess payments to users beyond the predefined welfare aid allocations specified by the administrator's eligibility criteria, all while keeping audit costs low. Our contributions are four-fold:
\color{blue}
\vspace{-5pt}
\begin{enumerate}[wide, labelwidth=2pt, labelindent=15pt]
    \item \emph{Modeling:} We formulate an audit game and design the administrator's payoff structure in Section~\ref{sec:mainModel} to reflect the operational features of benefits programs, providing a micro-founded representation of misreporting fraud and the role of auditing in such contexts. In our audit game, utility-maximizing users \color{purple} report \color{black} their type to the administrator, who then decides whether to audit each user and, based on the audit outcome, allocates aid according to the user’s reported or verified eligibility. The administrator’s objective, motivated by the practical features of benefits programs, balances the cost of auditing against the reduction in excess payments from deterring fraud. \vspace{-5pt}
    \item \emph{Computationally Tractable Equilibrium Characterization:} In our studied audit game, we establish the existence of \color{purple} Bayesian Stackelberg equilibria with privately informed senders \color{black} and show that the \color{purple} administrator's least-favorable \color{black} equilibria can be computed via a linear program that can be solved in polynomial time in the problem parameters. We first establish this result in the budget-unconstrained setting (Section~\ref{sec:misreportingSingle}), where the administrator has enough budget to audit all users. The tractability of equilibrium computation arises from the structured design of the audit rules and the induced player payoffs, under which the administrator’s least-favorable \color{purple} Stackelberg \color{black} equilibria correspond to the sender-preferred subgame-perfect equilibria of the associated Bayesian persuasion game, where users know only their type distribution rather than their realized types. We emphasize that while \color{purple} computing Stackelberg equilibria is computationally challenging in general~\citep{korzhyk2010complexity}, \color{black} we leverage our problem's structure to transform the equilibrium computation problem to a polynomial-time solvable linear program. In this setting, the audit mechanism plays the role of a \textit{deterrent} in the induced equilibrium. Even though an audit mechanism exists, no audits are conducted at equilibrium because it is the recipients' best response to keep their misreporting rate sufficiently low to avoid suspicion from the administrator. \color{green} While no audits are conducted at equilibrium in our studied static audit game, our framework also \color{purple} provides guarantees under other natural equilibrium notions and admits extensions in which audits occur with positive probability. In particular, we show that our LP provides bounds on misreporting and excess payments under classical signaling equilibria, which may involve positive auditing (Appendix~\ref{app:classical_signaling_bound}). Further, \color{black} in Appendix~\ref{apdx:uniform-random-costs}, we show that a stochastic cost variant of the static audit game, where the realized audit cost is observed by the administrator but users know only its distribution, can generate positive audit probabilities. When audit costs are uniformly distributed, the resulting sender-preferred Bayesian persuasion benchmark is computable via a second-order cone program and provides a tractable upper bound for the corresponding user payoffs in the Stackelberg model. Thus, our studied deterministic cost static audit game and its LP-based equilibrium characterization serve as a transparent building block for richer settings in which audits occur with positive probability. \vspace{-5pt}

    We also extend the equilibrium existence and LP-characterization results \color{black} to the setting when the administrator is budget-constrained, limiting the amount it can spend on audits, in Section~\ref{sec:multiRecipientMain}. While \color{purple} Stackelberg \color{black} equilibria may, in general, not exist (Proposition~\ref{prop:counter-eg-nonExistence}), we show that if the administrator's budget scales with the size of the largest user coalition, Stackelberg equilibria exist and can be computed using an analogous LP method as in the unconstrained budget setting (Theorems~\ref{thm:eqEx-suff-budget} and~\ref{thm:coalitions-ExEq}). Since user coalitions are typically small~\citep{dunbar-number}, our results imply that with a small audit budget, the administrator can significantly reduce misreporting fraud.  \vspace{-5pt}

    \item \emph{Prescriptive Insights on Misreporting Fraud Deterrence:} For the equilibria of the audit game, we derive upper bounds on both the probability that users misreport and the expected excess payments made by the administrator (Corollaries \ref{cor:misrepotingProbBound} and \ref{cor:equilibrium_suboptimality}) as a function of the audit parameters, i.e., the cost of an audit and the fine levied on users for misreporting. These results are \textit{prescriptive} in the sense that for any given tolerance level on the probability of misreporting and excess payments, Corollaries~\ref{cor:misrepotingProbBound} and \ref{cor:equilibrium_suboptimality} respectively give sufficient conditions on the audit cost and the fine under which the tolerances will not be violated in the resulting equilibrium. \vspace{-5pt}

    \item \emph{Cost Effectiveness and Empirical Validation:} In Section~\ref{sec:mainCompAudit}, we evaluate the overall cost-effectiveness of the proposed audit mechanism and show that the reduction in misreporting fraud it achieves offsets the administrator’s expenditure on audits. Specifically, the total program cost, comprising both the audit budget and the excess payments, is no higher than in the status quo without an auditing authority: (i) for all parameter ranges when there are two user types, and (ii) under a mild condition on the audit fine when there are more than two types. To complement the theoretical results, we conduct a numerical case study of a transit benefits program, which demonstrates the practical efficacy of our mechanism in substantially mitigating misreporting fraud at a low overall cost in a real-world context. The numerical results further show that the largest gains from auditing occur in regimes where the potential for misreporting fraud is greatest.
\end{enumerate}
\vspace{-10pt}
\color{black}

\color{blue}

In the appendix, we provide omitted proofs, present additional theoretical details to emphasize certain concepts, and provide additional numerical results. 





\vspace{-4pt}
\section{Related Literature} \label{sec:literature}
\vspace{-2pt}

Informational asymmetries are present in many markets with strategic agents~\citep{milgrom-weber-1982,milgrom-roberts,spence-signaling,akerlof-1970,info-asym-leland-pyle}. Akin to this line of work, we also study a context where one group of agents has an informational advantage~\citep{holmstrom-moral-hazard}. \color{purple} However, unlike classical models of costly signaling, our model does not impose an intrinsic cost on reporting; instead, reporting incentives arise from the benefits associated with a report and the possibility of audits and fines. While our model shares the underlying private-information structure of strategic communication models~\citep{gneezy-2005,strategic-info-transmission}, with payoffs tailored to benefits programs, we study a different equilibrium notion. Specifically, we consider Bayesian Stackelberg equilibria, where privately informed sender types anticipate the administrator's optimal response to changes in their reporting strategies, rather than classical signaling equilibria, where senders optimize their reporting decisions against a fixed auditing strategy. While we focus on Bayesian Stackelberg equilibria, we also show that our LP formulation provides bounds on misreporting probabilities and excess payments under signaling equilibria. Moreover, unlike the non-algorithmic equilibrium characterization results in much of this literature, our main contribution lies in algorithms to compute equilibria in settings with and without audit budgets and user coalitions. 

Since this work studies a law enforcement application, it is closely related to the literature on security games~\citep{tambe2011security,pita2008deployed,kjtjft-2009}, where a law enforcement agency decides on an allocation of security resources to different locations to mitigate damage by an adversary that best responds by maximizing its utility. However, in contrast to typical security games, where the law enforcement agency first commits to a strategy to which the adversaries best respond, in our setting, the role of the players is reversed: \color{purple} privately informed users choose their reporting strategies while anticipating the administrator's optimal auditing response to the resulting reports. \color{black}



In this regard, unlike empirical or mechanism design approaches to reduce information manipulation~\citep{perez2023fraud,bjerre2023playing,RePEc:hal:psewpa:halshs-03828729,mcmichael2022stealing}, our work aligns with the auditing literature~\citep{henderson-audit}, which spans applications from resource allocation~\citep{lundy-taylor-etal-2019} and classification systems~\citep{estornell2021incentivizing,estornell2023incentivizing} to public administration, e.g., welfare programs~\citep{appelgren2019optimal}, tax evasion~\citep{ALLINGHAM1972323}, and environmental compliance~\citep{MACHOSTADLER2006110,YANIV1997435}. 
\color{blue} Our focus on auditing as the mechanism to deter fraud, rather than on classical mechanism design approaches that seek incentive-compatible reporting, reflects the operational realities of welfare programs, where fines for misreporting are typically fixed irrespective of a user's reported type and audits serve as the standard enforcement tool. \color{black}
While our model shares similarities with prior works on audit games (e.g., a two-stage structure), our work differs in the solution methodology and the payoff and information structures. \color{green} Moreover, while the deterrent role of audits has been studied extensively in this literature, our contribution is to show that, for the benefits program payoff structure studied here, the resulting worst-case \color{purple} Stackelberg \color{black} equilibria admit a tractable LP characterization, which extends naturally to settings with budget constraints and user coalitions.
\color{black}

Methodologically, unlike the Stackelberg security game approach to auditing~\citep{blocki2013audit} that develops multiple optimization problems to characterize equilibria, in our setting, equilibria can be computed using a single linear program that can be solved efficiently. Moreover, compared to the economic theory literature on optimal auditing in both static~\citep{optimal-auditing-qje,optimal-auditing-rand} and dynamic settings~\citep{timing-inspections-predictable-2023}, the force behind our results is different: we derive a linear programming characterization of the equilibria in our studied game. Finally, unlike strategic classification~\citep{estornell2021incentivizing,estornell2023incentivizing} and public administration contexts (e.g., optimal tax auditing~\citep{Macho–Stadler-2002,HASHIMZADE201010,PARAMONOVAKUCHUMOVA2017162,fellingham1985strategic,irs-audit-2022}), we study audit games in the context of benefits programs, which induces a different payoff structure for the administrator and users, inspired by fairness considerations and high volumes of small aid common in such settings. Consequently, the structure of the equilibria in our audit game is considerably different and requires different methodological tools than prior works on audit games (see Appendix~\ref{apdx:comp-audit-estornell} for a detailed comparison).

This work is also related to the seminal work of~\citet{becker1968crime}, as we investigate whether the decrease in the misreporting fraud of our audit mechanism outweighs the spending to run it. However, the interactions in our audit mechanism are strategic. 

Our work is also closely connected to the Bayesian persuasion literature~\citep{bergemann2019information,kamenica-review-2019,asquith2013effects,segal-2010-optimal-disclosure,zu-iyer-xu-2025}. While we also investigate equilibrium formation under the Bayesian persuasion model (see Section~\ref{sec:misreportingSingle}), unlike the seminal work of~\citet{kamenica2011bayesian} that re-expressed the problem of choosing an optimal signal as a search over a distribution of posteriors, we leverage our problem structure to directly compute the optimal \color{purple} reporting policy \color{black} using linear programming. 

\vspace{-3pt}
\section{Model and Preliminaries}\label{sec:mainModel}
\vspace{-2pt}

\color{blue}

In the absence of an auditing authority, the administrator lacks any mechanism to verify the private information reported by users. Consequently, misreporting fraud goes unchecked, and excess payments, i.e., transfers made by administrators to users beyond their true entitlements, can be substantial. To address this concern, we study a setting in which the administrator can conduct audits to verify users’ reports at a cost. Audits deter fraud and reduce excess payments but also impose administrative and operational expenses. Thus, the administrator faces a fundamental trade-off between the benefit of reduced fraud and the cost of conducting audits.

This section introduces the model of a static audit game between the administrator and users that captures this trade-off, motivated by the practical operations of government benefits programs. We first introduce notation (Section~\ref{sec:model:game_params}), the administrator and user payoffs (Section~\ref{sec:model:strategies_payoffs}), and the notion of misreporting fraud (Section~\ref{sec:model:fraud_definition}). Then, we present our equilibrium notion in Section~\ref{sec:model:interim_equilibria} and discuss modeling assumptions in Section~\ref{subsec:discussion-modeling-asmpns}.


\color{black}


\vspace{-5pt}
\subsection{Audit Game Parameters}\label{sec:model:game_params}
\vspace{-2pt}

\color{blue}
In our audit game, at a high level, a user \color{purple} reports \color{black} their eligibility (type) for welfare aid (e.g., cash transfers, credits, or vouchers), which we henceforth collectively refer to as \emph{credits} throughout the paper, to an administrator, who can either accept the \color{purple} report \color{black} at face value or conduct an audit to learn the user’s true type. \color{black} 
Credits are then allocated to the user based on their true or \color{purple} reported \color{black} type, depending on whether an audit did or did not occur, respectively. 

To further elucidate this audit game, we first introduce some notation. Let $\S$ denote the set of user types, which we assume has finite cardinality. In line with the literature, we assume all players have a shared prior on users' type distribution~\citep{strategic-info-transmission,kamenica2011bayesian}, given by the probability vector $\q \! = \! (q_{\m})_{\m \in \S}$, where $\q \! \geq \! \0$ and $\sum_{m \in \S} \q_{\m} \! = \! 1$. Without loss of generality, we assume that the prior distribution has full support, i.e., $\q_m>0$ for every $m\in\mathcal S$, since types with zero probability can be excluded from the model. The function $\f \!:\!\S \! \rightarrow \! \mathbb{R}_{\geq 0}$ denotes the mapping from a user type $\m \! \in \! \S$ to a credit allocation $\f(\m)$, corresponding to the credit allotment to each user based on the administrator's eligibility criteria. 

The administrator's primary tool to combat user misreports, e.g., to obtain more credits than they are entitled to, is through audits, which reveal a user's type. A cost of $c$ is required to conduct an audit, and the administrator can levy a fine $k$ on a user if an audit reveals a misreport. We assume $0 \! \leq c \! \leq \! k$ for auditing to be economically viable.

To further elucidate the model, we provide examples of the function $\f$ and state space $\S$. 
\begin{example} [Federal Transit Benefits Program (FTBP)] \label{ex:ftbp}
    In an FTBP, federal employees are given a monthly transit benefit equal to their commuting costs to work. That is, each element $\s$ in a state space $\S$ is a tuple $\s \! = \! (b_1, b_2)$, where $b_1$ denotes whether an individual is a federal employee and $b_2$ represents the user's monthly commuting costs. In this context, the function $\f$ is zero if a user is not a federal employee and $b_2$ otherwise, i.e., only federal employees receive credits. 
\end{example}
\begin{example} [Food Assistance Program] \label{ex:foodStamps}
    In food assistance programs, food stamps are distributed to users based on many factors, such as a user's income and household size, i.e., $\S$ corresponds to a set of tuples, each specifying a user's income and household size. For such programs, the function $\f$ can be computed using income qualification tables~\citep{food-bank-table}.
\end{example}

\vspace{-2pt}
\subsection{Player Strategies and Payoffs} \label{sec:model:strategies_payoffs}
\vspace{-2pt}

\color{blue}

In this section, we elucidate the strategies and payoffs of the administrator in our audit game. We assume that users are utility maximizers, consistent with the audit games literature~\citep{lundy-taylor-etal-2019}, and model an administrator objective that captures the trade-off between the cost of conducting audits and the reduction in excess payments from deterring fraud, reflecting the operational realities of benefits programs. To this end, following the prior auditing literature~\citep{estornell2023incentivizing}, we assume that the administrator has a budget $B\!> \!0$, representing a cap on total audit spending. However, modeling the influence of the audit costs merely as a budget constraint is inadequate in benefits program contexts as these programs often serve thousands or even millions of users, and the cost of auditing is typically comparable to the benefits being disbursed. To this end, we also explicitly incorporate audit expenditures into the administrator’s objective, a departure from much of the existing literature~\citep{estornell2023incentivizing}. 



For simplicity of exposition, given the additional complexity introduced by an audit budget constraint, we first present the budget-unconstrained setting, and then extend the analysis to the budget-constrained case in Section~\ref{sec:multiRecipientMain}. In the budget unconstrained setting, we consider a game with two players: one administrator and one user. The user has a type $m$ drawn from the distribution $\q$, which is private information, i.e., only the user knows $m$. The setting with one administrator without a budget constraint and $n$ non-colluding users is a straightforward generalization, as the absence of collusion means that the $n$ user setting is simply $n$ single user games happening in parallel. We also consider the setting with user collusion in Section~\ref{sec:coalitions}. In the following, we introduce the strategies and payoffs of both the user and the administrator. 

\color{black}


\noindent \textbf{User strategies:} The user sends a \color{purple} report \color{black} $s \in \S$ to the administrator to demonstrate its eligibility for credits and its strategy is specified by a conditional distribution of $s$ given $m$ for all $m$. Namely, the user's strategy has the form $\ppi := \bigbrace{\ppi(\cdot | m)}_{m \in \S}$, and the strategy set for the user is $\Pi = \bigbrace{ \bigbrace{\ppi(\cdot | m)}_{m \in \S} : \ppi(\cdot|m) \geq \0 \text{ and } \sum_{\s \in \S} \ppi(\s|\m) = 1 \text{ for all } m \in \S}.$ 


\noindent \textbf{Administrator Strategies:} Upon seeing the \color{purple} report \color{black} $s$, the administrator makes a decision $a(s)$ on whether to audit the user to learn its type $m$.
$a(s) = 1$ ($a(s) = 0$) represents the outcome that an audit occurred (no audit occurred). To allow for randomization, an administrator strategy is given by a function $\sigmaa : \S \rightarrow [0,1]$, where $\sigmaa(s)$ is understood as $\mathbb{P}(a(s) = 1)$. 


\noindent \textbf{Player payoffs and utilities:} 
\color{blue} We study an administrator (A) who minimizes the sum of excess payments and net audit costs, i.e., the cost of auditing minus any fines collected from audited users, and a utility-maximizing user (R) with the following payoffs:

\vspace{-6pt}
{\setlength{\abovedisplayskip}{2pt}
\setlength{\belowdisplayskip}{2pt}
\small
\begin{equation} \label{eq:admin-payoff}
  \Tilde{U}^{A}(\a(\s), \m) \! = \!
    \begin{cases}
       f(\m)-\f(\s) & \text{if $\a(\s) = 0$}\\
       -\c  & \text{if $\a(\s) = 1$, $\s=\m$}\\
       \k \! - \! \c \! + \! \f(\m) \! - \! \min( \f(\m), \f(\s))  & \text{if $\a(\s) = 1$, $\s \neq \m$}
    \end{cases}     
\end{equation}}
{\setlength{\abovedisplayskip}{2pt}
\setlength{\belowdisplayskip}{2pt} \small
\begin{equation} \label{eq:user-payoff}
  U^{R}(\a(\s), \m) \! = \!
    \begin{cases}
      \f(\s) & \text{if $\a(s) = 0$}\\
      \f(\m) & \text{if $\a(\s) = 1$, $\s=\m$}\\
      \min(\f(\m), \f(\s))-\k & \text{if $\a(\s) = 1$, $\s \neq \m$}.
    \end{cases}       
\end{equation}}

The rationale for these payoffs is as follows. If the administrator does not audit given a \color{purple} report \color{black} $\s$, it issues $\f(\s)$ credits to the user, incurring an excess payment of \color{purple} $[\f(\s) - \f(\m)]_+$. \color{black} If the administrator audits with a cost $\c$ and the user truthfully reports its type, no excess payments occur. Conversely, if the user misreports, the administrator additionally levies a fine $\k$ against the user. The payout $\min(\f(\s), \f(\m))$ ensures that the administrator never pays more than its intended credit allotment $\f(\m)$ when it conducts an audit, while if the user asks for too little, i.e., $\f(\s) \! < \! \f(\m)$, then the administrator only pays out $\f(\s)$, thus discouraging under-reporting. This discouragement of under-reporting 
prevents collusion across user types (see Appendix~\ref{apdx:collusion-model-asmpn}) and 
is for mathematical completeness to ensure utility-maximizing users never under-report (see Lemma~\ref{lem:underreporting_dominated}), a close representation of reality. 

For expositional convenience, we redefine the administrator’s payoff as $U^{A}(\a(\s), \m) = \Tilde{U}^{A}(\a(\s), \m) - \f(\m)$, which represents the \color{purple} negative of the \color{black} administrator’s total expenditure, comprising credit payments to users and net audit costs. This normalization is without loss of generality as \color{green} while $f(m)$ varies across users ex ante, for any realized user type $m$, it is fixed and does not depend on the report, the audit decision, or either player’s strategy. Equivalently, from the administrator’s ex ante perspective, subtracting $f(m)$ subtracts the constant $\sum_m q_m f(m)$ from expected utility. Hence, the transformation does not affect administrator best responses or equilibrium outcomes. While the game is not zero-sum, this transformation is analytically useful because it allows the administrator’s utility to be interpreted as the negative of the user’s utility, offset by the audit cost $\c$. \color{black}




Given the above defined payoffs, \color{black} the utility functions for the administrator and user are the expected payoffs as a function of the strategy profile $(\ppi, \sigmaa)$, where the averaging is done over the player's belief and the randomness of the player strategies, and are given by

\vspace{-15pt}
{\setlength{\abovedisplayskip}{0pt}
\setlength{\belowdisplayskip}{0pt}
\small
\setlength{\jot}{-0.3ex}
\begin{align}
    &\mathcal{U}^A(\ppi, \sigmaa) \!:= \! \mathbb{E}_{m \sim \q} \mathbb{E}_{\substack{s \sim \ppi(\cdot | m),\\ a \sim \text{Bern}(\sigmaa(s))}} \! \! \bigbra{ U^A \! (a(s), \! m) } = \! \! \sum_{m, s} \q_m \ppi(s | m) \! \! \bigbra{ \!- \! f(s) \! + \! \sigma(s) \! \bigpar{- c \! + \! \mathbb{I}_{[s \neq m]} \! \bigpar{\bigbra{\! f(s) \! - \! f(m)}_+ \! + k}}}.\label{eqn:admin_utility_fn} \\ 
    &\mathcal{U}^R_m (\ppi, \sigmaa) \! \! := \mathbb{E}_{s \sim \ppi(\cdot | m), a \sim \text{Bern}(\sigmaa(s)) } \! \! \bigbra{ U^R (a(s), m) } = \! \sum_{s \in \S} \ppi(s | m) \bigbra{f(s) - \sigma(s) \mathbb{I}_{[s \neq m]}\bigpar{\bigbra{f(s) - f(m)}_+ + k}}. \label{eqn:sg_user_utility_fn}
\end{align}}

Note the asymmetry between the utilities. $\mathcal{U}^A$ averages over $m$ as the administrator does not know $m$, while $\mathcal{U}^R_m$ does not average over $m$ as the user knows its type. For further details and derivations of these payoff functions, see Appendix~\ref{app:utility_derivations}.

\vspace{-4pt}

\subsection{Misreporting Fraud as Excess Payments}\label{sec:model:fraud_definition}
\vspace{-2pt}

The main objective of the audit mechanism is to deter misreporting fraud whereby users receive more credits than they are entitled to. We thus formally define misreporting fraud as the excess payments that the administrator pays out. 
\vspace{-2pt}
\begin{definition} [Misreporting Fraud] \label{def:misreporting-fraud}
Given a strategy profile $(\ppi, \sigmaa)$, the level of misreporting fraud is defined as the total excess payments, which is given by $$\mathcal{E}(\ppi, \sigmaa) := \mathbb{E}_{m \sim \q} \mathbb{E}_{s \sim \ppi(\cdot | m), a \sim \text{Bern}(\sigmaa(s))} \bigbra{ \max \bigpar{ U^R(a(s), m) - f(m), 0 } }.$$
\end{definition}
\color{blue} We establish in Section~\ref{subsubsec:misreporting-prob-excess-payment-bd} a bound on the equilibrium excess payments under our audit mechanism for the equilibrium solution concepts defined in the next section. \color{black}






\vspace{-4pt}
\subsection{Equilibrium Notions}\label{sec:model:interim_equilibria}
\vspace{-2pt}
The solution concept of interest is that of a \color{purple} Bayesian Stackelberg equilibrium with privately informed senders\color{black}, where each user knows their type $m$ drawn from a prior $\q$, and the administrator knows $\q$ \color{purple} and the users' reporting strategies, but not their realized types. \color{black} Such a model is most appropriate for the benefits program applications we study, as it assumes users know their type. We also study the solution concept of a sender-preferred subgame perfect (i.e., Bayesian persuasion) equilibrium, wherein users only know a prior distribution for their type. While less practical for benefits program applications, the latter solution concept will be useful in bounding the misreporting fraud in \color{purple} Stackelberg \color{black} equilibria.


\paragraph{\color{purple}Stackelberg \color{black} Equilibrium:} A strategy profile $(\ppi', \sigmaa')$ is a \color{purple} Stackelberg \color{black} equilibrium of the audit game if and only if (a) the administrator cannot improve its utility via a unilateral deviation, and (b) no single user of type $m$ can improve its utility provided that the administrator responds to the \color{purple} user's reporting strategy \color{black} so as to maximize $\mathcal{U}^A$. 
\begin{definition}[\color{purple}Stackelberg \color{black} Equilibria]\label{def:signalingGameEqOneRecipient}
A strategy profile $(\ppi', \sigmaa')$ is a \color{purple} Stackelberg \color{black} equilibrium if and only if:
{\setlength{\abovedisplayskip}{2pt}
\setlength{\belowdisplayskip}{2pt}
\setlength{\jot}{-0.3ex}
\begin{align*}
    \sigmaa' \in \underset{\sigmaa: \sigmaa(s) \in [0, 1], \forall s \in \S}{\text{argmax }} \mathcal{U}^A(\ppi', \sigmaa)  \quad \text{ and for all } m \in \S, \quad \ppi' \in &\underset{\ppi \in \Pi}{\text{ argmax }} \; \mathcal{U}^R_m(\ppi, \sigma) \\
    \text{s.t. } &\ppi(\cdot | m') = \ppi'(\cdot | m') \; \forall \; m' \neq m \\
                 & \sigmaa \in \underset{\sigmaa_0: \sigmaa_0(s) \in [0, 1], \forall s \in \S}{\text{argmax }} \mathcal{U}^A(\ppi, \sigmaa_0).
\end{align*}    }

\color{purple}
While we focus on the above Stackelberg equilibrium notion to capture settings where users anticipate the administrator's optimal auditing response to their reporting strategies, several of our guarantees also extend to other natural equilibrium notions, including signaling equilibria, where users optimize their reporting strategies taking the administrator's auditing strategy as fixed (see Remark~\ref{rmk:no-audit-eq} and Appendix~\ref{app:classical_signaling_bound}). \color{black}

\end{definition}

\paragraph{Sender-Preferred Subgame Perfect (Bayesian Persuasion) Equilibrium:}
The notion of Sender-Preferred Subgame Perfect equilibrium \citep{kamenica2011bayesian}, which we will henceforth refer to as Bayesian persuasion equilibrium, will be useful in analyzing \color{purple} Stackelberg \color{black} equilibria. In a Bayesian persuasion equilibrium, it is assumed that users do not know their true type and thus choose their \color{purple} reporting \color{black} strategy to maximize expected utility where the expectation now includes averaging over $m$, as given by:
{\setlength{\abovedisplayskip}{2pt}
\setlength{\belowdisplayskip}{2pt}
\begin{align}
    &\mathcal{U}^R (\ppi, \sigmaa) := \mathbb{E}_{m \sim \q} \bigbra{ \mathcal{U}^R_m (\ppi, \sigmaa) } = \! \! \sum_{m \in \S} \sum_{s \in \S} \q_m \ppi(s | m) \bigbra{f(s) \! - \! \sigmaa(s) \mathbb{I}_{[s \neq m]}\bigpar{\bigbra{f(s) \! - \! f(m)}_+ \! + \! k}}.\label{eqn:user_utility_fn}
\end{align}}


A profile $(\ppi^*, \sigmaa^*)$ is a Bayesian persuasion equilibrium if and only if (a) the administrator responds optimally based on its posterior distribution given the signal $s$ (i.e., subgame perfection), and (b) the user chooses a utility-maximizing strategy.  
\vspace{-2pt}
\begin{definition}[Bayesian Persuasion Equilibria]\label{def:sender-preferred-subgame-perfect-equilibria}
A strategy profile $(\ppi^*, \sigmaa^*)$ is a Bayesian persuasion equilibrium if and only if
{\setlength{\abovedisplayskip}{2pt}
\setlength{\belowdisplayskip}{2pt}
\begin{align*}
    \sigmaa^* \in \underset{\sigmaa: \sigmaa(s) \in [0, 1], \forall s \in \S}{\text{argmax }} \mathcal{U}^A(\ppi^*, \sigmaa), 
    \quad \text{and} \quad \ppi^* \in &\underset{\ppi \in \Pi}{\text{ argmax }} \mathcal{U}^R(\ppi, \sigmaa) \\
    &\quad \text{ s.t. } \sigmaa \in \underset{\sigmaa_0: \sigmaa_0(s) \in [0, 1], \forall s \in \S}{\text{argmax }} \mathcal{U}^A(\ppi, \sigmaa_0).
\end{align*}}
\end{definition}
Note that the \color{purple} Stackelberg \color{black} and Bayesian persuasion solution concepts are not equivalent. While every Bayesian persuasion equilibrium is a \color{purple} Stackelberg \color{black} equilibrium in our setting, the converse is not true (see Appendix~\ref{app:equilibria_nonequiv}). 

\vspace{-4pt}
\subsection{Discussion of Modeling Assumptions} \label{subsec:discussion-modeling-asmpns}
\vspace{-2pt}

\color{blue}
Consistent with the practical operations of welfare programs, we model an administrator that minimizes the sum of excess payments and audit costs. This formulation reflects the fact that welfare programs often serve thousands or even millions of users, and that the cost of auditing is typically comparable in magnitude to the benefits being disbursed. Consequently, audit expenditures must be directly incorporated into the administrator’s objective, and we therefore design an audit mechanism wherein the administrator minimizes the sum of excess payments and auditing costs. Notably, this payoff structure, under an equilibrium condition we derive, induces the administrator to minimize excess payments at equilibrium (see Remark~\ref{rem:UA_interpret} in Section~\ref{subsec:misreportingSingle:discussion}). Furthermore, under this formulation, the probability of user misreporting remains low (see Corollary~\ref{cor:misrepotingProbBound}), and the overall budget required to implement the audit mechanism is small (Section~\ref{sec:multiRecipientMain}), which may not be simultaneously achieved when minimizing excess payments alone. While we focus on an administrator that minimizes the sum of its excess payments and audit costs, we note that considering other valid administrator objectives is an important direction for future research. 

Furthermore, while benefits programs vary in design, from direct cash transfers to credit or voucher-based systems, we model welfare aid as a reduction in out-of-pocket payments, as aid issued as cash transfers or non-cash credits (e.g., transit/food vouchers) directly offset a user’s expense. From the administrator’s perspective, distributing aid likewise constitutes a resource transfer, either through direct cash payments or, in credit-based systems, through in-kind provisions such as food in food stamp programs. Hence, the administrator incurs costs proportional to the value of aid provided, a consideration we explicitly incorporate in the administrator’s objective.




\color{black}

Next, even though we focus on the payoffs in Equations~\eqref{eq:admin-payoff} and~\eqref{eq:user-payoff}, our results can be extended naturally to a broader class of payoffs, e.g., to the more user-punishing setting where users pay a fine $\k$ without receiving any credits when an audit successfully detects a misreport (i.e., when $\a(\s_i) = 1$ and $\s_i \neq \m_i$). We emphasize that Equations~\eqref{eq:admin-payoff} and~\eqref{eq:user-payoff} ensure that the administrator will pay out $f(\m_i)$ credits to all users $i$ under truthful reporting, thus capturing the intended outcome of the benefits program that seeks to deliver the appropriate quantity of benefits to eligible users. 

Further, in investigating the budget constrained setting, we also consider the possibility of fractional budgets, wherein the administrator may only be able to allocate a fraction of the audit cost $\c$ to audit some users or may not have enough budget to audit even one user (see Proposition~\ref{prop:counter-eg-nonExistence}). We interpret a fractional budget as the probability that an audit is conducted. \color{green} We also note that the budget constraint and audit cost play distinct roles in our framework. While the audit budget $B$ captures the administrator’s enforcement capacity, the audit cost $c$ captures the per-audit trade-off between fraud reduction and audit spending. As noted in Section~\ref{sec:model:strategies_payoffs}, we retain the cost of conducting audits in the administrator's objective because audit costs in benefits programs can be comparable to the value of the benefits disbursed; modeling the budget only as a feasibility constraint would miss this key economic trade-off. 

Finally, for expositional simplicity, we focus on a fixed fine and a fully strategic user population. We note that our framework and results extend to type- and report-dependent fines $k_{m,s}$ and can also accommodate partially strategic populations when non-strategic users follow an exogenous reporting rule that is common knowledge. In the latter case, non-strategic reports enter the administrator's posterior as constant terms, so the strategic users' equilibrium problem retains the same structure as the fully strategic case, with modified affine no-audit constraints. However, these modified constraints need not be feasible for arbitrary exogenous reporting rules, as sufficiently extensive non-strategic misreporting can make auditing profitable at some signals. Nevertheless, at equilibrium, any signal that a strategic recipient type uses with positive probability to misreport will not be audited by the administrator. We omit the details of these extensions for brevity.



\color{black}



\vspace{-4pt}
\section{\color{blue} Equilibrium Analysis: \color{black} Budget Unconstrained Setting} \label{sec:misreportingSingle}
\vspace{-2pt}

\color{blue}
In this and the next section, we analyze the equilibrium outcomes of the audit game, focusing on the existence and efficient computation of \color{purple} Stackelberg \color{black} equilibria. In this section, we study the setting where the administrator’s budget is large enough to audit all users, and in Section~\ref{sec:multiRecipientMain}, we extend the analysis to the budget-constrained case. \color{black}
In the budget-unconstrained setting, we establish the existence of a \color{purple} Stackelberg \color{black} equilibrium and show that the administrator’s least-favorable equilibria can be computed efficiently via a linear program solvable in polynomial time in the problem parameters. The key insight in establishing this result involves showing that the least favorable \color{purple} Stackelberg \color{black} equilibria for the administrator are precisely the sender-preferred subgame perfect equilibria of the corresponding Bayesian persuasion game, where users only know the distribution of their type. Our equilibrium analysis also provides upper bounds that hold in any \color{purple} Stackelberg \color{black} equilibrium on the expected excess payments and the probability users misreport as a function of the audit parameters $c$ and $k$. This section presents our equilibrium characterization results (Section~\ref{sec:misreportingSingle:main-results}), elucidates the main ingredients in establishing equilibrium existence (Section~\ref{subsec:misreportingSingle:key-ideas}), and interprets and discusses several implications of our results (Section~\ref{subsec:misreportingSingle:discussion}).

\vspace{-6pt}
\subsection{Main Results}\label{sec:misreportingSingle:main-results}
\vspace{-2pt}
This section characterizes \color{purple} Stackelberg \color{black} (and Bayesian persuasion) equilibria in our audit game. 

\vspace{-3pt}
\subsubsection{Equilibrium Existence} \label{subsubsec:eq-existence}
We begin by presenting the main result of this section, which establishes the existence of a \color{purple} Stackelberg \color{black} equilibrium.
\vspace{-3pt}
\begin{theorem} [\color{purple}Stackelberg \color{black} Equilibrium Existence] \label{thm:signaling-game-main}
    For any audit cost $\c \!> \!0$ and fine $\k \!\geq \! \c$, there exists a feasible linear program $\mathcal{L}_{c,k}$ (defined in Equations~\eqref{eq:obj-lp-bp}-\eqref{eq:no-audit-con-bp}) whose solutions correspond to \color{purple} Stackelberg \color{black} equilibria of the audit game. Further, the \color{purple} Stackelberg \color{black} equilibria given by the solutions of $\mathcal{L}_{c,k}$ have the maximum excess payments among all \color{purple} Stackelberg \color{black} equilibria and the optimal value of $\mathcal{L}_{c,k}$ implies a tight upper bound on the excess payments in any \color{purple} Stackelberg \color{black} equilibrium. 
\end{theorem}
\vspace{-3pt}
The proof of Theorem~\ref{thm:signaling-game-main} has two main ingredients. First, we establish Lemma \ref{lem:interim_to_exante} which states that the set of Bayesian persuasion equilibria is a subset of \color{purple} Stackelberg \color{black} equilibria, and that, among all \color{purple} Stackelberg \color{black} equilibria, Bayesian persuasion equilibria are the least favorable for the administrator in that they have the highest excess payments (i.e., misreporting fraud). Second, we show that the linear program $\mathcal{L}_{c,k}$ precisely characterizes all Bayesian persuasion equilibria. 
\vspace{-3pt}
\begin{theorem} [Bayesian Persuasion Equilibrium Characterization] \label{thm:main}
    For any audit cost $\c>0$ and fine $\k\geq\c$, there exists a linear program $\mathcal{L}_{c,k}$ (defined in Equations~\eqref{eq:obj-lp-bp}-\eqref{eq:no-audit-con-bp}) that is always feasible and whose solutions characterize all Bayesian persuasion equilibria of the audit game. In particular, the administrator's strategy in all Bayesian persuasion equilibria is to not audit, i.e., $(\ppi^*, \sigmaa^*)$ is a Bayesian persuasion equilibrium if and only if $\sigmaa^*(\s) = 0$ for all $\s \in \S$, and $\ppi^*$ is a solution to $\mathcal{L}_{c,k}$. 
\end{theorem} 
\vspace{-3pt}
\begin{lemma}\label{lem:interim_to_exante}
    Every $(\ppi^*, \sigmaa^*)$ that is a Bayesian persuasion equilibrium is also a \color{purple} Stackelberg \color{black} equilibrium. Further, the \color{purple} Stackelberg \color{black} equilibria with the largest expected excess payments are precisely the Bayesian persuasion equilibria. 
\end{lemma}
\vspace{-3pt}
\color{green}
Lemma~\ref{lem:interim_to_exante} clarifies the role of the Bayesian persuasion formulation in our analysis. The behavioral model of interest remains the \color{purple} Stackelberg \color{black} game, where users know their realized types. The Bayesian persuasion formulation is used as a technical device: it provides an LP characterization that allows us to compute the \color{purple} Stackelberg \color{black} equilibria with the largest excess payments, i.e., the worst-case equilibria from the administrator’s perspective, and thereby yields bounds that apply to all \color{purple} Stackelberg \color{black} equilibria.
\color{black}

Both Theorem~\ref{thm:main} and Lemma~\ref{lem:interim_to_exante} are key to showing Theorem~\ref{thm:signaling-game-main}, which we prove in Appendices~\ref{pf:thm:main},~\ref{pf:lem:interim_to_exante}, and~\ref{pf:thm:signaling-game-main}, respectively. The rest of this section presents bounds on the misreporting fraud and excess payments at any Bayesian persuasion equilibrium (and, by Lemma~\ref{lem:interim_to_exante}, at any \color{purple} Stackelberg \color{black} equilibrium), describes the ideas to establish Theorem~\ref{thm:main}, and discusses its implications.


\subsubsection{Misreporting Probability and Excess Payment Bounds} \label{subsubsec:misreporting-prob-excess-payment-bd}
\vspace{-4pt}
A goal of the audit mechanism is to deter fraud, i.e., keep excess payments low. The following corollaries give bounds on excess payments and misreporting probabilities in both \color{purple} Stackelberg \color{black} and Bayesian persuasion equilibria. 
\begin{corollary}[Misreporting Probability Bound]\label{cor:misrepotingProbBound}
Let $(\ppi^*, \sigmaa^*)$ be a \color{purple} Stackelberg \color{black} equilibrium for the audit game. For any \color{purple} report \color{black} $\s$ and type $\m$ where $\s \neq \m$, if $f(s) < f(m)$ then $\ppi^*(\s|\m) = 0$. Otherwise, $\ppi^*(\s|\m) \leq \min \{\frac{q_{\s} \c}{q_{\m}(\k-\c+\f(\s)-\f(\m))}, 1 \}$. Because every Bayesian persuasion equilibrium is also a \color{purple} Stackelberg \color{black} equilibrium, these inequalities apply to Bayesian persuasion equilibria as well. 
\end{corollary}
\begin{corollary}[Excess Payments Bound]\label{cor:equilibrium_suboptimality}
    At any Bayesian persuasion equilibrium $(\ppi^*, \sigmaa^*)$, we have $\mathcal{E}(\ppi^*, \sigmaa^*) \leq \frac{c \Delta f_{ \max} }{k + \Delta f_{ \max} } $, where $\Delta f_{ \max} := \max_{s, s'} f(s) - f(s')$. By Lemma~\ref{lem:interim_to_exante}, this bound applies to all \color{purple} Stackelberg \color{black} equilibria as well. 
\end{corollary}

See Appendices \ref{apdx:corPf-misreportingBound} and \ref{apdx:corPf-equilibrium_suboptimality} for proofs of Corollaries \ref{cor:misrepotingProbBound} and \ref{cor:equilibrium_suboptimality} respectively. Both corollaries provide insights into the influence of an administrator’s choice of fine $k$ and audit cost $c$ on users’ misreporting probability and excess payments, which decrease as $c$ reduces or as $k$ increases. Furthermore, note that the bounds on both misreporting probability and excess payments converge to zero as $c \rightarrow 0$ or as $k \rightarrow \infty$. Such behavior makes sense as a low audit cost will make it easier for the administrator to audit, and a high fine $k$ will deter users from misreporting. 
While choosing a very large fine $\k$ can almost completely deter misreporting, in practice, the audit parameters must adhere to legal and political constraints. For instance, \color{green} fines are often constrained by administrative policy or political feasibility, and therefore cannot typically be \color{black} arbitrarily high, e.g., traffic light violations cost about \$100 in California. We note that the bounds in Corollaries~\ref{cor:misrepotingProbBound} and~\ref{cor:equilibrium_suboptimality} accommodate practical constraints on audit parameters by treating the fine and audit cost as exogenous, with the corresponding bounds holding for any permissible values of these inputs.
Finally, Corollaries \ref{cor:misrepotingProbBound} and \ref{cor:equilibrium_suboptimality} suggest that, although our studied objective in Equation \eqref{eq:admin-payoff} is not designed solely to minimize fraud, since it also accounts for audit costs to reflect the operational realities of benefits programs, the administrator nonetheless achieves low misreporting fraud and excess payments, ensuring that benefits are approximately delivered to the appropriate user groups for any allocation function $f$.

\vspace{-4pt}
\subsection{Key Ideas for Establishing Theorem~\ref{thm:main}} \label{subsec:misreportingSingle:key-ideas}
\vspace{-2pt}
Theorem~\ref{thm:main} is based on two ingredients. First, we show that the administrator's best response to a user strategy is a linear threshold rule. The substantive challenge lies in proving the second ingredient that when searching for Bayesian persuasion equilibria it suffices to restrict attention to user strategies for which the administrator's best response is identically zero, which, due to the first ingredient, holds if and only if the user strategy satisfies a collection of linear inequalities. Thus, equilibrium user strategies are solutions to a linear program $\mathcal{L}_{c,k}$ whose objective is $\mathcal{U}^R$ and whose constraints are the linear inequalities ensuring the administrator's best response is identically zero. 
We now present the key ingredients to prove Theorem~\ref{thm:main}.






\emph{Ingredient (i):} To see why the administrator's best response is a linear threshold rule, note that $\mathcal{U}^A$ is affine in $\sigmaa(s)$ for every $s \in \S$, where the coefficient of $\sigmaa(s)$ is the expected utility gain from conducting an audit when observing the \color{purple} report \color{black} $s$. This administrator's expected utility gain is $-c + \sum_{m \neq s} \mathbb{P}(m | s) \bigpar{ \max \bigbra{ f(s) - f(m)}_+ + k}$. 
By re-writing the above expression using Bayes' Rule, the administrator's best response to a user strategy $\ppi$, which we denote $\sigmaa_{\ppi}$, is given by:
{\setlength{\abovedisplayskip}{2pt}
\setlength{\belowdisplayskip}{2pt}
\begin{equation}\label{eqn:p2_best_resp}
\sigmaa_{\ppi}(\s)=
\begin{cases}
0, &
\begin{aligned}[t]
\text{if } &\sum_{\m\neq\s} \ppi(\s\mid\m)\,\q_{\m}\,[\f(\s)-\f(\m)]_+ + k \! \! \sum_{\m\neq\s} \! \ppi(\s \! \mid \! \m)\,\q_{\m} \!
\! \leq \! c \! \sum_{\m} \ppi(\s \! \mid \!\m)\,\q_{\m},
\end{aligned}
\\[6pt]
1, & \text{otherwise.}
\end{cases}
\end{equation}}

\begin{remark}[Tie-breaking and off-path audit decisions]
When the administrator is indifferent between auditing and not
auditing after a report $s$ with $\ppi(s) := \sum_m \q_m \ppi(s|m) >0$, we select the
no-audit action, i.e., $\sigma(s)=0$. When $\ppi(s)=0$, the
administrator's action after $s$ does not affect any player's
payoff under the current strategy profile, and every
$\sigma(s)\in[0,1]$ is a best response. We again select
$\sigma(s)=0$. Thus, throughout the paper, we adopt the
convention that the administrator does not audit whenever
auditing and not auditing yield the same expected payoff,
including after reports that occur with probability zero.
\end{remark}



\emph{Ingredient (ii):} Next, we show that any Bayesian persuasion equilibrium corresponds to a no audit administrator strategy. 
\begin{lemma} [Sub-Optimality of Auditing] \label{lem:signalingeqStrategy}
    Suppose that the best-response $\sigmaa_{\ppi}$ of the administrator to a strategy $\ppi$ of the user is such that there exists a \color{purple} report \color{black} $\Tilde{\s}$ with $\sigmaa_{\ppi}(\Tilde{s}) = 1$. Then, $(\ppi, \sigmaa_{\ppi})$ is not a Bayesian persuasion equilibrium. 
\end{lemma}
For a proof of Lemma~\ref{lem:signalingeqStrategy}, see Appendix~\ref{apdx:pfLemSignalingStrategy}. The proof of this lemma relies on showing that users will never under-report their type (i.e., choose $\s$ with $\f(s) < \f(\m)$). 
\begin{lemma}\label{lem:underreporting_dominated}
    For any user strategy $\ppi$, if there exists $\s,\m$ for which $\f(\s) < \f(\m)$ and $\ppi(\s | \m) > 0$, then there exists another strategy $\widetilde{\ppi}$ with utility larger than $\ppi$. 
\end{lemma}
For a proof of Lemma~\ref{lem:underreporting_dominated}, see Appendix~\ref{apdx:pf:lem:underreporting_dominated}. Given Lemma~\ref{lem:underreporting_dominated}, the key idea in proving Lemma~\ref{lem:signalingeqStrategy} is a procedure that takes any $\ppi$, where there exists $\Tilde{s}$ with $\sigmaa_{\ppi}(\Tilde{s}) \! = \! 1$, and produces another feasible $\Tilde{\ppi}$ with higher utility, which follows by our choice of \color{blue} the administrator's payoff function in Equation~\eqref{eq:admin-payoff}. \color{black} 


\emph{Linear Programming Characterization of Bayesian Persuasion Equilibria:} Leveraging Lemma~\ref{lem:signalingeqStrategy}, user strategies in Bayesian persuasion equilibria are strategies whose utility is maximum among all strategies for which the administrator's best response is identically zero. That is, user strategies at Bayesian persuasion equilibria can be computed by the linear program $\mathcal{L}_{c,k}$
{\setlength{\abovedisplayskip}{2pt}
\setlength{\belowdisplayskip}{2pt}
\begin{align}
    \underset{\ppi \in \Pi}{\text{max}} & \sum_{\m \in \S} \q_{\m} \sum_{\s \in \S} \ppi(\s|\m) \f(\s) \label{eq:obj-lp-bp}  \\ 
    \text{s.t. }
    & \sum_{\m \in \S \backslash \{\s\}}  \ppi(\s|\m) \q_{\m} \bigpar{\k + \bigbra{\f(\s) - \f(\m)}_+}  \leq \c \sum_{\m \in \S} \ppi(\s|\m) \q_{\m}, \forall \s \in \S, \label{eq:no-audit-con-bp}
\end{align}}

\noindent where Objective~\eqref{eq:obj-lp-bp} is the user's utility $\mathcal{U}^R$ when the administrator plays a no-audit strategy, and Constraint~\eqref{eq:no-audit-con-bp} represents the condition in Equation~\eqref{eqn:p2_best_resp} required for the administrator to not audit. 

\vspace{-3pt}
\begin{observation}
\label{obs:lp-eq-result}
    An optimal solution $\ppi^*$ to the linear program~\eqref{eq:obj-lp-bp}-\eqref{eq:no-audit-con-bp} maximizes the sender’s expected utility among reporting strategies satisfying the
no-audit constraints. 
\end{observation}
\vspace{-3pt}
Note that jointly Lemmas~\ref{lem:signalingeqStrategy},~\ref{lem:underreporting_dominated} and Observation~\ref{obs:lp-eq-result} establish Theorem~\ref{thm:main}, as we show in Appendix~\ref{pf:thm:main}.

\vspace{-5pt}
\subsection{Discussion} \label{subsec:misreportingSingle:discussion}
\vspace{-2pt}

This section discusses key implications of our equilibrium existence \color{purple} and characterization \color{black} results. First, the absence of audits in Bayesian persuasion equilibria provides a simple relationship between user utility and excess payments that is used to establish upper bounds on excess payments that occur in both \color{purple} Stackelberg \color{black} and Bayesian persuasion equilibria. 
\vspace{-3pt}
\begin{observation}
\label{obs:user_util_excess_payments}
    When $\sigma(s) \! = \! 0$ for all $s$, $\mathcal{U}^R(\ppi, \sigmaa) \! = \! \sum_{m} \! \q_m \! \sum_{s} \! \! \ppi(s | m) f(s)$. Because $\sum_{m} \! \q_m f(m)$ is a constant independent of the user strategy $\ppi$, maximizing $\sum_{m} \! \q_m \sum_{s} \! \ppi(s | m) f(s)$ is equivalent to maximizing $\sum_{m} \! \q_m \! \sum_{s} \! \ppi(s | m) \! \bigpar{f(s) \! - \! f(m)}$, which is precisely the expression for excess payments $\mathcal{E}(\ppi, \sigmaa)$ once we use Lemma~\ref{lem:underreporting_dominated} to rule out any cases where $f(s) \! < \! f(m)$. 
\end{observation}
\vspace{-3pt}
Second, a few remarks about our equilibrium existence \color{purple} and characterization \color{black} results are in order. 
\vspace{-2pt}
\begin{remark} [Interpretation and Robustness of No-Audit Equilibria] \label{rmk:no-audit-eq}
While not conducting audits is the administrator's equilibrium strategy, it does not imply users can arbitrarily misreport, as the administrator only does not audit when users' mixed-strategy satisfies Equation~\eqref{eq:no-audit-con-bp}, which limits the probability of user misreports (see Corollary~\ref{cor:misrepotingProbBound}). In this regard, our no-audit equilibrium aligns with the literature on auditing in public administration contexts, e.g., tax evasion~\citep{strategic-auditing-fraud-1993}, where no audits are conducted in certain parameter regimes and practical settings where the mere presence of an auditing authority, even without conducting audits (or conducting few audits), deters misreporting~\citep{sec-audit}.

\color{green}
While not conducting audits is an equilibrium strategy in the static audit game, \color{purple} auditing may arise with positive probability under other natural equilibrium notions, such as signaling equilibria (see Appendix~\ref{app:classical_signaling_bound}), as well as under different information structures. \color{black} 
In Appendix~\ref{apdx:uniform-random-costs}, we show that audits occur in a stochastic-cost variant of the static game in which the administrator observes the realized audit cost but users know only its distribution. In this setting, when audit costs are uniformly distributed, i.e., $c\sim\mathrm{Unif}[c_1,c_2]$, the sender-preferred Bayesian persuasion benchmark is characterized by a second-order cone program that replaces the linear no-audit constraints (Equation~\ref{eq:no-audit-con-bp}) with convex penalty terms. This SOCP upper bounds the user payoff in the corresponding \color{purple} Stackelberg game \color{black} and coincides with a \color{purple} Stackelberg \color{black} equilibrium under a necessary and sufficient verification condition. Thus, the deterministic-cost static game provides a clean LP-benchmark whose analysis naturally extends and serves as a building block for richer settings in which audits occur; we focus on this benchmark in the next two sections to obtain transparent equilibrium and cost-effectiveness insights. 
Overall, these results highlight how equilibrium notions and information structures influence auditing behavior and point to broader questions about audit transparency, e.g., the extent to which audit probabilities should be publicized.
\color{black}

\end{remark}

\vspace{-4pt}

\begin{remark}[Equilibrium Administrator Utility]\label{rem:UA_interpret}
At any equilibrium with no audits and no underreporting, the administrator's and user's expected utilities satisfy
\[
U^A(\pi,0)
=
-\sum_{m\in\mathcal S}q_m f(m)
-\mathcal E(\pi,0)
=
-U^R(\pi,0).
\]
Thus, for a fixed prior, higher excess payments correspond to lower administrator utility. Moreover, the linear program $\mathcal L_{c,k}$ characterizes the administrator's least-favorable Bayesian persuasion equilibria by maximizing excess payments over reporting strategies satisfying the no-audit constraints.
\end{remark}

\vspace{-2pt}
\begin{remark} [Equilibrium Computation Using Linear Programming]
Our proof of the existence of Bayesian persuasion equilibria is constructive, and, in particular, the linear program~\eqref{eq:obj-lp-bp}-\eqref{eq:no-audit-con-bp} provides a polynomial time equilibrium computation method (in the cardinality of the set $\S$).
We obtain an LP to compute equilibria in our audit game, due to our choice of \color{blue} administrator's payoff function in Equation~\eqref{eq:admin-payoff} \color{black} and the way that the audit fine and costs are incorporated into the user payoffs, wherein being caught misreporting leads to a lower payoff than telling the truth. We emphasize that this result is in contrast to the NP-hardness of computing \color{purple} Stackelberg \color{black} equilibria in general~\citep{korzhyk2010complexity}, which typically involves solving a mixed-integer and non-convex optimization problem. 
Theorem~\ref{thm:main} also contributes to the literature on developing LP approaches in Bayesian persuasion contexts~\citep{optimal-disclosure-lp}; however, in contrast to~\cite{optimal-disclosure-lp}, where equilibria are guaranteed to exist, we also consider a budget-constrained setting where equilibria may not exist. 
\end{remark}

\vspace{-9pt}
\section{\color{blue} Equilibrium Analysis: \color{black} Budget Constrained Setting} \label{sec:multiRecipientMain}
\vspace{-2pt}

\color{blue}
This section extends the equilibrium analysis in the previous section to \color{black} the setting when the administrator is budget-constrained, i.e., it does not have enough budget to audit all $n$ users ($\B < \n \c$), in line with real-world applications where budget constraints may arise for several reasons, e.g., labor limitations or caps on spending to run a benefits program. A budget constraint \color{green} influences the credibility of an audit mechanism and introduces a coupling of the administrator's decisions across users, which \color{black} requires a more nuanced equilibrium analysis. Moreover, since the administrator's decisions are coupled across users with budget constraints, we can no longer model this setting as $n$ parallel single-user games. Thus, we introduce a superscript or a subscript $i$ for each user $i \in [n]$ in the notation for the signal $\s$, type $\m$, and strategies $\ppi, \sigmaa$.

To this end, we first show that while \color{purple} Stackelberg \color{black} equilibria may, in general, not exist with a budget constraint, equilibria exist if there is enough budget to audit even one user (Section~\ref{sec:non-existence}). Further, we extend the notion of \color{purple} Stackelberg \color{black} equilibria under user collusion and show that equilibria exist if the budget scales with the largest coalition size (Section~\ref{sec:coalitions}). 


\paragraph{Administrator's audit policy and budget.}
The administrator observes the vector of reports
$\mathbf{s}=(s_1,\ldots,s_n)\in\mathcal{S}^n$.
We restrict attention to audit policies under which the
marginal probability of auditing recipient $i$ depends only
on that user's own report. Thus, the administrator's
policy is a collection of functions
$\sigma_i:\mathcal{S}\to[0,1]$, where
$\sigma_i(s_i)$ is the probability that recipient $i$ is
audited when reporting $s_i$.

The administrator is subject to the following budget
constraint:
\[
c\sum_{i=1}^n \sigma_i(s_i)\le B
\qquad
\text{for every }\mathbf{s}\in\mathcal{S}^n.
\]
Equivalently, $c\sum_{i=1}^n\max_{s\in\mathcal{S}}
\sigma_i(s)\le B$. This constraint bounds the
administrator's conditional expected expenditure for
every possible report vector.

\vspace{-4pt}
\subsection{Equilibrium Existence (and Non-Existence)} \label{sec:non-existence}
\vspace{-2pt}
While \color{purple} Stackelberg \color{black} equilibria always exist in the absence of a budget constraint (Theorem~\ref{thm:signaling-game-main}), we show that \color{purple} Stackelberg \color{black} equilibria may not exist if the administrator's budget is too small, where its budget is insufficient to audit even one user.

\vspace{-4pt}
\begin{proposition} [Non-Existence of \color{purple} Stackelberg \color{black} Equilibria] \label{prop:counter-eg-nonExistence}
    Suppose $|\S| \! = \! 2$ with $\S \! = \! \{ s_{\min}, s_{\max} \}$ and $\q \! = \! (q_{\min}, q_{\max})$, and $0 \! < \!\B \!< \!\c \frac{\Delta \f_{\max} \left( 1 - \min \{ \frac{q_{\max} \c}{q_{\min} (\k - \c + \Delta \f_{\max})} , 1\} \right)}{k+\Delta \f_{\max}}$. Then, 
    \color{purple} Stackelberg \color{black} equilibria $(\{\ppi_i^*\}_{i = 1}^{\n}, \sigmaa^*)$ may not exist.
\end{proposition}
\vspace{-2pt}

For a proof of Proposition~\ref{prop:counter-eg-nonExistence}, see Appendix~\ref{apdx:pfcounter-eg-nonExistence}. The example in Proposition~\ref{prop:counter-eg-nonExistence} is the smallest instance where equilibria do not exist, and we refer to Appendix~\ref{apdx:existence-regardless-of-budget} for a proof of equilibrium existence in the setting with one user when $|\S| = 2$ irrespective of the budget. 

Despite Proposition~\ref{prop:counter-eg-nonExistence}, we show that if the administrator has enough budget to audit at least one user, equilibria exist. We first establish the general existence result for any $|\S|$, following which we present a stronger result when $|\S| \! = \! 2$. 
\vspace{-2pt}
\begin{theorem} [Existence of Equilibria] \label{thm:eqEx-suff-budget}
    Suppose $\B \geq \c \max_{\s_k, \s_o \in \S} \{ \frac{[\f(\s_k) - \f(\s_o)]_+}{\k + [\f(\s_k) - \f(\s_o)]_+} \} = c \frac{\Delta \f_{\max}}{\k+\Delta \f_{\max}}$. Then, a \color{purple} Stackelberg \color{black} equilibrium $(\{\ppi_i^*\}_{i = 1}^{\n}, \sigmaa^*)$ exists and can be computed via the linear program $\mathcal{L}_{c,k}$ in Equations~\eqref{eq:obj-lp-bp}-\eqref{eq:no-audit-con-bp} for each user $i \in [n]$. 
\end{theorem}
\vspace{-3pt}
For a proof of Theorem~\ref{thm:eqEx-suff-budget}, see Appendix~\ref{apdx:pfEx-suff-budget}. 
In addition to equilibrium existence, the corresponding equilibrium properties reduce to those in the setting without budget constraints and that equilibria can be computed using the LP approach in Section~\ref{sec:misreportingSingle}. Consequently, given our upper bounds on the probability of user misreports and excess user payments in Section~\ref{sec:misreportingSingle}, Theorem~\ref{thm:eqEx-suff-budget} implies that our audit mechanism can effectively mitigate misreporting fraud at a low cost.

While \color{purple} Stackelberg \color{black} equilibria may, in general, not exist if $\B < \c \max_{\s_k, \s_l \in \S} \{ \frac{\f(\s_k) - \f(\s_l)}{\k + \f(\s_k) - \f(\s_l)} \}$, we show that when $|\S| = 2$ and users have the same prior $\q$, \color{purple} Stackelberg \color{black} equilibria exist under a weaker condition on the budget. 
\vspace{-3pt}
\begin{proposition} [Equilibrium Existence: $|\S| = 2$] \label{prop:strongerBudgetCondition}
    Suppose $|\S| = 2$ and all users have the same prior $\q$. Then, a \color{purple} Stackelberg \color{black} equilibrium $(\{\ppi_i^*\}_{i = 1}^{\n}, \sigmaa^*)$ exists if $\B \geq \c \frac{\Delta \f_{\max} \big( 1 - \min \{ \frac{q_{\max} \c}{q_{\min} (\k - \c + \Delta \f_{\max})} , 1\} \big)}{k+\Delta \f_{\max}}$. 
\end{proposition}
\vspace{-3pt}
For a proof of Proposition~\ref{prop:strongerBudgetCondition}, see Appendix~\ref{apdx:pfStrongerBudgetCondition}. 
We reiterate that Proposition~\ref{prop:strongerBudgetCondition} is stronger than Theorem~\ref{thm:eqEx-suff-budget} due to its lower budget requirement. Further, the budget threshold in Proposition~\ref{prop:strongerBudgetCondition} is met when the inequality in Proposition~\ref{prop:counter-eg-nonExistence} is met with equality, highlighting the tightness of this budget requirement for equilibria to exist when $|\S| = 2$ and users have the same prior $\q$. 
\vspace{-3pt}
\subsection{Extension to User Coalitions} \label{sec:coalitions}
\vspace{-2pt}
We now allow for the possibility of user collusion, as may happen in real-world social networks, and show that equilibria exist if the administrator's budget scales with the size of the largest coalition. 

We allow coalition members to coordinate the choice of their reporting strategies, but restrict their deviations to independent randomized reporting decisions. Specifically, a coalition $\mathcal C$ may jointly choose a collection of individual reporting strategies $\{\widehat{\pi}_i(\cdot\mid m_i)\}_{i\in\mathcal C}$, which may differ across members. However, each member's reporting strategy must depend only on that member's own type, and the randomizations used to generate different members' reports must be independent. Thus, conditional on any vector of coalition members' types $\mathbf m_{\mathcal C}$, the joint distribution of their reports must satisfy $\widehat{\Pi}_{\mathcal C}(\mathbf s_{\mathcal C}\mid\mathbf m_{\mathcal C})=\prod_{i\in\mathcal C}\widehat{\pi}_i(s_i\mid m_i)$. In particular, coalition members cannot condition their reports on other members' realized types or reports, or use a common randomization device to correlate their reporting decisions. When a coalition deviates, users outside the coalition retain their prescribed strategies, and the administrator chooses a best-response audit policy subject to the policy restrictions and budget constraint specified in this section. We say that a coalition deviation is profitable if every coalition member's expected utility weakly increases and at least one member's expected utility strictly increases.

To elucidate our result, we first extend the \color{purple} Stackelberg \color{black} equilibrium notion in Definition~\ref{def:signalingGameEqOneRecipient} to the setting with user coalitions of size at most $l$. We define $(\{\ppi_i^*\}_{i = 1}^{\n}, \sigmaa^*)$ to be a \color{purple} Stackelberg \color{black} equilibrium under $\l$-coalitions if there is no coalition of size at most $\l$, such that users in the coalition can deviate in a way where all users weakly improve their payoffs, with at least one user strictly increasing its payoff. We show that \color{purple} Stackelberg \color{black} equilibria under $\l$-coalitions exist if the administrator's budget scales with the largest coalition size.
\vspace{-3pt}
\begin{theorem} [Equilibrium Existence with User Coalitions] \label{thm:coalitions-ExEq}
    Suppose $|S|=2$, the largest user coalition size is $l$, all users have the same prior $q$, and $\B \geq l \c \frac{\Delta \f_{\max} \big( 1 - \min \{ \frac{q_{\max} \c}{q_{\min} (\k - \c + \Delta \f_{\max})} , 1\} \big)}{k+\Delta \f_{\max}}$. Then, a Stackelberg equilibrium under $l$-coalitions exists and can be computed via the linear program $L_{c,k}$ in Equations~\eqref{eq:obj-lp-bp}-\eqref{eq:no-audit-con-bp} for all users $i\in[n]$.
\end{theorem}
\vspace{-3pt}
For a proof, see Appendix~\ref{apdx:pfCoalitions-ExEq}. While we focus on the two-type setting for analytical simplicity, with Theorem~\ref{thm:coalitions-ExEq} paralleling Proposition~\ref{prop:strongerBudgetCondition}, an analogous extension of the $|S|>2$ type result in Theorem~\ref{thm:eqEx-suff-budget} to user coalitions may also be developed. Theorem~\ref{thm:coalitions-ExEq} implies that if the administrator has enough budget to audit all users in the largest coalition, equilibria exist and the resulting equilibrium strategies reduce to that in the setting without coalitions and a budget constraint. Since the coalition size $\l \! \ll \! \n$ in practice \citep{dunbar-number}, our mechanism mitigates misreporting fraud at a low cost. 

\vspace{-3pt}
\section{Comparison of Audit Mechanism to Status Quo without Audits} \label{sec:mainCompAudit}
\vspace{-2pt}

The previous two sections analyzed the equilibria of our audit game and developed an LP approach to efficiently computing the resulting equilibrium outcomes. This section builds on our equilibrium characterization results by investigating whether the decrease in misreporting fraud from our audit mechanism outweighs the \color{green} budget required to make the audit threat credible by comparing the total cost of our audit mechanism to the status quo without an auditing authority. This budget captures the enforcement capacity needed to support the audit mechanism, rather than only realized audit expenditures. \color{black} Here, the total cost is composed of: (i) the budget required to run the audit mechanism and (ii) the excess payments made to users relative to when all users are truthful. We show that the total cost of our audit mechanism is no more than that of the setting without an auditing authority (i) for all parameter ranges when $|\S| \! = \! 2$ and (ii) under a mild condition on the fine $\k$ when $|\S| \! > \! 2$. Throughout this section, we evaluate the total cost of the audit mechanism at the no-audit equilibria constructed in the preceding sections, using the audit budgets specified in each result. 



For ease of exposition, we first focus on the single-user setting in Section~\ref{sec:single-recipient-comp} and extend these results (for $|\S|=2$ and $|\S|>2$) to the multiple-user setting in Section~\ref{sec:multi-recipient-comp}. Then, in Section~\ref{sec:experiment-setup}, we present numerical experiments to compare the costs of our audit mechanism to the status quo without an auditing authority based on a case study of a federal transit benefits program.
\vspace{-5pt}
\subsection{Single-User Setting} \label{sec:single-recipient-comp}
\vspace{-2pt}
We first establish conditions under which the total cost of our audit mechanism ($C^{\text{Audit}}$) is lower than that of the status quo without an auditing authority ($C^{\text{No Audit}}$) in the single-user setting. To compare these costs, we consider the cases when $|\S| = 2$ and $|\S| > 2$ separately. 
We first consider the setting when $|\S| = 2$, in which case the total cost of our audit mechanism is at most that of the status quo without an auditing authority for all $q_{min} \in [0, 1]$.
\vspace{-3pt}
\begin{theorem} [Cost Comparison ($|\S| = 2$)] \label{thm:cost-comparison}
    Let $|\S| = 2$, where $\q = (q_{\min}, q_{\max})$,
    the number of users $\n = 1$, and suppose that $B = \frac{\c}{\k+\Delta f_{\max}} (\Delta f_{\max}) \big( 1 - \min \{ \frac{q_{\max} \c}{q_{\min} (\k - \c + \Delta f_{\max})} , 1\} \big)$. 
    Then, for any $q_{\min} \in [0, 1]$, $C^{\text{Audit}} \leq C^{\text{No Audit}}$.
\end{theorem}
\vspace{-3pt}

For a proof of Theorem~\ref{thm:cost-comparison}, see Appendix~\ref{apdx:pfCostCompTwoType}. This result is unconditional as it does not impose any restriction on the problem parameters to establish $C^{\text{Audit}} \! \leq \! C^{\text{No Audit}}$. However, this unconditional guarantee does not directly extend to the regime where $|\S| \! > \! 2$. 
Yet, we show that if the fine $k$ is sufficiently large, then $C^{\text{Audit}} \! \leq \! C^{\text{No Audit}}$.



\begin{theorem}[Cost Comparison ($|\mathcal S|>2$)]
\label{thm:cost-comparison-multitype}
Let $|\mathcal S|>2$, and let $\pi^*$ be an optimal solution
to the linear program $\mathcal L_{c,k}$ defined in
\eqref{eq:obj-lp-bp}--\eqref{eq:no-audit-con-bp}. Define $D:=\max_{s'\in\mathcal S}f(s')
-\mathbb E_{(s,m)\sim\pi^*}[f(s)].$ Choose the audit budget as
\[
B:=
\begin{cases}
\displaystyle\frac{c\Delta f_{\max}}{k+\Delta f_{\max}},
& D>0,\\[6pt]
0,&D=0.
\end{cases}
\]
If $D>0$, suppose additionally that
$
k\geq\Delta f_{\max}
\left(\frac{c}{D}-1\right).
$
Then $C^{\text{Audit}}\leq C^{\text{No Audit}}$.

\end{theorem}

For a proof of Theorem~\ref{thm:cost-comparison-multitype}, see Appendix~\ref{pf:thm:cost-comparison-multitype}. 
A few comments about the bound on the fine $k$ in Theorem~\ref{thm:cost-comparison-multitype} are in order. First, this bound scales linearly with the maximum difference in the credit allocation $\Delta f_{ \max}$ and the audit cost $\c$. Such a result is intuitive as the fine $k$ should be at least $c$ for it to be economically viable for the administrator to audit users, and the fine should scale proportionally to $\Delta f_{ \max}$ to deter the lowest type users from misreporting. Further, this bound is inversely proportional to the term $\max_{s'} f(s') - \mathbb{E}_{\ppi}[f(s)]$. 

Theorems~\ref{thm:cost-comparison} and~\ref{thm:cost-comparison-multitype} imply that the total cost of our audit mechanism is no more than that of the setting without an auditing authority for most parameters, i.e., the decrease in misreporting fraud from our audit mechanism outweighs the spending to run it. 

\vspace{-3pt}
\subsection{Multiple-User Setting} \label{sec:multi-recipient-comp}
\vspace{-2pt}
This section extends the results in the single-user setting to that with multiple users. Here, with $n$ users, we denote the cost without an auditing authority as $C^{\text{No Audit}}_{\n}$ and the total cost of the audit mechanism with a maximum coalition of size $\l$ as $C^{\text{Audit}}_{(\n, \l)}$. We now extend Theorem~\ref{thm:cost-comparison} (when $|\S| \! = \! 2$) to the multiple-user setting by showing that our audit mechanism achieves a lower total cost than that without an auditing authority for all $q_{\min}$. A generalization of Theorem~\ref{thm:cost-comparison-multitype} also holds when $|\S|>2$ (see Appendix~\ref{apdx:generalizationThm6}). 

\vspace{-3pt}
\begin{corollary} [Cost Comparison Multiple Users] \label{cor:cost-comparison-multiRecipients}
    Let $B = \l \c \frac{(\Delta f_{\max}) \big( 1 - \min \{ \frac{q_{\max} \c}{q_{\min} (\k - \c + \Delta f_{\max})} , 1\} \big)}{\k+\Delta f_{\max}}$, $|\S| \! = \!2$, and the largest coalition size be $\l$. Further, suppose all users have the same prior $\q$. Then, for any $q_{\min} \in [0, 1]$, $C^{\text{Audit}}_{(\n, \l)} \leq C^{\text{No Audit}}_{\n}$.
\end{corollary}
\vspace{-3pt}
For a proof of Corollary~\ref{cor:cost-comparison-multiRecipients}, see Appendix~\ref{apdx:pfCorCost-comparison-Multi}. 
Thus, irrespective of the number of users or largest coalition size, Corollary~\ref{cor:cost-comparison-multiRecipients} and its generalization for $|\S| \! > \! 2$ imply that a small audit budget helps alleviate misreporting fraud. 

\vspace{-3pt}
\subsection{Case Study of Misreporting Fraud} \label{sec:experiment-setup}
\vspace{-2pt}
We evaluate our audit mechanism through numerical experiments based on a case study of Washington D.C.'s federal transit benefits program (FTBP) that provides federal employees with transit passes equal to their monthly commuting costs to work. However, many federal employees provide inaccurate information on their transit benefits applications and the ``amount of potentially fraudulent transit benefits claimed during 2006 was at least \$17 million,'' constituting about 25\% of the total claims~\citep{kutz2007federal}. These numbers were based on an extrapolation of the claims of 4000 employees, where, of the \$4 million claimed, approximately \$1 million was identified as potentially fraudulent. In the following, we present our model calibration and results comparing the total cost of our audit mechanism to the misreporting fraud in this benefits program (and to the setting without an auditing authority).

\emph{Model Calibration:} 
\color{blue}
To design a numerical experiment that reflects Washington D.C.'s FTBP, we focus on the $n = 4000$ workers over a one-month horizon, corresponding to $\$83,333$ of fraudulent claims. This GAO-based figure provides a benchmark for contextualizing the magnitude of misreporting fraud and highlighting the potential of our audit mechanism to substantially reduce its associated costs. To highlight core insights and match available data, which lacks distributional information on user types, providing only the maximum credit allotment and an estimate of the minimum, 
we first present the results where users have two types, i.e., $|\S| \! = \! 2$, and let $\f_{min} \! = \! 50$ and $\f_{max} \! = \! 105$ denote the monthly commuting costs for the low and high types, respectively. We note that $\f_{max} = 105$ corresponds to the maximum allowable monthly benefit that a user could claim under Washington D.C.'s FTBP in 2006. Moreover, $\f_{min} = 50$ is an illustrative estimate of low monthly commuting costs. At the end of this section, we present the results for a setting where users have more than two types. 

In our experiments, we let the fine $\k$ of the audit mechanism be $\$300$, which corresponds to the fine for the first-time offense of not declaring items at airport customs in the United States~\citep{airport-customs} and is on the same order of magnitude as the fines corresponding to traffic violations, e.g., speeding tickets or skipping red lights, in the United States. Moreover, we assume the cost $\c$ of conducting an audit is equivalent to the cost of hiring an individual to look through the data records or visit an employee's home to validate their approximate commuting costs for a month. Given that the average hourly wage in Washington D.C. is \$22.84 and that looking through data records or auditing an individual will at most take a few hours, we consider illustrative audit costs $\c \in \{ \$25, \$75, \$125 \}$, corresponding to the approximate average wage for one hour, three hours, and five hours. While we fix the fine $\k$ and vary the audit cost $\c$ in this section, we also vary the fine $\k$, for which we present the results in Appendix~\ref{apdx:varyFinePlots}. Moreover, we depict the variation in the misreporting probability of users for different ranges of the audit fines and costs for the above defined FTBP instance in Appendix~\ref{apdx:misreportingPlots}.
\color{black}

\emph{Results $(|S| = 2)$:} Figure~\ref{fig:GAOFraud} compares the total costs of our audit mechanism to the setting without an auditing authority (termed ``No Audit'') as a function of $q_{\min}$, the low type probability, for audit costs $ c \! \in \! \{ \$25, \! \$75, \! \$125 \}$ in the setting with two user types. Here, we depict the total costs of our audit mechanism when users do not collude (with a coalition size 1) and when users form coalitions of size 150~\citep{dunbar-number}, along with the GAO benchmark of $\$83,333$ of monthly fraud in Washington D.C.'s FTBP (termed ``Fraud FTBP''). The total cost of the setting without an auditing authority analyzed in Sections~\ref{sec:single-recipient-comp} and~\ref{sec:multi-recipient-comp} is likely to differ from the GAO benchmark of $\$83,333$ of fraud in Washington D.C.'s FTBP, as our analysis for the no audit setting assumes all users are rational and thus any low-type user will misreport; however, in practice, not all users may behave strategically or have incentives to misreport their information. 

Our results in Figure~\ref{fig:GAOFraud} validate Theorem~\ref{thm:cost-comparison} and Corollary~\ref{cor:cost-comparison-multiRecipients}, as the total costs of our audit mechanisms (irrespective of coalition size) are no more than that of the no audit setting. 
Next, Figure~\ref{fig:GAOFraud} shows that as the audit cost $\c$ increases, the total cost of our audit mechanisms increases, aligning with the properties of our audit game's equilibria (see Section~\ref{sec:misreportingSingle}). Despite this observation, for each audit cost $\c$ in Figure~\ref{fig:GAOFraud}, the total cost of our audit mechanisms, under both coalition sizes and all values of $q_{\min}$, are lower than the GAO benchmark of $\$83,333$ of fraud. 
\color{green} For sufficiently large audit costs outside the plotted range (e.g., $c>\$125$, which would imply audits take more than five hours, unlikely in practice), the total cost of the audit mechanism can exceed this empirical benchmark for some values of $q_{\min}$. \color{black} Even in such regimes, however, Theorem~\ref{thm:cost-comparison} implies that the audit mechanism has weakly lower total cost than the benchmark without an auditing authority where all users are rational.



Finally, the total cost of our audit mechanism without coalitions is smaller than that with a coalition size of 150, as the administrator's budget scales with the size of the largest coalition (Theorem~\ref{thm:coalitions-ExEq}). However, the gap in the total costs of the two audit mechanisms is small, suggesting that the reduction in excess payments far outweighs the budget required for audits.

\color{blue}
\emph{Results $(|S| > 2)$:} We now extend our numerical results to the setting with three user types, i.e., $|S| = 3$. Specifically, we consider the type set $\{ \text{min}, \text{mid}, \text{max} \}$, where $f_{min} = 50$, $f_{mid} = 75$, and $f_{max} = 105$. In our experiments, we set $q_{mid} = q_{max}$ for clarity of presentation, noting that our core insights remain qualitatively unchanged when $q_{mid} \neq q_{max}$. To assess the robustness of our two-type cost comparison insights in Corollary~\ref{cor:cost-comparison-multiRecipients}, we also consider coalitions of size 150 in our three-type numerical experiments.

Analogous to Figure~\ref{fig:GAOFraud}, Figure~\ref{fig:GAOFraud_threetype} compares the total costs of our audit mechanism to the setting without an auditing authority (termed ``No Audit'') as a function of $q_{\min}$, the low type probability, for audit costs $ c \! \in \! \{ \$25, \! \$75, \! \$125 \}$ for the three-type setting. The results in Figure~\ref{fig:GAOFraud_threetype} largely parallel those from the two-type case. In particular, they empirically validate our theoretical guarantees established in Theorem~\ref{thm:cost-comparison-multitype} and its generalization to the multiple user setting, while exhibiting similar qualitative trends in the total cost as $c$ increases. A key new observation in the three-type setting is that the total cost gap between the audit mechanism and the no-audit baseline is larger than in the two-type setting. This difference arises because, with more than two user types, satisfying the no-audit constraint in Equation~\eqref{eq:no-audit-con-bp} requires that the weighted sum of the misreporting probabilities of both low and medium-type users be bounded. Consequently, the constraint’s burden is shared between low and medium-type users, rather than concentrated solely on the low type, leading to lower overall misreporting fraud under our audit mechanism. Hence, the presence of additional user types further amplifies the benefits of an audit mechanism in effectively mitigating misreporting fraud.



Our results in Figures~\ref{fig:GAOFraud} and~\ref{fig:GAOFraud_threetype} for both the two and multi-type settings also reveal a nonlinear relationship between the total cost of the audit mechanism and the probability of the low type. When the low-type probability is small, it is also difficult for the administrator to distinguish low-type users from higher types when high signals are observed. As a result, some misreporting may go undetected under the audit mechanism, though its overall impact is limited given the low potential for misreporting fraud in this region. Beyond a threshold probability of the low type, however, the audit mechanism becomes more effective at separating low- and high-type users, leading to a substantial reduction in misreporting precisely in the regime where the potential for fraud is greatest. Thus, overall, our results in Figures~\ref{fig:GAOFraud} and~\ref{fig:GAOFraud_threetype} demonstrate the promise of our audit mechanism in alleviating misreporting fraud, even improving upon the total costs of the status quo without an auditing authority by several orders of magnitude for certain parameters, e.g., higher $q_{\min}$.


\color{black}

\vspace{-16pt}

\begin{figure*}[tbh!]
  \centering
  \begin{subfigure}[b]{0.3\columnwidth}
      \include{Figs/NRL_c_25}
  \end{subfigure} \hspace{25pt}
  \begin{subfigure}[b]{0.3\columnwidth}
      \include{Figs/NRL_c_75}
  \end{subfigure} 
  \begin{subfigure}[b]{0.3\columnwidth}
      \include{Figs/NRL_c_125}
  \end{subfigure}
     \vspace{-40pt}
    \caption{{\footnotesize \sf Comparison of the total costs of the no audit mechanism to our audit mechanism with a coalition size $l=1$ and $l=150$ for audit costs $\c = 25$, $c = 75$, and $c = 125$ as $q_{\min}$ varies for a setting with two user types. We also depict the GAO benchmark in Washington D.C.'s FTBP. 
    }} 
    \label{fig:GAOFraud}
\end{figure*}

\vspace{-20pt}

\begin{figure*}[tbh!]
  \centering
  \begin{subfigure}[b]{0.3\columnwidth}
      \include{Figs/three_type_c_25}
  \end{subfigure} \hspace{25pt}
  \begin{subfigure}[b]{0.3\columnwidth}
      \include{Figs/three_type_c_75}
  \end{subfigure} 
  \begin{subfigure}[b]{0.3\columnwidth}
      \include{Figs/three_type_c_125}
  \end{subfigure}
     \vspace{-40pt}
    \caption{{\footnotesize \sf Comparison of the total costs of the no audit mechanism to our audit mechanism with a coalition size $l=1$ and $l=150$ for audit costs $\c = 25$, $c = 75$, and $c = 125$ as $q_{\min}$ varies for a setting with three user types. We also depict the GAO benchmark in Washington D.C.'s FTBP. 
    }} 
    \label{fig:GAOFraud_threetype}
\end{figure*}

\vspace{-25pt}

\vspace{-5pt}
\section{Conclusion and Future Work} \label{sec:conclusion}
\vspace{-2pt}

We developed an audit mechanism to address misreporting fraud in government benefits programs. In our audit game, we designed administrator and user payoffs in line with the operational realities of benefits programs and showed that, under the proposed audit rules, the administrator's least-favorable Stackelberg equilibria can be computed tractably using linear programming. Our analysis also provided prescriptive upper bounds on the expected excess payments made by the administrator and the probability that users misreport their information, which hold under both Stackelberg and classical signaling equilibria. 
Additionally, with an appropriately chosen budget level, the reduction in excess payments in any least-favorable Bayesian Stackelberg equilibrium offsets the audit budget required to sustain the mechanism, yielding total costs no greater than those of a baseline without an auditing authority.
Finally, we presented numerical experiments based on an application case of Washington D.C.'s federal transit benefits program. Our numerical results demonstrated that our audit mechanism achieves several orders of magnitude improvement in total cost compared to the status quo without an auditing authority for some parameter ranges.

There are several future research directions. First, there is scope to generalize our model to a broader range of audit mechanisms, including alternative deterrents to fines, such as the loss of user eligibility. Further, it would be valuable to analyze the influence of other administrator objectives in mitigating misreporting fraud beyond the objective of minimizing the sum of the excess payments and audit costs studied in this work. Finally, while this paper focuses on the static audit game with a known user type distribution, the framework could be extended to online settings in which the administrator learns the type distribution over time, which we have omitted here to preserve the paper's focus.

\color{black}

\section*{Acknowledgments}
This work was supported in part by the Stanford
Interdisciplinary Graduate Fellowship.





\bibliographystyle{apacite} 
\bibliography{main} 






\clearpage

\appendix

\section{Additional Comparison to Prior Work on Audit Games} \label{apdx:comp-audit-estornell}

The auditing theory literature~\citep{henderson-audit} spans a variety of applications from resource allocation~\citep{lundy-taylor-etal-2019} and classification systems~\citep{estornell2021incentivizing} to public administration contexts, such as social welfare programs~\citep{appelgren2019optimal}, tax evasion~\citep{ALLINGHAM1972323}, and environmental compliance~\citep{MACHOSTADLER2006110,YANIV1997435}. While our model shares many similarities with prior works on audit games (e.g., a two-stage structure), our work differs in the solution methodology and the payoff and information structures leading to very different equilibrium behavior, summarized below.


\emph{Payoffs:} In many prior works on auditing (e.g., strategic classification~\citep{estornell2023incentivizing} and tax evasion~\citep{Macho–Stadler-2002,PARAMONOVAKUCHUMOVA2017162}), audits do not cost the principal anything; hence, the entire audit budget is always used. Instead, in our work, the administrator pays for each audit, creating incentives to lower audit rates, even when within budget. This modeling choice is motivated by real-world government benefits programs, where the cost of auditing is often comparable to the value of the benefits being disbursed, which is in contrast to settings (e.g., tax evasion) where potential fraud amounts can far exceed audit costs.

\emph{Information Structure:} Because of this difference in cost structure, the audit cost and budget affect users’ best response strategies. In our equilibria, the degree to which users misreport is such that, from the administrator’s perspective, the value-at-risk is not high enough for auditing to be financially worthwhile. Contrastingly, in prior work (e.g.,~\cite{estornell2023incentivizing}), users do not know (nor need to know) the administrator's budget. Hence, we believe our model better captures strategic user behavior in welfare programs, accounting for the administrator’s practical limitations (e.g., financial sustainability). Indeed, for welfare programs the per-user aid is comparable to the cost to audit a user.

\emph{Other consequences of modeling differences:} Unlike prior work (e.g., \cite{estornell2023incentivizing}) that yields pure-strategy equilibria, our \color{purple} model also admits mixed-strategy equilibria. \color{black} Moreover, when the budget is sufficiently large, our equilibria are budget-independent, unlike in~\cite{estornell2023incentivizing}, though this comparison concerns the setting without user collusion (the setting in~\cite{estornell2023incentivizing}).

\emph{Incorporating Audit Costs in Administrator's Objective:} Unlike the above discussed works that do not incorporate audit costs in the administrator's objective, our work is closely aligned with~\cite{strategic-auditing-fraud-1993}, which incorporates audit costs into the administrator’s objective. Nonetheless, our work differs from~\cite{strategic-auditing-fraud-1993} in several important ways.

The main difference between~\cite{strategic-auditing-fraud-1993} and this work is in the structure of penalties against discovered instances of misreporting, which has fairness implications in the resulting equilibria. In~\cite{strategic-auditing-fraud-1993}, upon an audit discovering a user's dishonesty, that user is punished by an amount proportional to the excess payments they demanded, while in our work, the punishment is a constant value. This distinction leads to several important structural differences in the emerging equilibria, which in particular have implications on the fairness of these systems. The equilibria arising from~\cite{strategic-auditing-fraud-1993} involve the administrator auditing individuals who are the most costly to the system, i.e., those who are in the greatest need of aid. As a result, all of the misreporting fraud is conducted by those with the highest economic standing, i.e., those with the least need for aid, allowing these users to extract more value from the system than what they are entitled to, while users with the greatest need for aid are unable to get ahead in the same way. Such equilibria \textit{exacerbate} wealth inequality, which is antithetical to the goal of welfare programs, and to be fair,~\cite{strategic-auditing-fraud-1993} studies interactions between profit maximizing agents in a business setting, not welfare programs. On the other hand, the fixed-penalty structure in our work treats misreporting uniformly across user types, rather than imposing penalties proportional to the amount of excess aid requested. This avoids a potential source of unequal treatment across beneficiary groups, a consideration that may be relevant in welfare program settings. 
This is one reason why we believe that the cost structure in this work is more appropriate for welfare program settings. Finally, compared to~\cite{strategic-auditing-fraud-1993}, we also extend our model to settings with user collusion and administrator budget constraints. 

\section{Proofs}

\subsection{Proof of Theorem~\ref{thm:main}}\label{pf:thm:main}

By Definition~\ref{def:sender-preferred-subgame-perfect-equilibria}, a strategy profile $\bigpar{\ppi^*, \sigmaa^*}$ is a Bayesian persuasion equilibrium for the audit game if and only if
{\setlength{\abovedisplayskip}{2pt}
\setlength{\belowdisplayskip}{2pt}
\begin{align*}
    \sigmaa^* \in \underset{\sigmaa: \sigmaa(s) \in [0, 1], \forall s \in \S}{\text{argmax }} \mathcal{U}^A(\ppi^*, \sigmaa), 
    \quad \text{and} \quad \ppi^* \in &\underset{\ppi \in \Pi}{\text{ argmax }} \mathcal{U}^R(\ppi, \sigmaa) \\
    &\quad \text{ s.t. } \sigmaa \in \underset{\sigmaa_0: \sigmaa_0(s) \in [0, 1], \forall s \in \S}{\text{argmax }} \mathcal{U}^A(\ppi, \sigmaa_0).
\end{align*}}
Lemma~\ref{lem:signalingeqStrategy} shows, however, that any $\ppi$ for which $\sigmaa^* \not\equiv 0$ cannot be optimal. Therefore, the addition of constraints \eqref{eq:no-audit-con-bp} for every $s \in \S$ does not change the solution set or objective value of the optimization problem. Hence the Bayesian persuasion equilibria conditions are equivalent to
{\setlength{\abovedisplayskip}{2pt}
\setlength{\belowdisplayskip}{2pt}
\begin{align*}
    \ppi^* \in &\underset{\ppi \in \Pi}{\text{ argmax }} \mathcal{U}^R(\ppi, 0) \\
    \text{s.t. } & \mathbf{0} \in \underset{\sigma_0}{\text{argmax }} \mathcal{U}^A(\ppi, \sigmaa_0).
\end{align*}}
Substituting the definition for $\mathcal{U}^R$ and the best response rule of the administrator defined in $\eqref{eqn:p2_best_resp}$, the Bayesian persuasion equilibria conditions are equivalent to:
{\setlength{\abovedisplayskip}{2pt}
\setlength{\belowdisplayskip}{2pt}
\begin{align*}
    \ppi^* \in \underset{\ppi \in \Pi }{\text{argmax}} & \sum_{\m \in \S} \q_{\m} \sum_{\s \in \S} \ppi(\s|\m) \f(\s)   \\ 
    \text{subject to }
    & \sum_{\m \in \S \backslash \{\s\}}  \ppi(\s|\m) \q_{\m} \bigpar{\k + \bigbra{\f(\s) - \f(\m)}_+} \leq \c \sum_{\m \in \S} \ppi(\s|\m) \q_{\m}, \forall \s \in \S.
\end{align*}}
and $\sigmaa^* \equiv 0$. Thus the Bayesian persuasion equilibrium user strategies can be computed from the linear program $\mathcal{L}_{c,k}$ defined in~\eqref{eq:obj-lp-bp}-\eqref{eq:no-audit-con-bp}, establishing the desired result.  

\subsection{Proof of Lemma~\ref{lem:signalingeqStrategy}} \label{apdx:pfLemSignalingStrategy}

Consider $\ppi := \{ \ppi(\cdot| \m) \}_{\m \in \S}$ for which there exists $\s'$ so that $\sigmaa^*(\s') = 1$. Throughout this proof, we will use $\ppi(\s,\m)$ to denote the probability that the true type is $m$ and the reported type is $s$, as shorthand for $\ppi(\s|\m) \q_{\m}$. Additionally, due to Lemma~\ref{lem:underreporting_dominated}, we can assume that $\pi(\s,\m) > 0$ only if $\f(\s) \geq \f(\m)$. The proof has 3 main steps:

\begin{enumerate}
    \item First, we show how to construct a new action $\Tilde{\ppi}$ from $\ppi$ (which we will show has a better expected utility for the user than $\ppi$ does).
    \item Next, we show that the probability that the administrator audits when responding to $\ppi$ is always at least as large as the probability of the administrator auditing when responding to $\Tilde{\ppi}$.
    \item Finally, we use the first two ingredients to show that $\Tilde{\ppi}$ has a higher expected utility for the user than $\ppi$. 
\end{enumerate}

\subsubsection{Step 1}

To construct another policy with better utility than $\ppi$, first define the sets:
{\setlength{\abovedisplayskip}{2pt}
\setlength{\belowdisplayskip}{2pt}
\begin{align*}
    E_1 = \{(\s,\m) : \s = \s', \m \neq \s' \}, \quad
    E_2 = \{(\s,\m) : \s \neq \s', \m = \s\}, \quad
    E_3 = (\S \times \S) \setminus (E_1 \cup E_2). 
\end{align*}}
We define $\Tilde{\ppi}$ from $\ppi$ by moving the probability mass from $E_1$ to $E_2$. Namely, 
{\setlength{\abovedisplayskip}{2pt}
\setlength{\belowdisplayskip}{2pt}
\begin{align*}
    \Tilde{\ppi}(s,m) := \left\{ 
        \begin{tabular}{cc}
            $0$ & if $(\s,\m) \in E_1$, \\
            $\ppi(\s,\m) + \ppi(\s',\m)$ & if $(\s,\m) \in E_2$, \\
            $\ppi(\s,\m)$ & if $(\s, \m) \in E_3$.
        \end{tabular}
    \right.
\end{align*}}
By construction, $\Tilde{\ppi}$ is a valid strategy (i.e., the marginal distribution of $\m$ is equal to $\q$) since for any $\m \in \S$, there are two cases to consider. First, if $m \neq s'$, we have
{\setlength{\abovedisplayskip}{2pt}
\setlength{\belowdisplayskip}{2pt}
\begin{align*}
    \sum_{\s \in \S} \widetilde{\ppi}(\s,\m) &= \widetilde{\ppi}(\m,\m) + \widetilde{\ppi}(\s',\m) + \sum_{\s \not\in \{\m,\s'\} } \widetilde{\ppi}(\s,\m) = \left( \ppi(\m,\m) + \ppi(\s',\m) \right) + 0 + \sum_{\s \not\in \{\m,\s'\} } \ppi(\s,\m) \\
    &= \sum_{\s \in \S } \ppi(\s,\m) = \q_m. 
\end{align*}}
Second, if $m = s'$, then $\{(s,m)\}_{s \in S}$ all belong to $E_3$ on which $\ppi$ and $\Tilde{\ppi}$ agree. Hence $\sum_{\s \in \S} \widetilde{\ppi}(\s,\m) = \sum_{\s \in \S} \ppi(\s,\m) = \q_m$. 

\subsubsection{Step 2}

We use $\Tilde{\sigmaa}^*$ to denote the administrator's best response to $\Tilde{\ppi}$ as defined in \eqref{eqn:p2_best_resp}. There are two cases to consider: $\s = \s'$ and $\s \neq \s'$.

\noindent \textit{Case 1:} When $\s = \s'$, we note that by construction $\sum_{\m \neq \s'} \Tilde{\ppi}(\s', m) = 0$ and hence $\Tilde{\sigmaa}^*(s') = 0 \leq \sigmaa^*(s')$.

\noindent \textit{Case 2:} When $\s \neq \s'$, we have 
\begin{itemize}
    \item $\ppi(\s, \m) = \Tilde{\ppi}(\s,\m)$ whenever $\s \neq \m$ (since in this case $(\s,\m) \in E_3$).
    \item $\ppi(\s, \m) \leq \Tilde{\ppi}(\s,\m)$ whenever $\s = \m$ (since in this case $(\s,\m) \in E_2$). 
\end{itemize}
Thus in case 2 we have 
{\setlength{\abovedisplayskip}{2pt}
\setlength{\belowdisplayskip}{2pt}
\begin{align*}
    k \sum_{\m \neq \s} \ppi(\s,\m) &= k \sum_{\m \neq \s} \Tilde{\ppi}(\s,\m) \quad \text{and} \\
    \sum_{\m \in \S} \ppi(\s,\m) (c + \f(\m) - \f(\s)) &\leq \sum_{\m \in \S} \Tilde{\ppi}(\s,\m) (c + \f(\m) - \f(\s)).
\end{align*}}
Thus according to \eqref{eqn:p2_best_resp}, $\Tilde{\sigma}^*(s) = 1$ only if $\sigma^*(\s) = 1$, and hence $\Tilde{\sigma}^*(\s) \leq \sigma^*(\s)$. To conclude, we have shown that $\Tilde{\sigma}^*(\s) \leq \sigma^*(\s)$ for every $\s \in \S$.

\subsubsection{Step 3}

Recall that $U^i(\sigmaa^*(\s), \m)$ represents the utility obtained by user $i$ when its true type is $\m$ and it sends \color{purple} report \color{black} $\s$. User $i$'s expected utility when playing according to $\Tilde{\ppi}$ is thus:
{\setlength{\abovedisplayskip}{2pt}
\setlength{\belowdisplayskip}{2pt}
\begin{align*}
    \mathbb{E}_{(\s,\m) \sim \Tilde{\ppi}} \left[ U^i(\Tilde{\sigmaa}^*(\s), \m) \right] &= \sum_{(\s,\m) \in \S \times \S} \Tilde{\ppi}(\s,\m) \left( \f(\s) + \Tilde{\sigmaa}^*(s) \mathbb{I}_{[\s \neq \m]} (\f(\m) - \f(\s) - k) \right) \\
    &\overset{(a)}{=} \underbrace{\sum_{(\s,\m) \in E_2} \Tilde{\ppi}(\s,\m) \left( \f(\s) + \Tilde{\sigmaa}^*(s) \mathbb{I}_{[\s \neq \m]} (\f(\m) - \f(\s) - k) \right)}_{\text{Term 1}} \\
    &+ \underbrace{\sum_{(\s,\m) \in E_3} \Tilde{\ppi}(\s,\m) \left( \f(\s) + \Tilde{\sigmaa}^*(s) \mathbb{I}_{[\s \neq \m]} (\f(\m) - \f(\s) - k) \right)}_{\text{Term 2}} 
\end{align*}}
Equality $(a)$ is because $\S \times \S = E_1 \cup E_2 \cup E_3$, and $\Tilde{\ppi}(\s,\m) = 0$ for every $(\s,\m) \in E_1$. We can lower bound Term 1 as follows:
{\setlength{\abovedisplayskip}{2pt}
\setlength{\belowdisplayskip}{2pt}
\begin{align*}
    \text{Term 1} &= \sum_{(\s,\m) \in E_2} \Tilde{\ppi}(\s,\m) \left( \f(\s) + \Tilde{\sigmaa}^*(s) \mathbb{I}_{[\s \neq \m]} (\f(\m) - \f(\s) - k) \right)  \overset{(b)}{=} \sum_{\m \neq \s' } \Tilde{\ppi}(\m,\m) \f(\m) \\
    &\overset{(c)}{=} \sum_{\m \neq \s' } \left( \ppi(\m,\m) + \ppi(\s',\m) \right) \f(\m) = \sum_{\m \neq \s'} \ppi(\s',\m) \f(\m) + \sum_{\m \neq \s'} \ppi(\m,\m) \f(\m) \\
    &\overset{(d)}{=} \sum_{\m \neq \s'} \ppi(\s',\m) \f(\m) + \sum_{(\s,\m) \in E_2} \ppi(\m,\m) \left( \f(\m) + \sigmaa^*(\s) \mathbb{I}_{[\s \neq \m]} (\f(\m) - \f(\s) - k) \right) \\
    &\overset{\hypertarget{eqn:truth_noninferiority}{(e)}}{\geq} \sum_{\m \neq \s'} \ppi(\s',\m) \left( \f(\s') + \sigmaa^*(\s')(\f(\m) - \f(\s') - k) \right) \\
    & + \sum_{(\s,\m) \in E_2} \ppi(\m,\m) \left( \f(\m) + \sigmaa^*(\s) \mathbb{I}_{[\s \neq \m]} (\f(\m) - \f(\s) - k) \right) \\
    &= \sum_{(\s,\m) \in E_1 \cup E_2} \ppi(\s,\m) \left( \f(\s) + \sigmaa^*(\s) \mathbb{I}_{[\s \neq \m]} (\f(\m) - \f(\s) - k) \right).
\end{align*}}
Equalities $(b), (d)$ are because by construction $\s = \m$ for every $(\s,\m) \in E_2$, and as a consequence, $\mathbb{I}_{[\s \neq \m]} = 0$. Equality $(c)$ is by the definition of $\Tilde{\ppi}$. Inequality $(e)$ holds whenever $\sigmaa^*(\s') \geq \frac{\f(\s') - \f(\m)}{k + \f(\s') - \f(\m)}$, which is true since we know $\sigmaa^*(\s') = 1$. 

Next, we can lower bound Term 2 as follows:
{\setlength{\abovedisplayskip}{2pt}
\setlength{\belowdisplayskip}{2pt}
\begin{align*}
    \text{Term 2} &= \sum_{(\s,\m) \in E_3} \Tilde{\ppi}(\s,\m) \left( \f(\s) + \Tilde{\sigmaa}^*(s) \mathbb{I}_{[\s \neq \m]} (\f(\m) - \f(\s) - k) \right) \\
    &\overset{(f)}{=} \sum_{(\s,\m) \in E_3} \ppi(\s,\m) \left( \f(\s) + \Tilde{\sigmaa}^*(s) \mathbb{I}_{[\s \neq \m]} (\f(\m) - \f(\s) - k) \right) \\
    &\overset{(g)}{\geq} \sum_{(\s,\m) \in E_3} \ppi(\s,\m) \left( \f(\s) + \sigmaa^*(s) \mathbb{I}_{[\s \neq \m]} (\f(\m) - \f(\s) - k) \right).
\end{align*}}
Equation $(f)$ is because $\ppi$ is equal to $\Tilde{\ppi}$ on $E_3$. Inequality $(g)$ is because $\sigmaa^*(\s) \geq \Tilde{\sigmaa}^*(s)$ for all $s$ (shown in step 2), and $\f(\m) - \f(\s) - k$ is always non-positive (see Lemma~\ref{lem:underreporting_dominated}). 

Putting everything together, 
{\setlength{\abovedisplayskip}{2pt}
\setlength{\belowdisplayskip}{2pt}
\begin{align*}
    \mathbb{E}_{(\s,\m) \sim \Tilde{\ppi}} \left[ U^i(\Tilde{\sigmaa}^*(\s), \m) \right] &= \text{Term 1} + \text{Term 2} \\
    &> \sum_{(\s,\m) \in E_1 \cup E_2 \cup E_3} \ppi(\s,\m) \left( \f(\s) + \sigmaa^*(s) \mathbb{I}_{[\s \neq \m]} (\f(\m) - \f(\s) - k) \right) \\
    &= \mathbb{E}_{(\s,\m) \sim \ppi} \left[ U^i(\sigmaa^*(\s), \m) \right],
\end{align*}}
where the inequality is strict because \(\sigmaa^*(s')=1\) implies that some misreporting mass is assigned to \(s'\), and transferring this mass to truthful reporting strictly improves user utility. This establishes the desired result. 

\subsection{Proof of Lemma~\ref{lem:underreporting_dominated}}\label{apdx:pf:lem:underreporting_dominated}

Consider any $\ppi$ for which there exists $\m', \s' \in \S$ such that $\ppi(\s' | \m') > 0$ and $\f(\s') < \f(\m')$. We will show that there exists another user strategy $\widetilde{\ppi}$ which leads to a higher expected utility for the user. 

Throughout this proof, we will use $\ppi(\s,\m)$ to denote the probability that the true type is $m$ and the reported type is $s$, as shorthand for $\ppi(\s|\m) \q_{\m}$. 

Consider the following policy $\widetilde{\ppi}$:
{\setlength{\abovedisplayskip}{2pt}
\setlength{\belowdisplayskip}{0pt}
\begin{align*}
    \widetilde{\ppi}(s, m) &= \left\{ 
    \begin{tabular}{cc}
        $\ppi(\s,\m)$ & if $(\s,\m) \not\in \{(\s', \m'), (m',m')\}$ \\
        $\ppi(\s',\m') + \ppi(\m', \m')$ & if $\s = \m = \m'$ \\
        $0$ & otherwise. 
    \end{tabular}
    \right.
\end{align*}}
In other words, $\widetilde{\ppi}$ is identical to $\ppi$ except in one situation: any time when $\m = \m'$ and $\ppi$ would \color{purple} report \color{black} $\s = \s'$, $\widetilde{\ppi}$ would instead \color{purple} report \color{black} $\m'$. 

Let $\sigmaa^*$ be the best response to $\ppi$ and let $\widetilde{\sigmaa}^*$ be the best response to $\widetilde{\ppi}$. Next we will show that $\widetilde{\sigmaa^*}(\s) \leq \sigmaa^*(\s)$ for any $\s \in \S$. This will be useful because the user's expected utility is a decreasing function of $\sigmaa$. To this end, recall that the administrator's best response to a user strategy is given by \eqref{eqn:p2_best_resp}, which, by using the law of total probability, we can re-write as
{\setlength{\abovedisplayskip}{2pt}
\setlength{\belowdisplayskip}{2pt}
\begin{align*}
    \sigmaa^*(s) = \left\{ 
    \begin{tabular}{cc}
        $0$ & $\sum_{m \neq s} \frac{\ppi(\s,\m)}{\sum_{\m''} \ppi(\s, \m'')} (\max ( \f(\s) - \f(\m), 0 ) + k) \leq c$ \\
        $1$ & otherwise.  
    \end{tabular}
    \right.
\end{align*}}
There are three cases to consider:
\begin{itemize}
    \item[Case 1:]  $\s \not \in \{s', m'\}$. Since $\widetilde{\ppi}(\s,\m) = \ppi(\s,\m)$ whenever $\m \neq \m'$, we have $\widetilde{\sigmaa}^*(\s) = \sigmaa^*(\s)$ whenever $\s \not \in \{s', m'\}$. 
    \item[Case 2:] $\s = \s'$. By definition of $\widetilde{\ppi}$, $\widetilde{\ppi}(\s',\m') = 0 < \ppi(\s',\m')$ and $\widetilde{\ppi}(\s', \m) = \ppi(\s', \m)$ for any $\m \neq \m'$. Using this information, we see that
    {\setlength{\abovedisplayskip}{2pt}
    \setlength{\belowdisplayskip}{2pt}
    \begin{align*}
        c \sum_{\m \in \S} \widetilde{\ppi}(\s',\m) - &\sum_{\m \neq \s'} \widetilde{\ppi}(\s', \m) (\max(\f(\s') - \f(\m), 0) + k) \\
        &= \left( c \sum_{\m \in \S} \ppi(\s',\m) \right) - \left( \sum_{\m \neq \s'} \ppi(\s', \m) (\max(\f(\s') - \f(\m), 0) + k) \right) \\
        &- c \ppi(\s',\m') + (\max(\f(\s') - \f(\m'), 0) + k)) \ppi(\s',\m') \\
        &\overset{(a)}{\geq} \left( c \sum_{\m \in \S} \ppi(\s',\m) \right) - \left( \sum_{\m \neq \s'} \ppi(\s', \m) (\max(\f(\s') - \f(\m), 0) + k) \right).
    \end{align*}}
    Here $(a)$ is because $\max(\f(\s') - \f(\m'), 0) + k = k \geq c$. Hence $\widetilde{\sigmaa}^*(\s') = 1$ only if $\sigmaa^*(\s') = 1$, i.e., $\widetilde{\sigmaa}^*(\s') \leq \sigmaa^*(\s')$.
    
    \item[Case 3:] $\s = \m'$. By definition of $\widetilde{\ppi}$, $\widetilde{\ppi}(\m',\m') > \ppi(\m',\m')$ and $\widetilde{\ppi}(\m', \m) = \ppi(\m', \m)$ for any $\m \neq \m'$. Therefore $\frac{\widetilde{\ppi}(\m',\m)}{\sum_{\m''} \widetilde{\ppi}(\m', \m'')} \leq \frac{\ppi(\m',\m)}{\sum_{\m''} \ppi(\m', \m'')}$ for any $\m \neq \m'$. We can then conclude that
    {\setlength{\abovedisplayskip}{2pt}
    \setlength{\belowdisplayskip}{2pt}
    \begin{align*}
        \sum_{\m \neq \m'} \frac{\widetilde{\ppi}(\m',\m)}{\sum_{\m''} \widetilde{\ppi}(\m', \m'')} ([\f(\m') - \f(\m)]_+ + k) \leq \sum_{\m \neq \m'} \frac{\ppi(\m',\m)}{\sum_{\m''} \ppi(\m', \m'')} ([\f(\m') - \f(\m)]_+ + k).
    \end{align*}}
    Hence $\widetilde{\sigmaa}^*(\m') = 1$ only if $\sigmaa^*(\m') = 1$. This implies $\widetilde{\sigmaa}^*(\m') \leq \sigmaa^*(\m')$.
\end{itemize}

So we see that $\widetilde{\sigmaa}^*(\s) \leq \sigmaa^*(\s)$ for any $\s \in \S$. Putting everything together, the expected utility for the user under $\ppi$ is:
{\setlength{\abovedisplayskip}{2pt}
\setlength{\belowdisplayskip}{2pt}
\begin{align*}
    \sum_{(\s,\m) \in \S \times \S} \ppi(\s,\m) &\left( \f(\s) + \sigma(\s) \mathbb{I}_{[s \neq m]} (\min(\f(\m) - \f(\s), 0) - k) \right) \\
    &\overset{(b)}{\leq} \sum_{(\s,\m) \in \S \times \S} \ppi(\s,\m) \left( \f(\s) + \widetilde{\sigma}(\s) \mathbb{I}_{[s \neq m]} (\min(\f(\m) - \f(\s), 0) - k) \right) \\
    &= \sum_{(\s,\m) \not\in \{(\s',\m'), (\m',\m')\} } \ppi(\s,\m) \left( \f(\s) + \widetilde{\sigma}(\s) \mathbb{I}_{[s \neq m]} (\min(\f(\m) - \f(\s), 0) - k) \right) \\
    &+ \ppi(\s',\m') \left( \f(\s') + \widetilde{\sigma}(\s) \mathbb{I}_{[s \neq m]} (\min(\f(\m) - \f(\s), 0) - k)\right) + \ppi(\m',\m') \f(\m') \\
    &\overset{(c)}{\leq} \sum_{(\s,\m) \not\in \{(\s',\m'), (\m',\m')\} } \ppi(\s,\m) \left( \f(\s) + \widetilde{\sigma}(\s') \mathbb{I}_{[s' \neq m']} (\min(\f(\m') - \f(\s'), 0) - k) \right) \\
    &+ \ppi(\s',\m') \f(\m') + \ppi(\m',\m') \f(\m') \\
    &\overset{(d)}{=} \sum_{(\s,\m) \not\in \{(\s',\m'), (\m',\m')\} } \widetilde{\ppi}(\s,\m) \left( \f(\s) + \widetilde{\sigma}(\s) \mathbb{I}_{[s \neq m]} (\min(\f(\m) - \f(\s), 0) - k) \right) \\
    &+ \widetilde{\ppi}(\m',\m') \f(\m') \\
    &= \sum_{(\s,\m) \not\in \{(\s',\m')\} } \widetilde{\ppi}(\s,\m) \left( \f(\s) + \widetilde{\sigma}(\s) \mathbb{I}_{[s \neq m]} (\min(\f(\m) - \f(\s), 0) - k) \right) \\
    &\overset{(e)}{=} \sum_{(\s,\m) \in \S \times \S } \widetilde{\ppi}(\s,\m) \left( \f(\s) + \widetilde{\sigma}(\s) \mathbb{I}_{[s \neq m]} (\min(\f(\m) - \f(\s), 0) - k) \right).
\end{align*}}
Since the first term in this chain of inequalities is the user's expected utility under $\ppi$, and the last term in this chain of inequalities is the user's expected utility under $\widetilde{\ppi}$, we have shown that $\widetilde{\ppi}$ gives more expected utility to the user than $\ppi$, establishing the desired result.

Inequality $(b)$ is because $\widetilde{\sigmaa}^*(\s) \leq \sigmaa^*(\s)$ for any $\s \in \S$, and because the expected user utility is a decreasing function of $\sigmaa^*$. Inequality $(c)$ is because $\f(\s') + \widetilde{\sigma}(\s') \mathbb{I}_{[s' \neq m']} (\min(\f(\m') - \f(\s'), 0) - k) \leq \f(\s') \leq \f(\m')$. Equality $(d)$ is by definition of $\widetilde{\ppi}$, indeed by construction we have $\widetilde{\ppi}(\m',\m') = \ppi(\s',\m') + \ppi(\m', \m')$. Finally, equality $(e)$ is because $\widetilde{\ppi}(\s',\m') = 0$ by construction. 

\subsection{Proof of Theorem~\ref{thm:signaling-game-main}}\label{pf:thm:signaling-game-main}
Theorem~\ref{thm:main} tells us that the user strategies corresponding to Bayesian persuasion equilibria are precisely those which maximize excess payments subject to satisfying the no-audit condition in \eqref{eqn:p2_best_resp} for every $s \in \S$ (See Observation~\ref{obs:user_util_excess_payments} for the equivalence between user utility and excess payments in any Bayesian persuasion equilibrium). Hence, for any Bayesian persuasion equilibrium $\bigpar{ \ppi^*, \sigmaa_{\ppi^*} }$, we have
{\setlength{\abovedisplayskip}{2pt}
\setlength{\belowdisplayskip}{2pt}
\begin{align*}
    \mathcal{E} \bigpar{\ppi^*, \sigmaa_{\ppi^*}} = \max_{\ppi \in \Pi} \mathcal{E} \bigpar{ \ppi, \sigmaa_{\ppi} } \mathbb{I} \bigbra{ \ppi \text{ satisfies the no-audit condition in \eqref{eqn:p2_best_resp} for all } s \in \S}.
\end{align*}}

From Observation~\ref{obs:sg-no-audit-optimality}, for every \color{purple} Stackelberg \color{black} equilibrium $(\ppi', \sigmaa_{\ppi'})$, $\ppi'$ satisfies the no-audit condition in \eqref{eqn:p2_best_resp} for every $s \in \S$. In other words, $\ppi'$ is feasible for the linear program $\mathcal{L}_{c,k}$ defined in equations~\eqref{eq:obj-lp-bp}-\eqref{eq:no-audit-con-bp}. Hence
{\setlength{\abovedisplayskip}{2pt}
\setlength{\belowdisplayskip}{2pt}
\begin{align*}
    \mathcal{E} \bigpar{\ppi^*, \sigmaa_{\ppi^*}} &= \max_{\ppi \in \Pi} \mathcal{E} \bigpar{ \ppi, \sigmaa_{\ppi} } \mathbb{I} \bigbra{ \ppi \text{ satisfies the no-audit condition in \eqref{eqn:p2_best_resp} for all } s \in \S} \\
    &\geq \max_{\ppi \in \Pi} \mathcal{E} \bigpar{ \ppi, \sigmaa_{\ppi} } \mathbb{I} \bigbra{ \bigpar{\ppi, \sigmaa_{\ppi}} \text{ is a \color{purple} Stackelberg \color{black} equilibrium} }.
\end{align*}}

Hence, $\mathcal{E}(\ppi^*, \sigmaa_{\ppi^*})$, which equals $\operatorname{OPT}(\mathcal L_{c,k})-\sum_m q_mf(m)$, is an upper bound on the excess payments that occur in any \color{purple} Stackelberg \color{black} equilibrium.

Conversely, this bound is tight, because Lemma~\ref{lem:interim_to_exante} shows that every Bayesian persuasion equilibrium is also a \color{purple} Stackelberg \color{black} equilibrium. 

\subsection{Proof of Lemma~\ref{lem:interim_to_exante}} \label{pf:lem:interim_to_exante}


Let $\bigpar{\ppi^*, \sigmaa_{\ppi^*}}$ be any Bayesian Persuasion equilibrium. We will show that $\bigpar{\ppi^*, \sigmaa_{\ppi^*}}$ is also a \color{purple} Stackelberg \color{black} equilibrium. To this end, we will need the following observation.

\begin{observation}[No-Audit Optimality in \color{purple} Stackelberg \color{black} Games]\label{obs:sg-no-audit-optimality}
Consider any strategy profile $\bigpar{\ppi, \sigmaa_{\ppi}}$ where $\sigmaa_{\ppi}$ denotes the administrator's best response to the user strategy $\ppi$ as described in~\eqref{eqn:p2_best_resp}. If $\sigmaa_{\ppi}$ is not identically zero, i.e., there exists $s \in \S$ for which $\sigmaa_{\ppi}(s) = 1$, then $\bigpar{\ppi, \sigmaa_{\ppi}}$ cannot be a \color{purple} Stackelberg \color{black} equilibrium. 
\end{observation}

\begin{proof}[Proof of Observation~\ref{obs:sg-no-audit-optimality}.]
    Suppose there exists $s$ for which $\sigmaa_{\ppi}(s) = 1$. According to \eqref{eqn:p2_best_resp}, there must exist some $m \neq s$ for which $\ppi(s | m) > 0$. Consider $\ppi'(\cdot | m)$ defined as follows:
    \begin{align*}
        \ppi'(s' | m) &= \casewise{
            \begin{tabular}{cc}
                $\ppi(s' | m)$ & if $s' \not \in \bigbrace{s,m}$ \\
                $0$ & if $s' = s$ \\
                $\ppi(m | m) + \ppi(s | m)$ & if $s'=m$.
            \end{tabular}
        }
    \end{align*}

    For every $s' \not\in \bigbrace{s,m}$, $\ppi(s' | m) = \ppi'(s' | m)$, so by \eqref{eqn:p2_best_resp}, $\sigma_{\ppi}(s') = \sigmaa_{\ppi'}(s')$. For $s' = m$, when replacing $\ppi$ with $\ppi'$, the right hand side of \eqref{eqn:p2_best_resp} increases while the left hand side stays the same, implying $\sigma_{\ppi}(m) \geq \sigmaa_{\ppi'}(m)$. Using these observations, we note that 
    {\setlength{\abovedisplayskip}{2pt}
\setlength{\belowdisplayskip}{2pt}
    \begin{align*}
        \mathcal{U}^R_m(\ppi, \sigmaa_{\ppi}) &= \sum_{s' \in \S} \ppi(s' | m) \bigbra{f(s') - \sigma_{\ppi}(s') \mathbb{I}_{[s' \neq m]}\bigpar{\bigbra{f(s') - f(m)}_+ + k}} \\
        &= \sum_{s' \not\in \bigbrace{s,m}} \ppi(s' | m) \bigbra{f(s') - \sigma_{\ppi}(s') \mathbb{I}_{[s' \neq m]}\bigpar{\bigbra{f(s') - f(m)}_+ + k}} \\
        &+ \ppi(s | m) \bigbra{f(s) - \sigma_{\ppi}(s) \mathbb{I}_{[s \neq m]}\bigpar{\bigbra{f(s) - f(m)}_+ + k}} + \ppi(m | m) f(m) \\
        &\leq \sum_{s' \not\in \bigbrace{s,m}} \ppi'(s' | m) \bigbra{f(s') - \sigma_{\ppi'}(s') \mathbb{I}_{[s' \neq m]}\bigpar{\bigbra{f(s') - f(m)}_+ + k}} \\
        &+ \ppi(s | m) \bigbra{f(m) - k} + \ppi(m | m) f(m) \\
        &= - k \ppi(s | m) +  \sum_{s' \in \S} \ppi'(s' | m) \bigbra{f(s') - \sigma_{\ppi'}(s') \mathbb{I}_{[s' \neq m]}\bigpar{\bigbra{f(s') - f(m)}_+ + k}} \\
        &= - k \ppi(s | m) + \mathcal{U}^R_m(\ppi', \sigmaa_{\ppi'}).
    \end{align*}}
    Since $\ppi(s | m) > 0$, we see that under the strategy profile $\bigpar{\ppi, \sigmaa_{\ppi}}$, type $m$ can improve its utility by changing its \color{purple} reporting \color{black} strategy, hence $\bigpar{\ppi, \sigmaa_{\ppi}}$ cannot be a \color{purple} Stackelberg \color{black} equilibrium.  
\end{proof}

\begin{proof}[Proof of Lemma~\ref{lem:interim_to_exante}]
Let $(\ppi^*,\sigmaa_{\ppi^*})$ be any Bayesian persuasion
equilibrium. By Lemma~\ref{lem:signalingeqStrategy}, we know
that $\sigmaa_{\ppi^*}(s)=0$ for every $s\in\mathcal S$.
Moreover, by Theorem~\ref{thm:main}, $\ppi^*$ is an optimal
solution to the linear program $\mathcal L_{c,k}$ defined
in~\eqref{eq:obj-lp-bp}--\eqref{eq:no-audit-con-bp}.
By Lemma~\ref{lem:underreporting_dominated}, we also know
that
\[
  \ppi^*(s\mid m)>0
  \quad\Longrightarrow\quad
  f(s)\geq f(m).
\]

For any reporting strategy $\ppi$, define the administrator's
audit surplus at report $s$ as
\[
  G_s(\ppi)
  :=
  \sum_{\ell\neq s}q_\ell\ppi(s\mid\ell)
  \bigl(k+[f(s)-f(\ell)]_+\bigr)
  -c\sum_{\ell\in\mathcal S}q_\ell\ppi(s\mid\ell).
\]
The administrator audits report $s$ if and only if
$G_s(\ppi)>0$. In particular, the no-audit constraints in
$\mathcal L_{c,k}$ are equivalent to
$G_s(\ppi)\leq 0$ for every $s\in\mathcal S$.

We first establish a property of LP-optimal reporting strategies.
We claim that if a report $r$ receives positive misreporting
mass from some type $\ell\neq r$, then
\begin{equation}\label{eq:lemma1-binding-property}
  G_s(\ppi^*)=0
  \qquad
  \text{for every }s\in\mathcal S
  \text{ such that }f(s)>f(r).
\end{equation}
To prove this claim, suppose instead that there exists
$s\in\mathcal S$ with $f(s)>f(r)$ and
$G_s(\ppi^*)<0$. Since
$q_\ell\ppi^*(r\mid\ell)>0$, we can choose a sufficiently
small $\varepsilon>0$ and construct a new strategy by moving
joint probability mass $\varepsilon$ of type $\ell$ from
report $r$ to report $s$.

By the no-underreporting property,
$f(r)\geq f(\ell)$. Consequently, the changes in the
administrator's audit surpluses at reports $r$ and $s$
are respectively
\[
\begin{aligned}
  G_r(\widetilde\ppi)-G_r(\ppi^*)
  &=-\varepsilon
    \bigl(k+[f(r)-f(\ell)]_+-c\bigr)
    \leq 0,\\
  G_s(\widetilde\ppi)-G_s(\ppi^*)
  &=\varepsilon
    \bigl(k+f(s)-f(\ell)-c\bigr)
    >0.
\end{aligned}
\]
Since $G_s(\ppi^*)<0$, we can choose $\varepsilon$
small enough that the no-audit constraint at $s$ remains
satisfied. All other no-audit constraints are unchanged.
Thus $\widetilde\ppi$ is feasible for $\mathcal L_{c,k}$.
However, its objective exceeds that of $\ppi^*$ by
\[
  \varepsilon\bigl(f(s)-f(r)\bigr)>0,
\]
contradicting the optimality of $\ppi^*$. This proves
\eqref{eq:lemma1-binding-property}.

We now fix an arbitrary type $m\in\mathcal S$ and show that
its reporting strategy $\ppi^*(\cdot\mid m)$ is a best
response when the administrator optimally responds to any
unilateral deviation by that type.

Define
\[
  S(m):=\{s\in\mathcal S:\ppi^*(s\mid m)>0\},
  \qquad
  v_m:=\min_{s\in S(m)}f(s),
\]
and choose $r_m\in S(m)$ such that $f(r_m)=v_m$.
By the no-underreporting property, $v_m\geq f(m)$.

We next show that
\begin{equation}\label{eq:lemma1-support-binding}
  G_s(\ppi^*)=0
  \qquad\text{whenever }f(s)>v_m.
\end{equation}
If $r_m\neq m$, report $r_m$ receives positive
misreporting mass from type $m$, and
\eqref{eq:lemma1-binding-property} applies. If $r_m=m$
and report $m$ receives positive misreporting mass from
another type, the same claim applies.

It remains to consider the case where $r_m=m$ and no
other type sends positive probability to report $m$.
In this case,
\[
  G_m(\ppi^*)=-c q_m\ppi^*(m\mid m)<0.
\]
Suppose there were a report $s$ with $f(s)>f(m)$
and $G_s(\ppi^*)<0$. We could move a sufficiently small
amount $\varepsilon>0$ of type $m$'s truthful joint mass
from report $m$ to report $s$. The resulting changes in
the two audit surpluses would be
\[
\begin{aligned}
  G_m(\widetilde\ppi)-G_m(\ppi^*)&=c\varepsilon,\\
  G_s(\widetilde\ppi)-G_s(\ppi^*)
  &=\varepsilon\bigl(k+f(s)-f(m)-c\bigr)>0.
\end{aligned}
\]
Since both constraints were strictly slack, the transfer
can be chosen small enough to preserve feasibility.
But it increases the LP objective by
$\varepsilon(f(s)-f(m))>0$, again contradicting
optimality. Hence~\eqref{eq:lemma1-support-binding} holds.

Now consider an arbitrary alternative reporting strategy
$\widehat\ppi(\cdot\mid m)$ for type $m$, leaving all other
types' strategies unchanged. Let
\[
  \delta_s
  :=\widehat\ppi(s\mid m)-\ppi^*(s\mid m),
  \qquad s\in\mathcal S,
\]
so that $\sum_s\delta_s=0$.
Let $\widehat\sigma$ denote the administrator's best
response to the resulting strategy $\widehat\ppi$.

The changes in the administrator's audit surpluses are
\begin{equation}\label{eq:lemma1-audit-surplus-change}
G_s(\widehat\ppi)-G_s(\ppi^*)
=
\begin{cases}
q_m\delta_s
\bigl(k+[f(s)-f(m)]_+-c\bigr),&s\neq m,\\
-cq_m\delta_m,&s=m.
\end{cases}
\end{equation}

We compare the deviating type's payoff report by report.

First, suppose $\delta_s\leq0$. If $s\neq m$,
the coefficient multiplying $\delta_s$ in
\eqref{eq:lemma1-audit-surplus-change} is nonnegative
because $k\geq c$. Thus
$G_s(\widehat\ppi)\leq G_s(\ppi^*)\leq0$, and the
administrator does not audit report $s$. If $s=m$,
the report is truthful for type $m$, so the type receives
$f(m)$ regardless of whether the administrator audits.
Whenever $\delta_s<0$, we have
$\ppi^*(s\mid m)>0$, and therefore $f(s)\geq v_m$.
It follows that, whenever $\delta_s\leq0$,
\begin{equation}\label{eq:lemma1-negative-delta}
\widehat\ppi(s\mid m)
\left[
f(s)-\widehat\sigma(s)\mathbb I_{[s\neq m]}
\bigl([f(s)-f(m)]_++k\bigr)
\right]
\leq
\ppi^*(s\mid m)f(s)+\delta_s v_m.
\end{equation}

Next, suppose $\delta_s>0$. If $s=m$ or $f(s)\leq v_m$,
the payoff per unit of probability assigned to report $s$
is at most $f(s)\leq v_m$, since an audit cannot increase
the recipient's payoff. If instead $f(s)>v_m$, then
$s\neq m$ and, by~\eqref{eq:lemma1-support-binding},
$G_s(\ppi^*)=0$. Moreover, $f(s)>v_m\geq f(m)$, so
the coefficient multiplying $\delta_s$ in
\eqref{eq:lemma1-audit-surplus-change} is strictly
positive. Consequently,
$G_s(\widehat\ppi)>0$, and the administrator audits
report $s$. The deviating type's payoff from this report
is then
\[
  f(s)-[f(s)-f(m)]_+-k=f(m)-k<v_m.
\]
In either case, the payoff from the additional probability
mass assigned to report $s$ is at most $v_m$ per unit,
while the payoff from the probability mass already assigned
to $s$ under $\ppi^*$ can only weakly decrease. Therefore,
whenever $\delta_s>0$,
\begin{equation}\label{eq:lemma1-positive-delta}
\widehat\ppi(s\mid m)
\left[
f(s)-\widehat\sigma(s)\mathbb I_{[s\neq m]}
\bigl([f(s)-f(m)]_++k\bigr)
\right]
\leq
\ppi^*(s\mid m)f(s)+\delta_s v_m.
\end{equation}

Combining~\eqref{eq:lemma1-negative-delta} and
\eqref{eq:lemma1-positive-delta}, and using
$\sum_s\delta_s=0$, gives
\[
\begin{aligned}
  \mathcal U_m^R(\widehat\ppi,\widehat\sigma)
  &\leq
  \sum_s\ppi^*(s\mid m)f(s)
  +v_m\sum_s\delta_s\\
  &=\sum_s\ppi^*(s\mid m)f(s)\\
  &=\mathcal U_m^R(\ppi^*,\sigmaa_{\ppi^*}).
\end{aligned}
\]
Therefore, no type $m$ can profitably deviate when the
administrator optimally responds to the deviation.
Since $m$ was arbitrary, $(\ppi^*,\sigmaa_{\ppi^*})$
is a Stackelberg equilibrium.

It remains to show that the Stackelberg equilibria with
the largest expected excess payments are precisely the
Bayesian persuasion equilibria.

By Observation~\ref{obs:sg-no-audit-optimality}, every
Stackelberg equilibrium $(\bar\ppi,\bar\sigma)$ satisfies
$\bar\sigma(s)=0$ for every $s\in\mathcal S$.
Furthermore, no Stackelberg equilibrium can involve
underreporting. To see this, suppose that a type $m$
reports some $s$ with $f(s)<f(m)$ with positive
probability. Consider a unilateral deviation that moves
a small probability mass $\eta>0$ from report $s$
to the truthful report $m$. The changes in the
administrator's audit surpluses are
\[
\begin{aligned}
  G_s(\widetilde\ppi)-G_s(\bar\ppi)
  &=-q_m\eta(k-c)\leq0,\\
  G_m(\widetilde\ppi)-G_m(\bar\ppi)
  &=-cq_m\eta<0.
\end{aligned}
\]
Thus the administrator continues not to audit after
the deviation, while the type's expected utility
increases by $\eta(f(m)-f(s))>0$, a contradiction.

Hence every Stackelberg equilibrium satisfies the
no-underreporting property, and its expected excess
payments equal
\[
  \sum_{m\in\mathcal S}q_m\sum_{s\in\mathcal S}
  \bar\ppi(s\mid m)f(s)
  -\sum_{m\in\mathcal S}q_m f(m).
\]
Moreover, every Stackelberg equilibrium is feasible for
$\mathcal L_{c,k}$, since its administrator's best
response is to never audit. Since $\ppi^*$ solves
$\mathcal L_{c,k}$ and is itself a Stackelberg
equilibrium, it follows that Bayesian persuasion
equilibria attain the largest expected excess payments
among all Stackelberg equilibria.

Conversely, any Stackelberg equilibrium attaining this
maximum must be an optimal solution to
$\mathcal L_{c,k}$. By Theorem~\ref{thm:main}, its
strategy profile is therefore a Bayesian persuasion
equilibrium.
\end{proof}

\subsection{Proof of Corollary~\ref{cor:misrepotingProbBound}} \label{apdx:corPf-misreportingBound}

Let $(\ppi^*, \sigmaa^*)$ be a \color{purple} Stackelberg \color{black} equilibrium. We can deduce from Observation~\ref{obs:sg-no-audit-optimality} that $\ppi^*$ satisfies the no-audit condition in \eqref{eqn:p2_best_resp} for every $s \in \S$. Thus we have
{\setlength{\abovedisplayskip}{2pt}
\setlength{\belowdisplayskip}{2pt}
\begin{align*}
    & k \sum_{\m \in \S \setminus \{s\}} \q_{\m} \ppi(\s | \m) \leq \sum_{m \in \S} \q_{\m} \ppi(\s | \m) (c + \f(\m) - \f(\s)) \\
    &\iff \! \! \! \sum_{\m \in \S \setminus \{s\}} \! \! \! \q_{\m} \ppi(\s | \m) (k \! - \! c \! + \! \f(\s) \! - \! \f(\m))  \!  \leq \! \q_{\s} \ppi(s | m = s) c \\
    &\implies \! \! \! \sum_{\m \in \S \setminus \{s\}} \! \! \! \q_{\m} \ppi(\s | \m) (k \! - \! c \! + \! \f(\s) \! - \! \f(\m)) \! \leq \! \q_{\s} c
\end{align*}}
Thus for any $m \neq s$, the above inequality implies $\q_{\m} \ppi(\s | \m) (k - c + \f(\s) - \f(\m)) \leq \q_{\s} c,$ implying that $\ppi(\s | \m) \leq \frac{\q_{\s} c}{\q_{\m} (k-c+\f(\s) - \f(\m))}$, which implies the desired result, since $\ppi(\s | \m) \leq 1$.

\subsection{Proof of Corollary~\ref{cor:equilibrium_suboptimality}} \label{apdx:corPf-equilibrium_suboptimality}

Let $(\ppi, \sigmaa)$ be a Bayesian persuasion equilibrium in the audit game. To prove the desired result, we first establish the following two inequalities:
{\setlength{\abovedisplayskip}{2pt}
\setlength{\belowdisplayskip}{2pt}
\begin{align}
    \mathbb{E}_{(s,m) \sim \ppi}[f(s)] - \mathbb{E}_{m \sim q}[f(m)] &\leq \Delta f_{\max} \mathbb{P}_{\ppi}(s \neq m), \label{eqn:equilibrium_suboptimality_ineq1} \\ 
    \mathbb{E}_{(s,m) \sim \ppi}[f(s)] - \mathbb{E}_{m \sim q}[f(m)] &\leq c - k \mathbb{P}_{\ppi}(s \neq m). \label{eqn:equilibrium_suboptimality_ineq2}
\end{align}}
Once we establish these two inequalities, the desired result follows directly, since
{\setlength{\abovedisplayskip}{2pt}
\setlength{\belowdisplayskip}{2pt}
\begin{align*}
    &\mathbb{E}_{(s,m) \sim \ppi}[f(s)] - \mathbb{E}_{m \sim q}[f(m)] \leq \max_{\mathbb{P}_{\ppi}(s \neq m)} \min \left( \Delta f_{\max} \mathbb{P}_{\ppi}(s \neq m), c - k \mathbb{P}_{\ppi}(s \neq m) \right) = \frac{c \Delta f_{\max} }{k + \Delta f_{\max} },
\end{align*}}
which is achieved when $\mathbb{P}_{\ppi}(s \neq m) = \frac{c}{k + \Delta f_{\max} }$. Thus all that remains is to establish \eqref{eqn:equilibrium_suboptimality_ineq1} and \eqref{eqn:equilibrium_suboptimality_ineq2}. 

\textit{Establishing \eqref{eqn:equilibrium_suboptimality_ineq1}:} By H\"{o}lder's inequality,
{\setlength{\abovedisplayskip}{2pt}
\setlength{\belowdisplayskip}{2pt}
\begin{align*}
    \mathbb{E}_{(s,m) \sim \ppi}[f(s)] - \mathbb{E}_{m \sim q}[f(m)] &= \mathbb{E} \left[ (f(s) - f(m)) \mathbb{I}_{[s \neq m]} \right] \leq \mathbb{E} \left[ \Delta f_{\max} \mathbb{I}_{[s \neq m]} \right] = \Delta f_{\max} \mathbb{E} \left[ \mathbb{I}_{[s \neq m]} \right]\\
    &= \Delta f_{\max} \mathbb{P}_{\ppi}(s \neq m),
\end{align*}}
where the inequality is because $ f(s) - f(m)  \leq \max_{s,s'} f(s) - f(s') = \Delta f_{\max}$ with probability 1.

\textit{Establishing \eqref{eqn:equilibrium_suboptimality_ineq2}:} Recall from \eqref{eqn:p2_best_resp} that the administrator will conduct an audit upon seeing a \color{purple} report \color{black} $s$ if and only if $\sum_{m \neq s} \ppi(m | s) (k + f(s) - f(m)) > c$. For an equilibrium $(\ppi, \sigmaa)$, the administrator will never audit, which means that $\sum_{m \neq s} \ppi(m | s) (k + f(s) - f(m)) \leq c$ for every $s \in \S$. Taking expectation over $s$, we see that
{\setlength{\abovedisplayskip}{2pt}
\setlength{\belowdisplayskip}{2pt}
\begin{align*}
    c &\geq \mathbb{E}_{s \sim \ppi} \left[ \sum_{m \neq s} \ppi(m | s) (k + f(s) - f(m)) \right] = \sum_{s} \pi(s) \sum_{m \neq s} \pi(m | s) (k + f(s) - f(m)) \\
      &= \sum_{s,m} \ppi(\s,\m) (\f(\s) - \f(\m)) + k \sum_{s,m \neq s} \ppi(\s,\m) = \mathbb{E}_{(s,m) \sim \ppi}[f(s)] - \mathbb{E}_{m \sim q}[f(m)] + k \mathbb{P}_{\ppi}(s \neq m),
\end{align*}}
which establishes \eqref{eqn:equilibrium_suboptimality_ineq2}.

\subsection{Properties of the $|\S| = 2$ case} \label{apdx:S2_core_ideas}

This section derives some results that will be helpful for proofs that pertain to the $|\S| = 2$ setting. 

Let $S = \{L,H\}$ where $L, H$ represent the low and high types respectively with $\f(H) > \f(L)$ and $\Delta \f_{\max} = \f(H) - \f(L)$. Let $\q_L,\q_H$ denote the probabilities that a recipient has type $L, H$ respectively. In this simplified setting, we make the following observations:

\begin{enumerate}
    \item If $m = H$, choosing $s = H$ is a dominant strategy for the user. Indeed, if $m = H, s = H$ then the user receives utility $f(H)$ regardless of whether or not the administrator audits. On the other hand, if $m = H, s = L$, then the user receives $f(L)$ if the administrator does not audit, and $f(L) - k$ if the administrator audits.
    \item If $s = L$, then not auditing is a dominant strategy for the administrator. Indeed, by the first observation, $s = L$ only if $m = L$. Therefore the administrator's utility is $-f(L)$ if no audit is conducted, and $-f(L) - c$ if an audit is conducted. 
\end{enumerate}

From the first observation, an action $\ppi$ has optimal utility only if $\ppi(s = H | m = H) = 1$. From the second observation, we can conclude that $\sigmaa'(L) = 0$. 

We now prove the following lemmas which will be helpful in the proofs of Proposition~\ref{prop:counter-eg-nonExistence}, Proposition~\ref{prop:strongerBudgetCondition}, Theorem~\ref{thm:coalitions-ExEq}, and Theorem~\ref{thm:cost-comparison}. 

\begin{lemma}\label{lem:S=2_expected_utility}
    Let $\ppi$ be an action for which $\ppi(\s = H |\m = H) = 1$. Then the expected utility of $\ppi$ is
    {\setlength{\abovedisplayskip}{2pt}
    \setlength{\belowdisplayskip}{2pt}
    \begin{align*}
        \mathbb{E}_{\m \sim \q, \s \sim \ppi(\cdot | \m)} \left[U^i(\sigma(\s), \m)\right] = \q_L \f(L) + \q_H \f(H) + \q_L \ppi(\s = H | \m = L) (\Delta \f_{\max} - \sigmaa(H) (k + \Delta \f_{\max}) )
    \end{align*}}
\end{lemma}

\begin{proof}[Proof of Lemma \ref{lem:S=2_expected_utility}.]
    We show the result by direct calculation. By definition, we have 
    {\setlength{\abovedisplayskip}{2pt}
    \setlength{\belowdisplayskip}{2pt}
    \begin{align*}
        &\mathbb{E}_{\m \sim \q, \s \sim \ppi(\cdot | \m)} \left[U^i(\sigma(\s), \m)\right] \\
        &= \q_L \ppi(\s = L | \m = L) \f(L) + \q_L \ppi(\s = H | \m = L) \left( \f(H) + \sigmaa(H) (\f(L) - \f(H) - k) \right) + \q_H \f(H) \\
        &= \q_L \f(L) + \q_H \f(H) + \q_L \ppi(\s = H | \m = L) (\Delta \f_{\max} - \sigmaa(H) (k + \Delta \f_{\max}) ).
    \end{align*}}
    We used $\ppi(\s = L | \m = L) = 1 - \ppi(\s = H | \m = L)$ to go from the second to the third line. 
\end{proof}

\begin{lemma}\label{lem:S=2_audit_rule}
    We characterize the administrator's best response $\sigmaa$. First, from the second observation, we have $\sigmaa(L) = 0$. Furthermore, 
{\setlength{\abovedisplayskip}{2pt}
\setlength{\belowdisplayskip}{2pt}    
\begin{align*}
    \sigmaa(H) &= \left \{ 
    \begin{tabular}{cc}
        $0$ & if $\ppi(\s = H | \m = L) \leq \frac{\q_H c}{\q_L(k - c + \Delta \f_{\max})}$ \\
        $\min(1,\frac{B}{c})$ & otherwise.
    \end{tabular}
    \right.
\end{align*}}

\end{lemma}

\begin{proof}[Proof of Lemma \ref{lem:S=2_audit_rule}.]
From \eqref{eqn:p2_best_resp}, note that the administrator's utility is an affine function of $\sigmaa(H)$. The coefficient of $\sigmaa(H)$ is non-positive (i.e., $\sigmaa(H) = 0$) if and only if
{\setlength{\abovedisplayskip}{2pt}
\setlength{\belowdisplayskip}{2pt}
\begin{align*}
    k \q_L \ppi(s = H | m = L) &\leq  \q_L \ppi(s = H | m = L) (c - \Delta \f_{\max}) + \q_H \ppi (s = H | m = H) c \\
    \iff \q_L \ppi (s = H | m = L) (k - c + \Delta \f_{\max}) &\leq \q_H c \\
    \iff \ppi (s = H | m = L) &\leq \frac{\q_H c}{\q_L(k - c + \Delta \f_{\max})}.
\end{align*}}

If this condition is not satisfied, then the administrator will make $\sigmaa(H)$ as large as possible. 
\end{proof}

\begin{lemma}\label{lem:S=2_no_audit_optimum}
    Let $\Pi$ be the set of all user strategies for which the administrator's best response is to never audit, i.e., $\sigmaa(H) = \sigmaa(L) = 0$. The optimal user utility obtainable by a strategy in $\Pi$ is
    {\setlength{\abovedisplayskip}{2pt}
\setlength{\belowdisplayskip}{2pt}
    \begin{align*}
        \max_{\ppi \in \Pi} \mathbb{E}_{\m \sim \q, \s \sim \ppi(\cdot | \m)} \left[U^i(\sigma(\s), \m)\right] = \q_L \f(L) + \q_H \f(H) + \q_L \min \left(1, \frac{\q_H c}{\q_L(k - c + \Delta \f_{\max})}\right) \Delta \f_{\max}. 
    \end{align*}}
    Furthermore, the unique strategy $\ppi' \in \Pi$ which attains this optimal utility is given by
    {\setlength{\abovedisplayskip}{2pt}
\setlength{\belowdisplayskip}{2pt}
    \begin{align*}
        \ppi'(\s = H | \m = L) := \min \left(1, \frac{\q_H c}{\q_L(k - c + \Delta \f_{\max})}\right), \quad
        \ppi'(\s = H | \m = H) := 1.
    \end{align*}}
\end{lemma}

\begin{proof}[Proof of Lemma \ref{lem:S=2_no_audit_optimum}.]
Let $\ppi'$ be any optimal user utility strategy in $\Pi$, and let $\sigmaa'$ be the administrator's best response to $\ppi'$. From the first observation, we can conclude that $\ppi'(s = H | m = H) = 1$ and $\ppi' (s = L | m = H) = 0$. 

From Lemma \ref{lem:S=2_audit_rule}, $\ppi \in \Pi$ if and only if $\sigmaa(H) = 0$, which gives $\ppi (s = H | m = L) \leq \frac{\q_H c}{\q_L(k - c + \Delta \f_{\max})}.$

From Lemma \ref{lem:S=2_expected_utility}, the expected utility of a strategy $\ppi \in \Pi$ is given by 
{\setlength{\abovedisplayskip}{2pt}
\setlength{\belowdisplayskip}{2pt}
\begin{align*}
    &\mathbb{E}_{\m \sim \q, \s \sim \ppi(\cdot | \m)} \left[ \f(s) + \sigma(\s) \mathbb{I}_{[\s \neq \m]} (\f(m) - \f(\s) - k) \right] \\
    &= \q_L \left( \ppi(\s = L | \m = L) \f(L) + \ppi(\s = H | \m = L) \f(H) \right) + \q_H \f(H) \\
    &= \q_L \f(L) + \q_H \f(H) + \q_L \ppi(\s = H | \m = L) \Delta \f_{\max}.
\end{align*}}

Thus, the expected utility is an affine function of $\ppi(\s = H | \m = L)$. To maximize expected utility, $\ppi(\s = H | \m = L)$ should be maximized subject to the constraints imposed by $\Pi$, which also means that the optimal policy is unique. Hence, $\ppi'(\s = H | \m = L) = \min(1, \frac{\q_H c}{\q_L(k - c + \Delta \f_{\max})})$. 
\end{proof}

\subsection{Proof of Proposition~\ref{prop:counter-eg-nonExistence}} \label{apdx:pfcounter-eg-nonExistence}

Let $B < \c \frac{\Delta \f_{\max} \big( 1 - \min \{ \frac{q_{\max} \c}{q_{\min} (\k - \c + \Delta \f_{\max})} , 1\} \big)}{k+\Delta \f_{\max}}$. Consider a setting with two users, i.e., $n = 2$. Let $S = \{L,H\}$ where $L, H$ represent the low and high types respectively with $\f(H) > \f(L)$ and $\Delta \f_{\max} = \f(H) - \f(L)$. Let $\q_L,\q_H$ denote the probabilities that a recipient has type $L, H$ respectively. 

The proof has three main steps:

\begin{enumerate}
    \item First, we show that the best response of the administrator is to audit the user $i'$ which most frequently misreports their type, i.e., $i' = \arg\max_{i \in \{1,2\}} \ppi_i(\s = H | \m = L)$. 
    \item Second, we show that because the administrator's budget is very small, the audit penalty is not sufficient to deter users from misreporting their types. In other words, $\ppi(\s = H | \m = L) > \frac{\q_H c}{\q_L(k - c + \Delta \f_{\max})}$ gives users more utility than $\ppi(\s = H | \m = L) = \frac{\q_H c}{\q_L(k - c + \Delta \f_{\max})}$. 
    \item Third, we show that the first two steps lead to a race-to-the-bottom behavior. By the second step, both users want to misreport their types to increase their utility. However, the first step establishes that the administrator will audit the user which misreports most frequently. Thus both users want a misreporting probability that is positive but also smaller than the other user's misreporting probability. Such settings are a classical example where equilibria do not exist. 
\end{enumerate}

\hypertarget{audit_rule_max_violator}{\textit{Step 1:}} We show in the beginning of Appendix~\ref{apdx:S2_core_ideas} that choosing $\s = H$ is a dominant strategy for the user whenever $\m = H$. As a consequence, the administrator should never audit if $s = L$. With this in mind, when $\ppi(\s = H | \m = H) = 1$ and $\sigmaa(L) = 0$, the expected utility of the administrator when interacting with user $i \in \{1,2\}$ is given by:
{\setlength{\abovedisplayskip}{2pt}
\setlength{\belowdisplayskip}{2pt}
\begin{align*}
    &-\q_L \ppi_i(\s = L | \m = L) \f(L) \\
    &+ \q_L \ppi_i(\s = H | \m = L) \left( -\f(H) + \sigmaa_i(H) \left( -c + \mathbb{I}_{[\s \neq \m]} (k + \f(H) - \f(L) )\right) \right) \\
    &+ \q_H \ppi_i(\s = H | \m = H) \left( -\f(H) + \sigmaa_i(H) \left( -c + \mathbb{I}_{[\s \neq \m]} (k + \f(H) - \f(L) )\right) \right) \\
    &= -\q_L \ppi_i(\s = L | \m = L) \f(L) - \q_L \ppi_i(\s = H | \m = L) \f(H) - \q_H \f(H) \\
    &+ \sigmaa_i(H) \left(\q_L \ppi_i(\s = H | \m = L) (k + \f(H) - \f(L)) - \left[\q_H + \q_L \ppi_i(\s = H | \m = L)\right] c \right).
\end{align*}}

Let $\rho_i := \q_L \ppi_i(\s = H | \m = L) (k + \f(H) - \f(L)) - \left[\q_H + \q_L \ppi_i(\s = H | \m = L)\right] c$. The key observation here is that the administrator's utility when interacting with user $i$ is affine in $\sigmaa_i(H)$. Further, the administrator will prioritize auditing the user whose coefficient of $\sigmaa_i(H)$ is largest and positive. In other words, the administrator will audit user $i'$ where $i' = \arg\max_{i \in \{1,2\}} \rho_i$ and if $\rho_{i'} > 0$. 

Equivalently, the administrator will audit user $i'$ if
{\setlength{\abovedisplayskip}{2pt}
\setlength{\belowdisplayskip}{2pt}
\begin{align*}
    i' \in \arg\max_{i \in \{1,2\}} \rho_i \text{ and } \ppi_{i'}(\s = H | \m = L) > \frac{\q_H c}{\q_L(k - c + \Delta \f_{\max})}.
\end{align*}}

\textit{Step 2:} Using Lemma \ref{lem:S=2_expected_utility} from Appendix \ref{apdx:S2_core_ideas}, the expected utility received by user $i$ from action $\ppi$ is given by
{\setlength{\abovedisplayskip}{2pt}
\setlength{\belowdisplayskip}{2pt}
\begin{align*}
    \mathbb{E}_{\m \sim \q, \s \sim \ppi_i(\cdot | \m)} \left[U^i(\sigma(\s), \m)\right] = \q_L \f(L) + \q_H \f(H) + \q_L \ppi(\s = H | \m = L) (\Delta \f_{\max} - \sigmaa_i(H) (k + \Delta \f_{\max}) )
\end{align*}}

Note that for a fixed value of $\sigmaa_i^*(H)$, the expected utility is an affine function of $\ppi_i(\s = H | \m = L)$. Furthermore, according to Lemma \ref{lem:S=2_audit_rule}, 
{\setlength{\abovedisplayskip}{2pt}
\setlength{\belowdisplayskip}{2pt}
\begin{align*}
    \sigmaa_i(H) &= \left \{ 
    \begin{tabular}{cc}
        $0$ & if $\ppi_i(\s = H | \m = L) \leq \frac{\q_H c}{\q_L(k - c + \Delta \f_{\max})}$ \\
        $\frac{B}{c}$ & otherwise.
    \end{tabular}
    \right.
\end{align*}}

Hence the optimal value of $\mathbb{E}_{\m \sim \q, \s \sim \ppi(\cdot | \m)} \left[U^i(\sigma(\s), \m)\right]$ must be attained at $ \ppi_i(\s = H | \m = L) \in \{ 0, \frac{\q_H c}{\q_L(k - c + \Delta \f_{\max})}, 1 \}$. From Lemma \ref{lem:S=2_no_audit_optimum}, we know that $0$ cannot be the optimum. For $\ppi_i(\s = H | \m = L) = 1$, we note that
{\setlength{\abovedisplayskip}{2pt}
\setlength{\belowdisplayskip}{2pt}
\begin{align*}
    &\mathbb{E}_{\m \sim \q, \s \sim \ppi(\cdot | \m)} \left[U^i(\sigma(\s), \m)\right] = \q_L \f(L) + \q_H \f(H) + \q_L \left(\Delta \f_{\max} - \frac{B}{c} (k + \Delta \f_{\max}) \right) \\
    &> \q_L \f(L) + \q_H \f(H) + \q_L \left(\Delta \f_{\max} - \frac{\Delta \f_{\max} \big( 1 - \min \{ \frac{\q_H \c}{\q_L (\k - \c + \Delta \f_{\max})} , 1\} \big) (k + \Delta \f_{\max})}{k+\Delta \f_{\max}} \right) \\
    &= \q_L \f(L) + \q_H \f(H) + \q_L \left(\Delta \f_{\max} - \Delta \f_{\max} \big( 1 - \min \{ \frac{\q_H \c}{\q_L (\k - \c + \Delta \f_{\max})} , 1\} \big) \right) \\
    &= \q_L \f(L) + \q_H \f(H) + \q_L \left(\Delta \f_{\max} \min \{ \frac{\q_H \c}{\q_L (\k - \c + \Delta \f_{\max})} , 1\} \right) 
\end{align*}}
which shows that $\ppi_i(\s = H | \m = L) = 1$ has greater utility than $\ppi_i(\s = H | \m = L) = \frac{\q_H c}{\q_L(k - c + \Delta \f_{\max})}$, since the final line is the expected utility obtained by the user when $\ppi_i(\s = H | \m = L) = \frac{\q_H c}{\q_L(k - c + \Delta \f_{\max})}$ (established by Lemma \ref{lem:S=2_no_audit_optimum}).

\textit{Step 3:} For notational convenience, let $p_i := \ppi_i(\s = H | \m = L)$. Consider any action profile $(p_1,p_2)$ for the players. We will show that no equilibrium exists by showing that there is always at least one player that can improve its utility by changing its probability of misreporting. There are three cases:
\begin{enumerate}
    \item There exists $i$ for which $p_i < \frac{\q_H c}{\q_L(k - c + \Delta \f_{\max})}$. By Lemma \ref{lem:S=2_no_audit_optimum}, user $i$ can improve its utility by increasing $p_i$ to $\frac{\q_H c}{\q_L(k - c + \Delta \f_{\max})}$. By step 2 of this proof, user $i$ can further increase its utility by increasing $p_i$ from $\frac{\q_H c}{\q_L(k - c + \Delta \f_{\max})}$ to $1$. Hence user $i$ can improve its utility by unilaterally changing its action. 
 
    \item $\frac{\q_H c}{\q_L(k - c + \Delta \f_{\max})} = p_1 = p_2$. By Step 2, a player who increases its misreporting probability to $1$ obtains strictly greater utility than if it were to report $p_i = \frac{\q_H c}{\q_L(k - c + \Delta \f_{\max})}$ and not be audited. Since auditing can only decrease a player's utility, both players have a profitable unilateral deviation to $p_i = 1$, regardless of how the administrator chooses to audit when it is indifferent on auditing each player.

    \item $\frac{\q_H c}{\q_L(k - c + \Delta \f_{\max})} \leq p_1 < p_2$. Due to step 1, we know that user 1 will not be audited (i.e., $\sigmaa_1(H) = 0$) provided that $p_1 \in [0,p_2)$. From Lemma \ref{lem:S=2_expected_utility}, when $\sigmaa_1(H) = 0$, the expected utility of user $1$ is an increasing and affine function of $p_1$. Therefore user 1 can unilaterally improve its utility by changing $p_1$ to any value in the open interval $(p_1,p_2)$. The same argument applies to the case where $\frac{\q_H c}{\q_L(k - c + \Delta \f_{\max})} \leq p_2 < p_1$.

    \item $p_1,p_2 > \frac{\q_H c}{\q_L(k - c + \Delta \f_{\max})}$ and $p_1 = p_2$. In this case, it is possible that both $\sigma_1(H)$ and $\sigma_2(H)$ are positive. Due to the budget constraint, we have $\sigmaa_1(H) + \sigmaa_2(H) = \frac{B}{c}$. Without loss of generality suppose $\sigmaa_1(H) \geq \sigmaa_2(H)$, implying that $\sigmaa_1(H) \geq \frac{B}{2c}$. From Lemma \ref{lem:S=2_expected_utility}, the expected utility for user 1 is:
    {\setlength{\abovedisplayskip}{2pt}
\setlength{\belowdisplayskip}{2pt}
    \begin{align*}
        & \q_L \f(L) + \q_H \f(H) + \q_L p_1 \left( \Delta \f_{\max} - \sigmaa_1(H)(k + \Delta \f_{\max}) \right) \\
        &\leq \q_L \f(L) + \q_H \f(H) + \q_L p_1 \left( \Delta \f_{\max} - \frac{B}{2c}(k + \Delta \f_{\max}) \right).
    \end{align*}}
    By the reasoning of step 1, if user $1$ chooses a misreporting probability $\pi_1(\s = H | \m = L) < p_1$, then they will no longer be audited. Thus user $1$ can reduce $p_1$ by an infinitesimal amount so that it is no longer tied for having the highest misreporting probability, thereby eliminating its audit penalty. In particular, changing $p_1$ to $p_1' < p_2$ changes user $1$'s audit probability from $\sigmaa_1(H)$ to $\sigmaa_1'(H) = 0$. User $1$'s new utility after this change is thus $\q_L \f(L) + \q_H \f(H) + \q_L p_1' \Delta \f_{\max}.$ This change leads to a utility improvement if  
    {\setlength{\abovedisplayskip}{2pt}
\setlength{\belowdisplayskip}{2pt}
    \begin{align*}
        \q_L p_1' \Delta \f_{\max} > \q_L p_1 \left( \Delta \f_{\max} - \frac{B}{2c}(k + \Delta \f_{\max}) \right) \quad 
        \iff \quad p_1' > p_1 \frac{\Delta \f_{\max} - \frac{B}{2c}(k + \Delta \f_{\max})}{ \Delta \f_{\max} }.
    \end{align*}}
    Thus user 1 can unilaterally improve its utility by changing $p_1$ to any value in the open interval $\left( p_2 \frac{\Delta \f_{\max} - \frac{B}{2c}(k + \Delta \f_{\max})}{ \Delta \f_{\max} } , p_2 \right)$.
\end{enumerate}

\subsection{Proof of Theorem~\ref{thm:eqEx-suff-budget}} \label{apdx:pfEx-suff-budget}

We prove the theorem by fixing an arbitrary recipient $i$, holding
all other recipients at their proposed LP-optimal strategies, and
showing that the administrator's best response to any unilateral
deviation by recipient $i$ reduces to the single-recipient
budgeted response problem.

For each recipient $i$, let $\ppi_i^*$ be an optimal solution to
the linear program $\mathcal{L}_{c,k}$ defined in
\eqref{eq:obj-lp-bp}--\eqref{eq:no-audit-con-bp}. Consider the
strategy profile in which each recipient uses $\ppi_i^*$ and
the administrator never audits. We show that this profile is a
Stackelberg equilibrium.

Throughout the proof, it is convenient to write
\begin{align*}
x_i(s,m) &:= q_m\ppi_i(s\mid m)
\end{align*}
for the joint probability that recipient $i$ has type $m$ and
reports $s$. Define the administrator's incremental expected
utility from auditing recipient $i$ after report $s$ by
\begin{align*}
G_{i,s}(x_i)
&:=
\sum_{m\ne s}x_i(s,m)
\bigl(k+[f(s)-f(m)]_+\bigr)
-c\sum_{m\in\mathcal S}x_i(s,m).
\end{align*}
Thus, $G_{i,s}(x_i)$ is the change in the administrator's
expected utility from auditing recipient $i$ after report $s$,
relative to not auditing that recipient. Since $\ppi_i^*$ is
feasible for $\mathcal{L}_{c,k}$, we have
\begin{align}
G_{i,s}(x_i^*) &\leq 0
\qquad
\text{for every }i\in[n]\text{ and }s\in\mathcal S.
\label{eq:B10-no-audit}
\end{align}

\paragraph{The administrator's best response to a unilateral deviation.}
Fix an arbitrary recipient $i$, and suppose that all other
recipients $j\ne i$ continue to use their proposed strategies
$\ppi_j^*$. Let $\widehat\ppi_i$ be an arbitrary reporting
strategy for recipient $i$. Because the administrator's
expected utility is additive across recipients, its
audit-dependent payoff is
\begin{align*}
\sum_{j=1}^n\sum_{s\in\mathcal S}
\sigma_j(s)G_{j,s}(x_j),
\end{align*}
where $x_j$ denotes the joint type-report distribution induced
by recipient $j$'s strategy.

For every $j\ne i$, inequality~\eqref{eq:B10-no-audit}
continues to hold. Hence auditing any such recipient cannot
increase the administrator's expected utility. By our
convention of not auditing when indifferent, the administrator's
best response sets
\begin{align*}
\sigma_j(s)=0
\qquad
\text{for all }j\ne i\text{ and }s\in\mathcal S.
\end{align*}
Consequently, the entire budget is available for recipient
$i$. Under the budget constraint in Section~\ref{sec:multiRecipientMain},
the administrator's best response to $\widehat\ppi_i$ is therefore
\begin{align}
\sigma_i^{\mathrm{BR}}(s)
&=
\begin{cases}
\bar\sigma, & G_{i,s}(\widehat x_i)>0,\\
0, & G_{i,s}(\widehat x_i)\leq 0,
\end{cases}
\qquad
\bar\sigma:=\min\left\{1,\frac{B}{c}\right\},
\label{eq:B10-audit-response}
\end{align}
where $\widehat x_i(s,m)=q_m\widehat\ppi_i(s\mid m)$.
In particular, the administrator's response to recipient $i$
is precisely the single-recipient budgeted response rule.

\paragraph{Eliminating underreports.}
We first show that, without loss of generality, we may restrict
attention to deviations that do not underreport. Suppose that
$\widehat x_i(s,m)>0$ for some pair $(s,m)$ satisfying
$f(s)<f(m)$. Transfer an amount $\varepsilon>0$ of joint
probability mass from $(s,m)$ to $(m,m)$, where
$\varepsilon\leq\widehat x_i(s,m)$. This preserves the
marginal distribution of recipient $i$'s type.

The transfer changes the administrator's audit incentives
only at reports $s$ and $m$, and the corresponding changes are
\begin{align*}
G_{i,s}(\widehat x_i')-G_{i,s}(\widehat x_i)
&=-(k-c)\varepsilon\leq 0,\\
G_{i,m}(\widehat x_i')-G_{i,m}(\widehat x_i)
&=-c\varepsilon\leq 0.
\end{align*}
It follows from \eqref{eq:B10-audit-response} that the
administrator's audit probabilities weakly decrease at every
report. Moreover, the transferred mass now yields the
truthful benefit $f(m)$, rather than $f(s)$ less the
audit penalty. Since $f(m)>f(s)$, this transfer strictly
increases the utility associated with the transferred mass.
The utility associated with all other reports weakly increases
because their audit probabilities do not increase. Thus this
transfer weakly increases recipient $i$'s expected utility.

Repeating this argument, we can transform any deviation into
one satisfying
\begin{align}
x_i(s,m)>0
\quad\Longrightarrow\quad
f(s)\geq f(m),
\label{eq:B10-no-underreports}
\end{align}
without decreasing the recipient's expected utility.

\paragraph{Eliminating reports that induce audits.}
Now consider a deviation satisfying \eqref{eq:B10-no-underreports}.
We show that it can be transformed into a strategy satisfying
all the no-audit constraints, without decreasing recipient
$i$'s expected utility.

Suppose there exists a report $s'\in\mathcal S$ such that
\begin{align*}
G_{i,s'}(x_i)>0.
\end{align*}
By \eqref{eq:B10-audit-response}, the administrator audits
recipient $i$ after this report with probability $\bar\sigma$.
Moreover, $G_{i,s'}(x_i)>0$ implies that there is positive
misreporting mass at report $s'$. Construct a new joint
distribution $\widetilde x_i$ by transferring all such
misreporting mass to truthful reporting. Specifically, for
each $m\ne s'$, set
\begin{align*}
\widetilde x_i(s',m)&=0,\\
\widetilde x_i(m,m)&=x_i(m,m)+x_i(s',m),
\end{align*}
and leave all other entries unchanged. This transformation
preserves the marginal distribution of recipient $i$'s type.

We first show that the administrator's audit incentives
weakly decrease at every report. At report $s'$, all
misreporting mass has been removed, so
\begin{align*}
G_{i,s'}(\widetilde x_i)
=-c\widetilde x_i(s',s')\leq 0.
\end{align*}
For every other report $r\ne s'$, the misreporting mass
at $r$ is unchanged, while the truthful mass at $r$
increases by $x_i(s',r)$. Therefore,
\begin{align*}
G_{i,r}(\widetilde x_i)
=G_{i,r}(x_i)-c x_i(s',r)
\leq G_{i,r}(x_i).
\end{align*}
It follows from \eqref{eq:B10-audit-response} that the
administrator's best-response audit probabilities satisfy
\begin{align}
\widetilde\sigma_i(r)
\leq
\sigma_i(r)
\qquad
\text{for every }r\in\mathcal S,
\label{eq:B10-audit-monotonicity}
\end{align}
where $\widetilde\sigma_i$ denotes the best response to
$\widetilde x_i$.

We next show that this transformation weakly increases
recipient $i$'s expected utility. For any mass
$x_i(s',m)>0$ being transferred, condition
\eqref{eq:B10-no-underreports} implies $f(s')\geq f(m)$.
The recipient's utility per unit of this mass before the
transfer is
\begin{align*}
f(s')-\bar\sigma\bigl(k+f(s')-f(m)\bigr),
\end{align*}
whereas its utility after the transfer is $f(m)$. Since
\begin{align*}
0\leq f(s')-f(m)\leq\Delta f_{\max}
\end{align*}
and the budget condition in the theorem gives
\begin{align*}
\bar\sigma
=\min\left\{1,\frac{B}{c}\right\}
\geq
\frac{\Delta f_{\max}}{k+\Delta f_{\max}},
\end{align*}
we have
\begin{align*}
\bar\sigma
\geq
\frac{f(s')-f(m)}
{k+f(s')-f(m)}.
\end{align*}
Consequently,
\begin{align}
f(m)
\geq
f(s')-\bar\sigma\bigl(k+f(s')-f(m)\bigr).
\label{eq:B10-transfer-utility}
\end{align}
Thus, the transferred mass weakly increases the recipient's
utility. For all mass that is not transferred, the audit
probability weakly decreases by \eqref{eq:B10-audit-monotonicity},
so its contribution to the recipient's expected utility also
weakly increases.

We can repeat this transformation whenever some report has
a strictly positive audit incentive. Each transformation
removes all misreporting mass from the selected report and
does not create any new misreporting mass at other reports.
Since $\mathcal S$ is finite, after finitely many
transformations we obtain a joint distribution
$\widetilde x_i$ satisfying
\begin{align*}
G_{i,s}(\widetilde x_i)\leq 0
\qquad\text{for every }s\in\mathcal S,
\end{align*}
and hence the administrator's best response is not to audit.
By the preceding arguments, this sequence of transformations
weakly increases recipient $i$'s expected utility.

\paragraph{Applying LP optimality.}
We have shown that any unilateral deviation
$\widehat\ppi_i$ can be transformed into a strategy
$\widetilde\ppi_i$ satisfying the no-audit constraints
and yielding weakly greater expected utility to recipient
$i$. Since $\widetilde\ppi_i$ is feasible for
$\mathcal L_{c,k}$, the optimality of $\ppi_i^*$ implies
\begin{align*}
&\mathbb E\!\left[
U_i^R\bigl(\sigma_i^{\mathrm{BR}},m_i\bigr)
\mid \widehat\ppi_i
\right]\\
&\qquad\leq
\sum_{m\in\mathcal S}\sum_{s\in\mathcal S}
\widetilde x_i(s,m)f(s)\\
&\qquad\leq
\sum_{m\in\mathcal S}\sum_{s\in\mathcal S}
x_i^*(s,m)f(s)\\
&\qquad=
\mathbb E\!\left[
U_i^R(0,m_i)\mid\ppi_i^*
\right].
\end{align*}
Thus recipient $i$ cannot profitably deviate from
$\ppi_i^*$. Since $i$ was arbitrary, no recipient has
a profitable unilateral deviation. Furthermore, by
\eqref{eq:B10-no-audit}, the administrator's best response
to the proposed profile is not to audit. Therefore the
proposed strategy profile is a Stackelberg equilibrium,
which establishes the theorem.

\subsection{Proof of Proposition~\ref{prop:strongerBudgetCondition}} \label{apdx:pfStrongerBudgetCondition}

As in the proof of Theorem~\ref{thm:eqEx-suff-budget} in Appendix~\ref{apdx:pfEx-suff-budget}, we establish that the proposed strategy profile is an equilibrium by considering unilateral deviations. Since $p^*$ satisfies the no-audit condition in Lemma~\ref{lem:S=2_audit_rule}, the administrator's best response to the proposed profile is not to audit any recipient. Moreover, when considering a unilateral deviation by recipient $i$, all other recipients remain at their proposed strategies, which continue to satisfy the no-audit condition. Thus, by the argument in Appendix~\ref{apdx:pfEx-suff-budget}, the administrator's best response to recipient $i$'s deviation reduces to the single-recipient budgeted response, with audit probability $\bar\sigma:=\min{1,B/c}$ whenever auditing is optimal.

We will show that under the condition $B \geq \c \frac{\Delta \f_{\max} \big( 1 - \min \{ \frac{q_{\max} \c}{q_{\min} (\k - \c + \Delta \f_{\max})} , 1\} \big)}{k+\Delta \f_{\max}}$, an equilibrium exists where the administrator does not audit, and all users' strategies satisfy constraints \eqref{eqn:p2_best_resp}. Throughout the proof we will use some results from Appendix \ref{apdx:S2_core_ideas}. 

Let $S = \{L,H\}$ where $L, H$ represent the low and high types respectively with $\f(H) > \f(L)$ and $\Delta \f_{\max} = \f(H) - \f(L)$. Let $\q_L,\q_H$ denote the probabilities that a recipient has type $L, H$ respectively. 

We consider the symmetric strategy profile in which every recipient $i$ plays
\begin{align*}
    \ppi_i(\s = H \mid \m = L) &= \min\bigpar{1,
    \frac{\q_H c}{\q_L(k-c+\Delta \f_{\max})}},\\
    \ppi_i(\s = H \mid \m = H) &= 1.
\end{align*}
We will show that this strategy profile, together with the administrator's no-audit strategy, constitutes an equilibrium.

First note that the budgeted setting is equivalent to the unlimited budget setting if $B \geq c$, in which case Theorem \ref{thm:main} establishes existence and computability of equilibria. Thus, in the rest of the proof, we consider the $B < c$ case. 

We show in the beginning of Appendix~\ref{apdx:S2_core_ideas} that choosing $\s = H$ is a dominant strategy for the user whenever $\m = H$. Therefore all that remains is to determine the optimal value for $\ppi(\s = H | \m = L)$.

Using Lemma \ref{lem:S=2_expected_utility} from Appendix \ref{apdx:S2_core_ideas}, the expected utility of any user action $\ppi$ is given by
{\setlength{\abovedisplayskip}{2pt}
\setlength{\belowdisplayskip}{2pt}
\begin{align*}
    \mathbb{E}_{\m \sim \q, \s \sim \ppi(\cdot | \m)} \left[U^i(\sigma(\s), \m)\right] = \q_L \f(L) + \q_H \f(H) + \q_L \ppi(\s = H | \m = L) (\Delta \f_{\max} - \sigmaa(H) (k + \Delta \f_{\max}) )
\end{align*}}


If $\frac{\q_H c}{\q_L(k-c+\Delta \f_{\max})} \geq 1$, then every feasible misreporting probability satisfies the no-audit condition in Lemma~\ref{lem:S=2_audit_rule}. Consequently, the administrator's best response is to never audit, including when it is indifferent between auditing and not auditing, by our equilibrium selection convention. Since choosing $\s = H$ is a dominant strategy for high-type recipients, it remains only to consider the reporting probability $\ppi_i(\s = H \mid \m = L)$ of low-type recipients. Given that the administrator never audits, a recipient's expected utility is strictly increasing in this probability. Hence, each recipient optimally chooses $\ppi_i(\s = H \mid \m = L) = 1$, establishing that the proposed strategy profile is an equilibrium.

Note that for a fixed value of $\sigmaa(H)$, the expected utility is an affine function of $\ppi(\s = H | \m = L)$. Furthermore, according to 
Lemma~\ref{lem:S=2_audit_rule}, 
{\setlength{\abovedisplayskip}{0pt}
\setlength{\belowdisplayskip}{0pt}
\begin{align*}
    \sigmaa(H) &= \left \{ 
    \begin{tabular}{cc}
        $0$ & if $\ppi(\s = H | \m = L) \leq \frac{\q_H c}{\q_L(k - c + \Delta \f_{\max})}$ \\
        $\frac{B}{c}$ & otherwise.
    \end{tabular}
    \right.
\end{align*}}

Hence the optimal value of $\mathbb{E}_{\m \sim \q, \s \sim \ppi(\cdot | \m)} \left[U^i(\sigma(\s), \m)\right]$ must be attained at $ \ppi(\s = H | \m = L) \in \{ 0, \frac{\q_H c}{\q_L(k - c + \Delta \f_{\max})}, 1 \}$. From Lemma \ref{lem:S=2_no_audit_optimum}, we know that $0$ cannot be the optimum. For $\ppi(\s = H | \m = L) = 1$, we note that
{\setlength{\abovedisplayskip}{0pt}
\setlength{\belowdisplayskip}{0pt}
\begin{align*}
    &\mathbb{E}_{\m \sim \q, \s \sim \ppi(\cdot | \m)} \left[U^i(\sigma(\s), \m)\right] = \q_L \f(L) + \q_H \f(H) + \q_L \left(\Delta \f_{\max} - \frac{B}{c} (k + \Delta \f_{\max}) \right) \\
    &\leq \q_L \f(L) + \q_H \f(H) + \q_L \left(\Delta \f_{\max} - \frac{\Delta \f_{\max} \big( 1 - \min \{ \frac{q_{\max} \c}{q_{\min} (\k - \c + \Delta \f_{\max})} , 1\} \big) (k + \Delta \f_{\max})}{k+\Delta \f_{\max}} \right) \\
    &= \q_L \f(L) + \q_H \f(H) + \q_L \left(\Delta \f_{\max} - \Delta \f_{\max} \big( 1 - \min \{ \frac{q_{\max} \c}{q_{\min} (\k - \c + \Delta \f_{\max})} , 1\} \big) \right) \\
    &= \q_L \f(L) + \q_H \f(H) + \q_L \left(\Delta \f_{\max} \min \{ \frac{q_{\max} \c}{q_{\min} (\k - \c + \Delta \f_{\max})} , 1\} \right) 
\end{align*}}
The final line is the expected utility obtained when $\ppi(\s = H | \m = L) = \frac{\q_H c}{\q_L(k - c + \Delta \f_{\max})}$, due to Lemma \ref{lem:S=2_no_audit_optimum}. The inequality is due to the lower bound on $B$ in the condition of Proposition~\ref{prop:strongerBudgetCondition}. 

Hence the optimal user action is given by
{\setlength{\abovedisplayskip}{0pt}
\setlength{\belowdisplayskip}{2pt}
\begin{align*}
    \ppi(\s = H | \m = L) = \frac{\q_H c}{\q_L(k - c + \Delta \f_{\max})}, \quad
    \ppi(\s = H | \m = H) = 1
\end{align*}}
and by Lemma \ref{lem:S=2_no_audit_optimum}, the administrator's best response is to never audit. 

\subsection{Proof of Theorem~\ref{thm:coalitions-ExEq}} \label{apdx:pfCoalitions-ExEq}

Let $S = \{L,H\}$ where $L, H$ represent the low and high types respectively with $\f(H) > \f(L)$ and $\Delta \f_{\max} = \f(H) - \f(L)$. Let $\q_L,\q_H$ denote the probabilities that a recipient has type $L, H$ respectively. Further, define $p^* := \min\left\{1,\frac{\q_H c}{\q_L(k - c + \Delta \f_{\max})}\right\}.$

Consider the strategy profile $\ppi^* := \{ \ppi_i^*\}_{i \in [n]}$ where users play the following strategy:
{\setlength{\abovedisplayskip}{2pt}
\setlength{\belowdisplayskip}{2pt}
\begin{align*}
    \ppi_i^*(\s = H | \m = L) = p^*, 
    \quad 
    \ppi_i^*(\s = H | \m = H) = 1.
\end{align*}}
By Lemma \ref{lem:S=2_audit_rule}, the administrator's best response is to never audit. 

We will show that the strategy profile $\ppi^*$ is an equilibrium under $l$-coalitions by showing that no coalition of size at most $l$ can unilaterally deviate and receive pareto-improvements to their utilities. Let $\mathcal{C}$ be any coalition of size at most $l$, and consider any action profile $p = \{p_i\}_{i \in \mathcal{C}}$ where $p_i := \ppi_i(\s = H | \m = L)$. There are two cases to consider.

\textit{Case 1:} If $\max_{i \in \mathcal{C}} p_i \leq p^*$,
then by Lemma \ref{lem:S=2_audit_rule}, the administrator will not audit any of the users in $\mathcal{C}$, but by Lemma \ref{lem:S=2_no_audit_optimum}, none of the users in $\mathcal{C}$ have better utility than they would have if they play according to $\ppi^*$. Hence $p$ does not lead to pareto-improvement for all coalition members.

\textit{Case 2:} If $\max_{i \in \mathcal{C}} p_i > p^*$,
then define the set $\mathcal{C}' := \arg\max_{i \in \mathcal{C}} p_i$. We showed in \hyperlink{audit_rule_max_violator}{step 1} of the proof of Proposition \ref{prop:counter-eg-nonExistence} that by maximizing its utility, it prioritizes auditing users with the largest value of $\ppi_i(\s = H | \m = L)$. Since non-coalition users are playing according to $\ppi^*$ and are therefore never audited by the administrator, the administrator prioritizes its budget to audit the users in $\mathcal{C}'$. Depending on how large the administrator's budget is, there are two things that can happen: 
\begin{itemize}
    \item If $B > c|\mathcal{C'} |$, then $\sigmaa_i(H) = 1$ for all $i \in \mathcal{C}'$, i.e., the administrator can audit every member of $\mathcal{C}'$.
    \item If $B \leq c |\mathcal{C'}|$, then $\sum_{i \in \mathcal{C}'} \sigmaa_i(H) = \frac{B}{c}$, i.e., the administrator does not have enough budget to fully audit all members of $\mathcal{C}'$, but nevertheless it will spend its entire budget to audit $\mathcal{C}'$. In particular, there must exist some $j \in \mathcal{C}'$ for which $\sigmaa_j(H) \geq \frac{B}{c |\mathcal{C}'| }$. 
\end{itemize}
Regardless of whether $B$ is bigger than $c |\mathcal{C'}|$, there exists some $j \in \mathcal{C}'$
for which $p_j = \ppi_j(\s = H | \m = L) > p^*$ 
and $\sigmaa_j(H) \geq \min \{ 1, \frac{B}{c |\mathcal{C}'| } \}$. 

If $B \geq l c \frac{\Delta \f_{\max}}{k + \Delta \f_{\max}} (1-p^*)$ as specified by the conditions of Theorem~\ref{thm:coalitions-ExEq}, then $\sigmaa_j(H) \geq \min\left\{1,\frac{B}{c|\mathcal{C}'|}\right\} \geq \min\left\{1,\frac{l}{|\mathcal{C}'|}
\frac{\Delta \f_{\max}}{k+\Delta \f_{\max}}(1-p^*)\right\} \geq \frac{\Delta \f_{\max}}{k + \Delta \f_{\max}} (1-p^*).$ In other words, user $j$ is violating constraint \eqref{eqn:p2_best_resp}, and the administrator has allocated a budget of at least $c \frac{\Delta \f_{\max}}{k + \Delta \f_{\max}} (1-p^*)$ to audit user $j$. By Lemma~\ref{lem:S=2_expected_utility}, the expected utility of user $j$ under $p$ satisfies
{\setlength{\abovedisplayskip}{0pt}
\setlength{\belowdisplayskip}{0pt}
\begin{align*}
\begin{aligned}
U_j &\leq \q_L\f(L)+\q_H\f(H)
+\q_L p_j p^*\Delta\f_{\max} \leq \q_L\f(L)+\q_H\f(H)
+\q_L p^*\Delta\f_{\max},
\end{aligned}
\end{align*}}
which is the utility user $j$ receives under $\ppi^*$. Thus, user $j$ cannot strictly improve its utility under $p$. If equality holds, the budget condition implies that all $l$ coalition members belong to $\mathcal{C}'$ and receive the threshold audit probability, so none can strictly improve. Therefore $p$ does not provide a pareto-improvement in utility for all members of $\mathcal{C}$. 


Since we showed that no coalition of at most $l$ users can unilaterally deviate from $\ppi^*$ to obtain a pareto-improvement to their utilities, $\ppi^*$ is an equilibrium under $l$-coalitions, which establishes the desired result. 

\subsection{Proof of Theorem~\ref{thm:cost-comparison}} \label{apdx:pfCostCompTwoType}

Let $S = \{L,H\}$ where $L, H$ represent the low and high types respectively with $\f(H) > \f(L)$ and $\Delta \f_{\max} = \f(H) - \f(L)$. Let $\q_L,\q_H$ denote the probabilities that a recipient has type $L, H$ respectively. 

Recall that the cost of a mechanism is comprised of (i) its budget required to run the mechanism and (ii) the excess payments made to the users relative to the setting when all users report their types truthfully. 

First, note that if there is no audit, then a user's optimal action is to always \color{purple} report \color{black} the high type, i.e. $\ppi(\s = H | \m = L) = 1, \ppi(\s = H | \m = H) = 1$. Therefore,
{\setlength{\abovedisplayskip}{2pt}
\setlength{\belowdisplayskip}{2pt}
\begin{align*}
    C^{\text{No Audit}} = \f(H) - (q_L \f(L) + \q_H \f(H)) = \q_L( \f(H) - \f(L) ) = \q_L \Delta \f_{\max}. 
\end{align*}}

For the audit mechanism, note that its budget is $B = \frac{\c \Delta f_{\max}}{\k+\Delta f_{\max}} \left( 1 - \min \left( \frac{\q_H \c}{\q_L (\k - \c + \Delta f_{\max})} , 1\right) \right)$, and by Lemma \ref{lem:S=2_no_audit_optimum}, its excess payments are
{\setlength{\abovedisplayskip}{2pt}
\setlength{\belowdisplayskip}{2pt}
\begin{align*}
   & \q_L \f(L) + \q_H \f(H) + \q_L \min \left(1, \frac{\q_H c}{\q_L(k - c + \Delta \f_{\max})}\right) \Delta \f_{\max} - (\q_L \f(L) + \q_H \f(H)) \\
   &= \q_L \min \left(1, \frac{\q_H c}{\q_L(k - c + \Delta \f_{\max})}\right) \Delta \f_{\max}.
\end{align*}}

Therefore the audit mechanism's cost is
{\setlength{\abovedisplayskip}{2pt}
\setlength{\belowdisplayskip}{2pt}
\begin{align*}
    C^{\text{Audit}} = \frac{\c \Delta f_{\max}}{\k+\Delta f_{\max}} \left( 1 - \min \left( \frac{\q_H \c}{\q_L (\k - \c + \Delta f_{\max})} , 1\right) \right) + \q_L \min \left(1, \frac{\q_H c}{\q_L(k - c + \Delta \f_{\max})}\right) \Delta \f_{\max}.
\end{align*}}
When comparing $C^{\text{Audit}}$ to $C^{\text{No Audit}}$, there are two cases to consider. 

\textit{Case 1:} If $\q_L < \frac{c}{k + \Delta \f_{\max}}$, then $\min \left(1, \frac{\q_H c}{\q_L(k - c + \Delta \f_{\max})}\right) = 1$, and thus $C^{\text{Audit}} = \q_L \Delta \f_{\max} = C^{\text{No Audit}}.$

\textit{Case 2:} If $\q_L \geq \frac{c}{k + \Delta \f_{\max}}$, then $\min \left(1, \frac{\q_H c}{\q_L(k - c + \Delta \f_{\max})}\right) = \frac{\q_H c}{\q_L(k - c + \Delta \f_{\max})}$, and therefore
{\setlength{\abovedisplayskip}{2pt}
\setlength{\belowdisplayskip}{2pt}
\begin{align*}
    C^{\text{Audit}} &= \frac{\c \Delta f_{\max}}{\k+\Delta f_{\max}} \left( 1 - \frac{\q_H \c}{\q_L (\k - \c + \Delta f_{\max})} \right) + \q_L \frac{\q_H c}{\q_L(k - c + \Delta \f_{\max})} \Delta \f_{\max} \\
    &\leq \q_L \Delta f_{\max} \left( 1 - \frac{\q_H \c}{\q_L (\k - \c + \Delta f_{\max})} \right) + \q_L \frac{\q_H c}{\q_L(k - c + \Delta \f_{\max})} \Delta \f_{\max} = \q_L \Delta \f_{\max}.
\end{align*}}
Hence in either case we have shown that $C^{\text{Audit}} \leq C^{\text{No Audit}}$, which establishes the desired result. 

\subsection{Proof of Corollary~\ref{cor:cost-comparison-multiRecipients}} \label{apdx:pfCorCost-comparison-Multi}

To prove this claim, we first note that the total cost of not performing an audit is $C^{\text{No Audit}}_{\n} = \n q_{\min} \Delta f_{\max} = \n C^{\text{No Audit}}$. Next, we consider two cases (as in the proof of Theorem~\ref{thm:cost-comparison}): (i) $q_{\min} \leq \frac{c}{k+\Delta f_{\max}}$ and (ii) $q_{\min} \geq \frac{c}{k+\Delta f_{\max}}$. In the regime when $q_{\min} \leq \frac{c}{k+\Delta f_{\max}}$, no budget is set aside for the audit and thus the audit mechanism reduces to the no audit setting as the budget $B = 0$ in this case. Hence, the total expected cost when $q_{\min}$ is below the above specified threshold is $C^{\text{Audit}}_{(\n, \l)} = \n q_{\min} \Delta f_{\max} = C^{\text{No Audit}}_{\n}$. On the other hand, in the regime when $q_L \geq \frac{c}{k+\Delta f_{\max}}$, the administrator sets aside a budget scaling with the size of the largest coalition $\l$ to conduct audits and thus the total expected cost is given by
{\setlength{\abovedisplayskip}{2pt}
\setlength{\belowdisplayskip}{2pt}
\begin{align*}
    C^{\text{Audit}}_{(n, l)} &\stackrel{(a)}{=} \l \c \frac{(\Delta f_{\max}) \left(1- \frac{(1-q_{\min})\c}{q_{\min}(\k-\c+\Delta f_{\max})} \right)}{\k+\Delta f_{\max}} + \n q_{\min} \pi^*(s_{\max}|s_{\min}) \Delta f_{\max}, \\
    &\stackrel{(b)}{=} \l \c \frac{(\Delta f_{\max}) \left(1- \frac{q_{\max}\c}{q_{\min}(\k-\c+\Delta f_{\max})} \right)}{\k+\Delta f_{\max}} + \n \frac{q_{\max}\c \Delta f_{\max}}{\k-\c+\Delta f_{\max}},
\end{align*}}
%
where (a) follows from the definition of $C^{\text{Audit}}_{(n,\l)}$: the first component is the coalition budget threshold specified in Theorem~\ref{thm:coalitions-ExEq}, equal to $\l$ times the budget threshold in Proposition~\ref{prop:strongerBudgetCondition}, and the second component is the associated excess payments made to low-type users, 
and (b) follows from the fact that $\pi^*(s_{\max}|s_{\min}) = \frac{q_{\max} c}{q_{\min} (\k - \c + \Delta \f_{\max})}$ at the equilibrium. Finally, noting that $C^{\text{Audit}} =  \c \frac{(\Delta f_{\max}) \left(1- \frac{(1-q_{\min})\c}{q_{\min}(\k-\c+\Delta f_{\max})} \right)}{\k+\Delta f_{\max}} + \frac{q_{\max}\c}{\k-\c+\Delta f_{\max}} \Delta f_{\max}$ in the proof of Theorem~\ref{thm:cost-comparison}, $C^{\text{Audit}} \leq C^{\text{No Audit}}$, and $\l \leq \n$, it follows that $C^{\text{Audit}}_{(\n, \l)} \leq \n C^{\text{Audit}} \leq \n C^{\text{No Audit}} = C^{\text{No Audit}}_{\n},$ which establishes our claim.

\subsection{Proof of Theorem~\ref{thm:cost-comparison-multitype}}\label{pf:thm:cost-comparison-multitype}

If $D=0$, then $B=0$ and $\mathbb{E}_{\pi^*}[f(s)] =\max_{s'\in S}f(s')$, implying $C^{\mathrm{Audit}}=C^{\mathrm{No\ Audit}}$. Hence, we consider $D>0$ in the remainder of the proof.

First, note that $C^{\text{Audit}} = B + \mathbb{E}_{\ppi}[f(s) - f(m)]$. Since $B = \frac{c \Delta \f_{\max}}{k + \Delta \f_{\max}}$, $k \geq \Delta f_{\max} \left(\frac{c}{\max_{s'} f(s') - \mathbb{E}_{s \sim \ppi} [f(s)] } - 1 \right)$, which implies $B \leq \frac{c \Delta \f_{\max} }{\Delta f_{\max} \left(\frac{c}{\max_{s'} f(s') - \mathbb{E}_{s \sim \ppi} [f(s)] } - 1 \right) + \Delta \f_{\max}} = \max_{s'} f(s') - \mathbb{E}_{\ppi}[f(s)]$. Hence, $C^{\text{Audit}} \leq \max_{s'} f(s') - \mathbb{E}_{\ppi}[f(s)] + \mathbb{E}_{\ppi}[f(s) - f(m)] = \max_{s'} f(s') - \mathbb{E}_{\ppi}[f(m)]$. 

On the other hand, in the absence of an audit, recipients will always \color{purple} report \color{black} $s' := \arg\max_{s \in \S} f(s)$, leading to $C^{\text{No Audit}} = \max_{s'} f(s') - \mathbb{E}_{\ppi}[f(m)]$, and hence $C^{\text{Audit}} \leq C^{\text{No Audit}}$.

\section{Derivation of Utility Functions from Payoffs}\label{app:utility_derivations}

\subsection{Deriving the Administrator's Utility function}\label{sec:admin_deriv}

The administrator's utility is intended to represent its expected payoffs. Thus in this section we will show that $\mathbb{E}[U^A(a(s), m)] = \mathcal{U}_A(\ppi, \sigmaa)$, where the randomness is taken over the recipient's \color{purple} reporting \color{black} strategy, their type and the administrator's audit strategy. Thus, it is sufficient to show 
{\setlength{\abovedisplayskip}{2pt}
\setlength{\belowdisplayskip}{2pt}
\begin{align}\label{eqn:admin_equal}
    U^A(a(s), m) = - f(\s) + a(\s) \left( -c + \max(f(\s) - f(\m), 0) + \mathbb{I}_{[s \neq m]} k \right),
\end{align}}
since taking expectation on both sides gives the desired result (since the expectation of $a(s)$ is $\sigma(s)$). We proceed by considering 3 different cases. 

\noindent \textbf{Case 1:} If the administrator does not audit, then $a(s) = 0$. Hence, $$- f(\s) + a(\s) \left( -c + \max(f(\s) - f(\m), 0) + \mathbb{I}_{[s \neq m]} k \right) = -f(\s).$$

\noindent \textbf{Case 2:} If the administrator conducts an audit, and the recipient is \color{purple} reporting \color{black} truthfully, then $a(s) = 1$ and $s = m$. Hence, $- f(\s) + a(\s) \left( -c + \max(f(\s) - f(\m), 0) + \mathbb{I}_{[s \neq m]} k \right) = -c-f(\s) = -c-f(\m).$

\noindent \textbf{Case 3:} If the administrator conducts an audit, and the recipient is \color{purple} reporting \color{black} untruthfully, then $a(s) = 1$ and $s \neq m$. If $f(\s) > f(\m)$, we have:
{\setlength{\abovedisplayskip}{2pt}
\setlength{\belowdisplayskip}{2pt}
\begin{align*}
    - f(\s) + a(\s) \left( -c + \max(f(\s) - f(\m), 0) + \mathbb{I}_{[s \neq m]} k \right) \!&=\! - f(\s) + \left( -c + f(\s) - f(\m) + k \right) \!=\! k\!-\!c\!-\!f(\m).
\end{align*}}
If instead $f(\s) < f(\m)$, we have
{\setlength{\abovedisplayskip}{2pt}
\setlength{\belowdisplayskip}{2pt}
\begin{align*}
    - f(\s) + a(\s) \left( -c + \max(f(\s) - f(\m), 0) + \mathbb{I}_{[s \neq m]} k \right) &= - f(\s) + \left( -c +  k \right) = k-c-f(\s).
\end{align*}}
Hence in either situation we have $- f(\s) + a(\s) \left( -c + \max(f(\s) - f(\m), 0) + \mathbb{I}_{[s \neq m]} k \right) = k-c-\min(f(\s),f(\m))$. 

\noindent Comparing the outcomes of the three cases with the definition of $U^A(a(s), m)$ from \eqref{eq:admin-payoff}, we see that \eqref{eqn:admin_equal} holds, and thus $\mathbb{E}[U^A(a(s), m)] = \mathcal{U}_A(\ppi, \sigmaa)$.

\subsection{Deriving the Recipient's Utility function}\label{sec:recipient_deriv}

The recipient's utility is intended to represent its expected payoffs. Thus in this section we will show that $\mathbb{E}[U(a(s), m)] = \mathcal{U}_R(\ppi, \sigmaa)$, where the randomness is taken over the recipient's \color{purple} reporting \color{black} strategy, their type and the administrator's audit strategy. This derivation follows very similarly to that of Section \ref{sec:admin_deriv}. To this end, it is sufficient to show that 
{\setlength{\abovedisplayskip}{2pt}
\setlength{\belowdisplayskip}{2pt}
\begin{align}\label{eqn:recipient_equal}
    U(a(s), m) = f(\s) + a(\s) \left( \min(f(\m) - f(\s), 0) - \mathbb{I}_{[s \neq m]} k \right),
\end{align}}
since taking expectation on both sides gives the desired result (since the expectation of $a(s)$ is $\sigma(s)$). We proceed by considering 3 different cases. 

\noindent \textbf{Case 1:} If the administrator does not audit, then $a(s) = 0$. Hence, $$f(\s) + a(\s) \left( \min(f(\m) - f(\s), 0) - \mathbb{I}_{[s \neq m]} k \right) = f(s).$$

\noindent \textbf{Case 2:} If the administrator conducts an audit, and the recipient is \color{purple} reporting \color{black} truthfully, then $a(s) = 1$ and $s = m$. Hence, $f(\s) + a(\s) \left( \min(f(\m) - f(\s), 0) - \mathbb{I}_{[s \neq m]} k \right) = f(s) = f(m).$

\noindent \textbf{Case 3:} If the administrator conducts an audit, and the recipient is \color{purple} reporting \color{black} untruthfully, then $a(s) = 1$ and $s \neq m$. If $f(\s) > f(\m)$, we have:
{\setlength{\abovedisplayskip}{2pt}
\setlength{\belowdisplayskip}{2pt}
\begin{align*}
    f(\s) + a(\s) \left( \min(f(\m) - f(\s), 0) - \mathbb{I}_{[s \neq m]} k \right) &= f(s) + \left( f(\m) - f(\s) - k \right) = f(\m) - k.
\end{align*}}
If instead $f(\s) < f(\m)$, we have $f(\s) + a(\s) \left( \min(f(\m) - f(\s), 0) - \mathbb{I}_{[s \neq m]} k \right) = f(s) - k.$ Hence, in either situation we have $f(\s) + a(\s) \left( \min(f(\m) - f(\s), 0) - \mathbb{I}_{[s \neq m]} k \right) = \min(f(\s),f(\m)) - k$. 

\noindent Comparing the outcomes of the three cases with the definition of $U(a(s), m)$ from \eqref{eq:user-payoff}, we see that \eqref{eqn:recipient_equal} holds, and thus $\mathbb{E}[U(a(s), m)] = \mathcal{U}_R(\ppi, \sigmaa)$.

\section{Non-equivalence of \color{purple} Stackelberg \color{black} and Bayesian Persuasion Equilibria}\label{app:equilibria_nonequiv}

In this section, we show that there exist Stackelberg equilibria of the audit game that are not Bayesian persuasion equilibria.

Consider a setting with three user types $\S = \{x,y,z\}$, with
$f(z) = f(y) + \Delta f$ and $f(y) = f(x) + \Delta f$ for some
$\Delta f > 0$. Furthermore, let $\q_x = \q_y = \q_z = 1/3$.
As in the main text, we assume $\c > 0$ and $\k \geq \c$, and
we additionally require $\Delta f \geq 2\c$.

Define
\begin{align*}
    t := \frac{\c}{\Delta f + \k - \c}.
\end{align*}
Note that $t \in (0,1/2]$ under our assumptions.

Consider the following reporting strategy profile
$(\ppi_x,\ppi_y,\ppi_z,\sigmaa_{\ppi})$, where
$\ppi_m(s) := \mathbb{P}(s \mid m)$:
\begin{align*}
    \ppi_x(x) &= 1-t(1-t), &
    \ppi_x(y) &= t(1-t), &
    \ppi_x(z) &= 0, \\
    \ppi_y(x) &= 0, &
    \ppi_y(y) &= 1-t, &
    \ppi_y(z) &= t, \\
    \ppi_z(x) &= 0, &
    \ppi_z(y) &= 0, &
    \ppi_z(z) &= 1.
\end{align*}
We will show that $(\ppi_x,\ppi_y,\ppi_z,\sigmaa_0)$ is a
Stackelberg equilibrium, where $\sigmaa_0$ denotes the administrator's
strategy of not auditing any report.

\noindent
\textbf{Administrator optimality.}
By Bayes' rule, we have
\begin{align*}
    \mathbb{P}(m=x \mid s=y)
    &= \frac{t(1-t)}{t(1-t)+(1-t)}
    = \frac{t}{1+t}
    = \frac{\c}{\Delta f+\k}, \\
    \mathbb{P}(m=y \mid s=y)
    &= \frac{1}{1+t}
    = \frac{\Delta f+\k-\c}{\Delta f+\k}, \\
    \mathbb{P}(m=y \mid s=z)
    &= \frac{t}{1+t}
    = \frac{\c}{\Delta f+\k}, \\
    \mathbb{P}(m=z \mid s=z)
    &= \frac{1}{1+t}
    = \frac{\Delta f+\k-\c}{\Delta f+\k}.
\end{align*}
Furthermore, $\mathbb{P}(m=x \mid s=x)=1$, since only type $x$
reports $x$.

Hence, after observing reports $y$ and $z$, the administrator's
expected gross benefit from auditing is, respectively,
\begin{align*}
    &\sum_{m\neq y}\mathbb{P}(m\mid s=y)
    \big( [f(y)-f(m)]_+ + \k \big)
    = \mathbb{P}(m=x\mid s=y)(\Delta f+\k)
    = \c, \\
    &\sum_{m\neq z}\mathbb{P}(m\mid s=z)
    \big( [f(z)-f(m)]_+ + \k \big)
    = \mathbb{P}(m=y\mid s=z)(\Delta f+\k)
    = \c.
\end{align*}
After observing report $x$, the administrator knows that the
user is truthful and therefore strictly prefers not to audit,
since $\c>0$.

Thus, the administrator is indifferent between auditing and not
auditing after reports $y$ and $z$, and strictly prefers not to
audit after report $x$. By the tie-breaking convention in
\eqref{eqn:p2_best_resp}, the administrator's best response is
$\sigmaa_{\ppi}=\sigmaa_0$.

\noindent
\textbf{Optimality of type $z$.}
A type $z$ user maximizes its utility by reporting $z$ truthfully.
Indeed, reporting any lower type yields fewer credits if no audit
occurs and at most the same credits, minus a fine, if an audit
occurs. Thus, $\ppi_z$ is optimal.




\noindent
\textbf{Optimality of type $y$.}
Consider an arbitrary deviation by type $y$, and let
$a:=\ppi_y(x)$ and $r:=\ppi_y(z)$, so that
$\ppi_y(y)=1-a-r$. We first show that any deviation with
$a>0$ is strictly suboptimal. In particular, consider the
alternative strategy $\widetilde{\ppi}_y$ that shifts the
probability $a$ from reporting $x$ to reporting $y$, while
leaving the probability of reporting $z$ unchanged. Under
$\widetilde{\ppi}_y$, type $y$ reports truthfully whenever it
reports $y$, and hence receives $f(y)$ regardless of whether
the administrator audits that report. By contrast, when type
$y$ reports $x$, it receives at most $f(x)$, regardless of
the administrator's audit decision. Moreover, this change
does not affect the probability of observing report $z$ or
the administrator's response to that report. Consequently,
shifting probability from report $x$ to report $y$ increases
type $y$'s expected utility by at least
$a(f(y)-f(x))=a\Delta f>0$.
Thus, any optimal strategy must have $a=0$.

It therefore suffices to consider deviations in which type
$y$ reports $z$ with probability $r\in[0,1]$ and reports
$y$ otherwise, while holding the strategies of the other
types fixed.

Following such a deviation, the administrator's expected
gross benefit from auditing reports $y$ and $z$ is,
respectively,
\begin{align*}
    G_y(r)
    &= \frac{t(1-t)}{t(1-t)+1-r}(\Delta f+\k), \\
    G_z(r)
    &= \frac{r}{1+r}(\Delta f+\k).
\end{align*}
Both expressions are strictly increasing in $r$, and, by
the definition of $t$, satisfy
\[
    G_y(t)=G_z(t)=\c.
\]
Consequently, the administrator strictly prefers not to audit
either report when $r<t$, is indifferent after both reports
when $r=t$, and strictly prefers to audit both reports when
$r>t$. By the tie-breaking convention in
\eqref{eqn:p2_best_resp}, the administrator does not audit
when $r=t$.

Therefore, type $y$'s expected utility from choosing $r$ is
\begin{align*}
    \mathcal{U}^R_y(r,\ppi_{-y},\sigmaa_{\ppi})
    =
    \begin{cases}
        f(y)+r\Delta f, & r\leq t,\\
        f(y)-r\k, & r>t.
    \end{cases}
\end{align*}
The first expression is strictly increasing in $r$, while
the second is strictly below $f(y)$. Hence, type $y$
maximizes its expected utility by choosing $r=t$, so
$\ppi_y$ is optimal.

\noindent
\textbf{Optimality of type $x$.}
Consider an arbitrary deviation by type $x$, and let $p_y$ and
$p_z$ denote the probabilities of reporting $y$ and $z$,
respectively. The remaining probability $1-p_y-p_z$ is assigned
to truthful reporting of $x$.

Given the strategies of types $y$ and $z$, the administrator's
expected gross benefit from auditing reports $y$ and $z$ is,
respectively,
\begin{align*}
    G_y(p_y)
    &= \frac{p_y}{p_y+1-t}(\Delta f+\k), \\
    G_z(p_z)
    &= \frac{p_z(2\Delta f+\k)+t(\Delta f+\k)}
    {p_z+t+1}.
\end{align*}
Using the definition of $t$, we obtain
\begin{align*}
    G_y(p_y)\leq \c
    &\iff p_y\leq t(1-t),\\
    G_z(p_z)\leq \c
    &\iff p_z=0.
\end{align*}
In the second equivalence, we use
$t(\Delta f+\k)=\c(1+t)$ and
$2\Delta f+\k-\c>0$.

Thus, report $y$ is not audited if and only if
$p_y\leq t(1-t)$, while report $z$ is not audited if and
only if $p_z=0$. At equality, the administrator does not audit
by \eqref{eqn:p2_best_resp}.

We now consider four cases:

\begin{itemize}
    \item If $p_z=0$ and $p_y\leq t(1-t)$, neither report is
    audited, and type $x$ obtains
    \[
        f(x)+p_y\Delta f
        \leq f(x)+t(1-t)\Delta f.
    \]

    \item If $p_z=0$ and $p_y>t(1-t)$, report $y$ is audited,
    and type $x$ obtains
    \[
        f(x)-p_y\k < f(x).
    \]

    \item If $p_z>0$ and $p_y\leq t(1-t)$, report $z$ is
    audited but report $y$ is not. Type $x$ obtains
    \[
        f(x)+p_y\Delta f-p_z\k
        \leq f(x)+t(1-t)\Delta f.
    \]

    \item If $p_z>0$ and $p_y>t(1-t)$, both reports are
    audited, and type $x$ obtains
    \[
        f(x)-\k(p_y+p_z)<f(x).
    \]
\end{itemize}

Since $\c>0$, we have $t(1-t)\Delta f>0$. Hence, in every
case, type $x$'s expected utility is no greater than
$f(x)+t(1-t)\Delta f$, which is attained by the proposed
strategy $\ppi_x$. Thus, $\ppi_x$ is optimal.

Combining the administrator's optimality and the optimality of
all three types, we conclude that
$(\ppi_x,\ppi_y,\ppi_z,\sigmaa_0)$ is a Stackelberg equilibrium.

\medskip
\noindent
\textbf{Failure of Bayesian persuasion optimality.}
We next show that the above Stackelberg equilibrium is not a
Bayesian persuasion equilibrium by constructing a reporting
strategy that is feasible for the linear program
$\mathcal{L}_{\c,\k}$ in \eqref{eq:obj-lp-bp}--\eqref{eq:no-audit-con-bp}
and yields strictly higher expected user utility.

Define
\begin{align*}
    u:=\frac{\c}{2\Delta f+\k-\c}.
\end{align*}
Consider the following alternative reporting strategy
$(\ppi'_x,\ppi'_y,\ppi'_z)$:
\begin{align*}
    \ppi'_x(x)&=1-t-u, &
    \ppi'_x(y)&=t, &
    \ppi'_x(z)&=u,\\
    \ppi'_y(x)&=0, &
    \ppi'_y(y)&=1, &
    \ppi'_y(z)&=0,\\
    \ppi'_z(x)&=0, &
    \ppi'_z(y)&=0, &
    \ppi'_z(z)&=1.
\end{align*}
Note that $t\leq 1/2$ and $u\leq 1/4$ under our assumptions,
so these probabilities are well-defined.

Under this strategy, types $y$ and $z$ report truthfully,
while type $x$ may report either $y$ or $z$.

By Bayes' rule,
\begin{align*}
    \mathbb{P}(m=x\mid s=y)
    &=\frac{t}{1+t}
    =\frac{\c}{\Delta f+\k},\\
    \mathbb{P}(m=x\mid s=z)
    &=\frac{u}{1+u}
    =\frac{\c}{2\Delta f+\k}.
\end{align*}
Consequently, the administrator's expected gross benefit from
auditing either report is exactly $\c$. Since report $x$ is
generated only by truthful type $x$ users, the administrator's
best response is again not to audit any report. Thus,
$(\ppi'_x,\ppi'_y,\ppi'_z)$ satisfies the no-audit constraints
\eqref{eq:no-audit-con-bp} and is feasible for
$\mathcal{L}_{\c,\k}$.

We now compare the excess payments under the two reporting
strategies. Under the Stackelberg equilibrium, excess payments
are
\begin{align*}
    \mathcal{E}(\ppi,\sigmaa_0)
    &= \q_x\ppi_x(y)\Delta f
      +\q_y\ppi_y(z)\Delta f\\
    &=\frac{\Delta f}{3}\big(t(1-t)+t\big)\\
    &=\frac{\Delta f}{3}(2t-t^2).
\end{align*}
Under the alternative strategy, excess payments are
\begin{align*}
    \mathcal{E}(\ppi',\sigmaa_0)
    &=\q_x\ppi'_x(y)\Delta f
      +\q_x\ppi'_x(z)2\Delta f\\
    &=\frac{\Delta f}{3}(t+2u).
\end{align*}
The difference is
\begin{align*}
    \mathcal{E}(\ppi',\sigmaa_0)
    -\mathcal{E}(\ppi,\sigmaa_0)
    &=\frac{\Delta f}{3}(2u-t+t^2)\\
    &=\frac{\Delta f}{3}
    \left(
        \frac{\c(\k-\c)}
        {(2\Delta f+\k-\c)(\Delta f+\k-\c)}
        +t^2
    \right)\\
    &>0,
\end{align*}
where the inequality follows from $\k\geq\c$ and $t>0$.

Finally, because both reporting strategies involve only truthful
reports or upward misreports, their expected user utilities under
the no-audit strategy equal the expected entitlement
$\sum_{m\in\S}\q_m f(m)$ plus their respective excess payments.
Hence,
\[
    \mathcal{U}^R(\ppi',\sigmaa_0)
    >
    \mathcal{U}^R(\ppi,\sigmaa_0).
\]
Therefore, the constructed Stackelberg equilibrium reporting strategy
$\ppi$ is not an optimal solution to the linear program
$\mathcal{L}_{\c,\k}$. By the characterization of Bayesian
persuasion equilibria in Section~4, it follows that
$(\ppi_x,\ppi_y,\ppi_z,\sigmaa_0)$ is not a Bayesian persuasion
equilibrium.

\section{Additional Theoretical Analyses and Results}


\subsection{Motivation for Payoffs that Discourage Under-reporting} \label{apdx:collusion-model-asmpn}

In modeling payoffs in Section~\ref{sec:model:strategies_payoffs}, we assume the administrator pays out only $f(s)$ if the user asks for too little, i.e., $f(s) < f(m)$, thereby discouraging under-reporting, which prevents collusion across user types. To see this, suppose there are types T1, T2, T3, where T1 (T3) is entitled to the least (most) aid. If T3 reports their type as T2 (and under-reporting is not fraud), this reduces P(true = T1 $|$ reported = T2), making it harder to catch T1 users reporting as T2. If T3 users receive T3 level aid even when reporting T2, this allows them to help T1 users commit fraud with no penalty. By penalizing under-reporting, we provably avoid these cross-type collusions.


\subsection{Uniqueness of Equilibria for $|\S| = 2$} \label{apdx:equilibriumUniqueness}

Let $S = \{L,H\}$ where $L, H$ represent the low and high types respectively with $\f(H) > \f(L)$ and $\Delta \f_{\max} = \f(H) - \f(L)$. Let $\q_L,\q_H$ denote the probabilities that a recipient has type $L, H$ respectively. 

We prove this result by using Lemma~\ref{lem:S=2_expected_utility} and Lemma~\ref{lem:S=2_no_audit_optimum}. Consider the user strategy in Lemma~\ref{lem:S=2_no_audit_optimum}:
{\setlength{\abovedisplayskip}{2pt}
\setlength{\belowdisplayskip}{2pt}
\begin{align*}
    \ppi'(\s = H | \m = L) := \min \left(1, \frac{\q_H c}{\q_L(k - c + \Delta \f_{\max})}\right), \quad
    \ppi'(\s = H | \m = H) := 1.
\end{align*}}

Let $\Pi$ be the set of no-audit strategies for the user, i.e., user strategies for which the administrator's best response is to never audit, meaning $\sigmaa(H) = \sigmaa(L) = 0$

By construction of $\ppi'$, the administrator's best response to $\ppi'$ is $\sigmaa(H) = \sigmaa(L) = 0$, and the administrator cannot improve its utility by unilaterally changing its strategy. 

On the other hand, choosing any other strategy will lead to a strictly lower utility for the user. If $\ppi(\s = H | \m = L) < \min \left(1, \frac{\q_H c}{\q_L(k - c + \Delta \f_{\max})}\right)$, then by Lemma~\ref{lem:S=2_no_audit_optimum}, a unilateral switch from $\ppi$ to $\ppi'$ strictly increases the user's utility. On the other hand if $\ppi(\s = H | \m = L) > \min \left(1, \frac{\q_H c}{\q_L(k - c + \Delta \f_{\max})}\right)$, then $\sigmaa(H) = 1$, and by Lemma~\ref{lem:S=2_expected_utility}, changing $\ppi(\s = H | \m = L)$ to $0$ strictly increases the user's utility. 

For the strategy profile $(\ppi', \sigmaa \equiv 0)$, neither player can improve their utility via unilateral deviation, thus this strategy profile is an equilibrium. Furthermore, for any other user strategy $\ppi$, the user can improve its utility via unilateral deviation. Therefore $(\ppi', \sigmaa \equiv 0)$ is the unique equilibrium. 

\subsection{Existence of Equilibria when $|\S| = 2$ Regardless of Budget $B$} \label{apdx:existence-regardless-of-budget}

Let $S = \{L,H\}$ where $L, H$ represent the low and high types respectively with $\f(H) > \f(L)$ and $\Delta \f_{\max} = \f(H) - \f(L)$. Let $\q_L,\q_H$ denote the probabilities that a recipient has type $L, H$ respectively. 

Using Lemma \ref{lem:S=2_expected_utility} from Appendix \ref{apdx:S2_core_ideas}, the expected utility of any user action $\ppi$ is given by
{\setlength{\abovedisplayskip}{2pt}
\setlength{\belowdisplayskip}{2pt}
\begin{align*}
    \mathbb{E}_{\m \sim \q, \s \sim \ppi(\cdot | \m)} \left[U^i(\sigma(\s), \m)\right] = \q_L \f(L) + \q_H \f(H) + \q_L \ppi(\s = H | \m = L) (\Delta \f_{\max} - \sigmaa(H) (k + \Delta \f_{\max}) )
\end{align*}}

Note that for a fixed value of $\sigmaa^*(H)$, the expected utility is an affine function of $\ppi(\s = H | \m = L)$. Furthermore, according to Lemma \ref{lem:S=2_audit_rule}, 
{\setlength{\abovedisplayskip}{2pt}
\setlength{\belowdisplayskip}{2pt}
\begin{align*}
    \sigmaa(H) &= \left \{ 
    \begin{tabular}{cc}
        $0$ & if $\ppi(\s = H | \m = L) \leq \frac{\q_H c}{\q_L(k - c + \Delta \f_{\max})}$ \\
        $\min \left(1, \frac{B}{c}\right)$ & otherwise.
    \end{tabular}
    \right.
\end{align*}}

Hence the optimal value of $\mathbb{E}_{\m \sim \q, \s \sim \ppi(\cdot | \m)} \left[U^i(\sigma(\s), \m)\right]$ must be attained at $ \ppi(\s = H | \m = L) \in \{ 0, \frac{\q_H c}{\q_L(k - c + \Delta \f_{\max})}, 1 \}$.

We discuss at the beginning of Appendix \ref{apdx:S2_core_ideas} why $\ppi(\s = H | \m = H) = 1$ is a dominant strategy for the user. As a consequence, $\sigmaa(L) = 0$ will always hold at an equilibrium. 

If $\frac{\q_H c}{\q_L(k-c+\Delta \f_{\max})} \geq 1$, then every feasible misreporting probability satisfies the no-audit condition in Lemma~\ref{lem:S=2_audit_rule}. Consequently, the administrator's best response is to never audit, including when it is indifferent between auditing and not auditing, by our equilibrium selection convention. Since choosing $\s = H$ is a dominant strategy for high-type recipients, it remains only to consider the reporting probability $\ppi_i(\s = H \mid \m = L)$ of low-type recipients. Given that the administrator never audits, a recipient's expected utility is strictly increasing in this probability. Hence, the recipient playing $\ppi_i(\s = H \mid \m = L) = 1$ and the administrator never auditing is an equilibrium.

All that remains is to handle the case where $\frac{\q_H c}{\q_L(k-c+\Delta \f_{\max})} < 1$. As for $\ppi(\s = H | \m = L)$ and $\sigmaa(H)$, there are two cases to consider:

If $B \leq \c \frac{\Delta \f_{\max} \big( 1 - \min \{ \frac{\q_H \c}{\q_L (\k - \c + \Delta \f_{\max})} , 1\} \big)}{k+\Delta \f_{\max}}$, then $\ppi(\s = H | \m = L) = 1$ is optimal, and thus it and its best response $\sigmaa(H) = B/c$ is an equilibrium.

If $B \geq \c \frac{\Delta \f_{\max} \big( 1 - \min \{ \frac{\q_H \c}{\q_L (\k - \c + \Delta \f_{\max})} , 1\} \big)}{k+\Delta \f_{\max}}$, then $\ppi(\s = H | \m = L) = \frac{\q_H c}{\q_L(k - c + \Delta \f_{\max})}$ is optimal, and thus it and its best response $\sigmaa(H) = 0$ is an equilibrium.

\subsection{Generalization of Theorem~\ref{thm:cost-comparison-multitype} to Multiple-User Setting} \label{apdx:generalizationThm6}


Let $\q_1, \q_2, ..., \q_n$ be the marginal distribution of the $n$ users' types, and let $\ppi_1,...,\ppi_n$ denote equilibrium actions of the $n$ players, i.e., the solutions to \eqref{eq:obj-lp-bp}-\eqref{eq:no-audit-con-bp}. We use $(\s_i,\m_i) \sim \ppi_i$ to represent sampling $\s_i,\m_i$ from their joint distribution by first sampling $\m_i$ from $\q_i$, and then sampling $\s_i$ from $\ppi_i(\cdot | \m_i)$. In such a setting, if $k \geq \Delta f_{\max} \left(\frac{c}{ \sum_{i=1}^n \max_{s'} f(s') - \mathbb{E}_{(s,m) \sim \ppi_i} [f(s)] } - 1 \right),$ then by the same reasoning as Theorem~\ref{thm:cost-comparison-multitype}, we have $B \leq \sum_{i=1}^n \max_{s'} f(s') - \mathbb{E}_{(s,m) \sim \ppi_i} [f(s)],$ meaning that
{\setlength{\abovedisplayskip}{2pt}
\setlength{\belowdisplayskip}{2pt}
\begin{align*}
    C^{\text{Audit}} 
    &\leq \sum_{i=1}^n \max_{s'} f(s') - \mathbb{E}_{(s,m) \sim \ppi_i} [f(s)] + \sum_{i=1}^n \mathbb{E}_{(s,m) \sim \ppi_i} [f(s)] - \mathbb{E}_{\m \sim \q_i}[\f(\m)] \\
    &= \sum_{i=1}^n \max_{s'} f(s') - \mathbb{E}_{\m \sim \q_i}[\f(\m)] = C^{\text{No Audit}}.
\end{align*}}

\subsection{User utility in the unbudgeted two-type audit game}\label{apdx:two_type_fig}

\vspace{-10pt}

\begin{figure}[tbh!]
    \centering
    \includegraphics[width=0.75\textwidth]{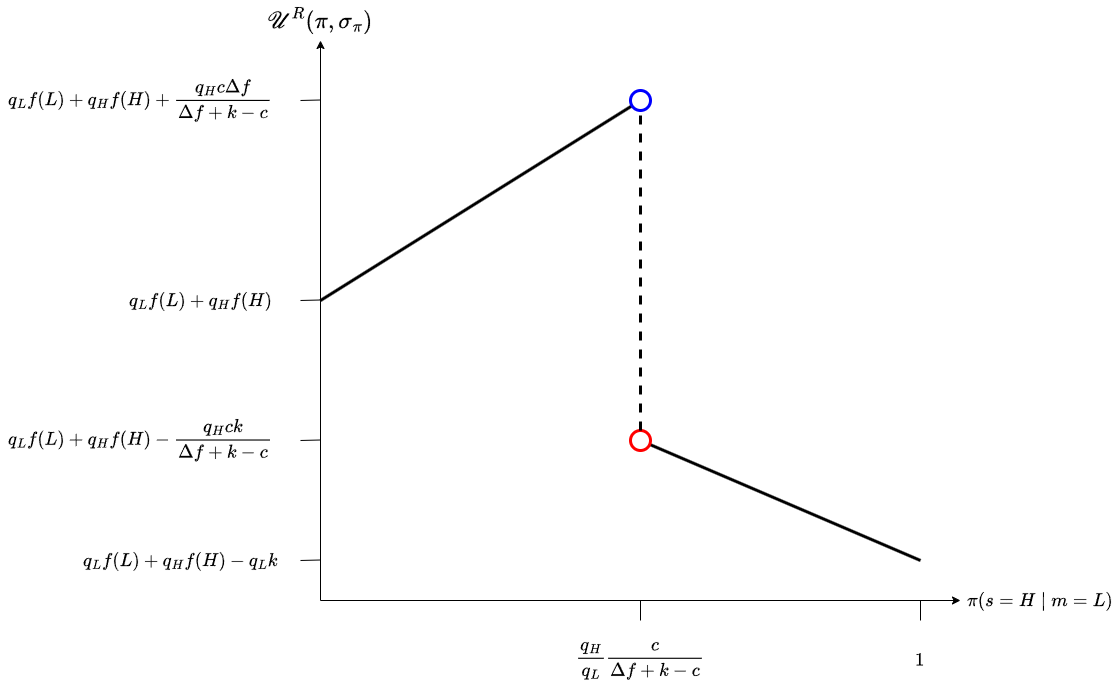}
    \caption{\footnotesize \sf User utility $\mathcal{U}^R$ in a two-type setting $\S = \bigbrace{L,H}$ with $f(H) - f(L) = \Delta f > 0$ and $\frac{q_H c}{q_L(\Delta f+k-c)}\in(0,1)$. The plot is made assuming that (a) users of type $H$ will always \color{purple} report \color{black} truthfully, since this is a dominant strategy and (b) the administrator is always playing a best response according to \eqref{eqn:p2_best_resp}. Under these assumptions, the user strategy is entirely determined by the misreporting probability $\ppi(s = H | m = L)$. We analyze the plot by looking at three different regimes: regime 1 is  $[0, \frac{\q_H}{\q_L} \frac{c}{\Delta f + k - c})$, regime 2 is $(\frac{\q_H}{\q_L} \frac{c}{\Delta f + k - c}, 1]$ and regime 3 is when $\ppi(s = H | m = L) = \frac{\q_H}{\q_L} \frac{c}{\Delta f + k - c}$. Regime 1 corresponds to the case where $\sigmaa(H) = 0$ is the administrator's unique best response. In the absence of an audit, $\mathcal{U}^R$ is an increasing function of the misreporting probability. Regime 2 corresponds to the case where $\sigmaa(H) = 1$ is the administrator's unique best response. In this regime, $\mathcal{U}^R$ is decreasing, since a higher misreporting probability increases the probability of getting caught and paying the fine. In regime 3 the administrator is indifferent between auditing and not auditing, and hence any $\sigmaa(H) \in [0,1]$ is a best response. In terms of the user utility, $\mathcal{U}^R$ interpolates from the blue point to the red point along the dashed line as $\sigmaa(H)$ varies from $0$ to $1$. Importantly, this figure shows that $\mathcal{U}^R$ has a maximum if and only if $\sigmaa(H)$ is chosen to be $0$ in regime 3. As a consequence, this says that the game has an equilibrium if and only if the administrator plays $\sigmaa(H) = 0$ in regime 3. }
    \label{fig:ag_two_players}
\end{figure}

\vspace{-20pt}

\color{purple}

\subsection{Signaling Equilibria and the LP Bound}
\label{app:classical_signaling_bound}

In this section, we consider the static, budget-unconstrained audit game under the conventional signaling game equilibrium notion. The difference from the Bayesian Stackelberg equilibrium notion in Section~\ref{sec:model:interim_equilibria} is how deviations are evaluated: in a classical signaling game, a recipient considers changing its report while holding the administrator's report-contingent auditing strategy fixed, whereas in the Stackelberg formulation the administrator best responds to the resulting change in the reporting strategy profile. We establish that, despite this difference, the linear program $\mathcal L_{c,k}$ still bounds fraud at every classical signaling equilibrium. Throughout, we retain the notation and primitive payoffs of Section~\ref{sec:mainModel}, assume $S$ is finite, $c>0$, and $k\ge c$.

\begin{definition}[Classical signaling equilibrium (perfect Bayesian equilibrium)]
\label{def:app_classical_pbe}
An assessment $(\boldsymbol\pi^*,\boldsymbol\sigma^*,\mu^*)$ is a \emph{perfect Bayesian equilibrium} (PBE) if the following hold:
\begin{enumerate}[label=\textup{(\roman*)},leftmargin=*,itemsep=2pt]
\item \textbf{Recipient optimality.} For every type $m\in S$ with $q_m>0$, each report sent with positive probability maximizes the recipient's expected payoff given the \emph{fixed} auditing strategy $\boldsymbol\sigma^*$:
\begin{equation}
 \pi^*(s\mid m)>0 \;\Longrightarrow\;
 s\in\underset{s'\in S}{\arg\max}\;
 \Bigl\{f(s')-\sigma^*(s')\mathbf 1_{\{s'\ne m\}}
 \bigl(k+[f(s')-f(m)]_+\bigr)\Bigr\}.
 \label{eq:app_pbe_sender}
\end{equation}
\item \textbf{Administrator optimality.} After each report $s\in S$, the administrator chooses a best-response audit probability given its belief $\mu^*(m\mid s)$ about the true type:
\begin{equation}
 \sigma^*(s)\in\underset{\alpha\in[0,1]}{\arg\max}\;
 \alpha\left(\sum_{m\in S\setminus\{s\}}\mu^*(m\mid s)
 \bigl(k+[f(s)-f(m)]_+\bigr)-c\right).
 \label{eq:app_pbe_admin}
\end{equation}
\item \textbf{Belief consistency.} If report $s$ occurs with positive probability, beliefs are determined by Bayes' rule:
\begin{equation}
 \mu^*(m\mid s)=\frac{q_m\pi^*(s\mid m)}{\sum_{m'\in S}q_{m'}\pi^*(s\mid m')},\qquad m\in S.
 \label{eq:app_pbe_bayes}
\end{equation}
At a report that is never sent, $\mu^*(\cdot\mid s)$ is a probability distribution over possible types for which $\sigma^*(s)$ satisfies \eqref{eq:app_pbe_admin}.
\end{enumerate}
\end{definition}

Since our audit game has finitely many user types, reports, and audit decisions, the existence of a classical signaling equilibrium follows from standard results for finite extensive-form games~\citep{krepsWilson1982}. We next show that the reporting strategy at any such equilibrium satisfies the constraints of our LP $\mathcal{L}_{c,k}$. Consequently, our existing bounds on misreporting probabilities and excess payments continue to hold under classical signaling equilibria.


\begin{proposition}[LP bound for signaling equilibria]
\label{prop:app_pbe_bound}
Let $(\boldsymbol\pi^*,\boldsymbol\sigma^*,\mu^*)$ be any PBE of the budget-unconstrained game with the payoffs in ~\eqref{eq:admin-payoff}-\eqref{eq:user-payoff}. Then $\pi^*(s\mid m)=0$ whenever $f(s)<f(m)$, and $\boldsymbol\pi^*$ is feasible for the \emph{same} $\mathcal L_{c,k}$. Consequently, the misreporting-probability bound of Corollary~\ref{cor:misrepotingProbBound} (where its denominator is positive) also applies to $\boldsymbol\pi^*$, and
\begin{equation}
 \boxed{\quad
 \mathcal E(\boldsymbol\pi^*,\boldsymbol\sigma^*)
 \le \operatorname{OPT}(\mathcal L_{c,k})-\sum_{m\in S}q_m f(m)
 \le \frac{c\,\Delta f_{\max}}{k+\Delta f_{\max}}.
 \quad}
 \label{eq:app_bound}
\end{equation}
In particular, $\mathcal L_{c,k}$ bounds misreporting fraud across \emph{all} classical signaling equilibria, although its optimizer need not itself be such an equilibrium.
\end{proposition}

\begin{proof}
Truthful reporting gives a type-$m$ recipient $f(m)$ regardless of auditing. A report $s$ with $f(s)<f(m)$ gives at most $f(s)$, and thus cannot be optimal.

Fix a report $s$ that occurs with positive probability. Suppose $\sigma^*(s)=1$. Any recipient of type $m\ne s$ would then receive $\min\{f(s),f(m)\}-k<f(m)$ from reporting $s$, so no such recipient would send it. Conditional on receiving $s$, the administrator would consequently know that the report was truthful. Auditing would then incur the strictly positive cost $c$ without any benefit, contradicting administrator optimality. Hence $\sigma^*(s)<1$ at every report that occurs with positive probability.

It follows that the administrator's conditional expected gain from auditing any such report cannot be strictly positive (which would require auditing with probability one). By \eqref{eq:app_pbe_admin},
\begin{equation}
 \sum_{m\in S\setminus\{s\}}\mu^*(m\mid s)
 \bigl(k+[f(s)-f(m)]_+\bigr)-c\le0.
 \label{eq:app_nonpositive_gain}
\end{equation}
Multiplying \eqref{eq:app_nonpositive_gain} by the positive probability $\sum_m q_m\pi^*(s\mid m)$ of observing $s$ and using \eqref{eq:app_pbe_bayes} gives exactly the report-$s$ constraint \eqref{eq:no-audit-con-bp}. If $s$ is never reported, both sides of that constraint are zero. Thus $\boldsymbol\pi^*$ is feasible for $\mathcal L_{c,k}$. Notice that this argument allows $\sigma^*(s)>0$: auditing can occur when the administrator is \emph{indifferent}, rather than strictly preferring to audit.

We next establish the excess payment bound. By Definition~\ref{def:misreporting-fraud} and the primitive recipient
payoffs in \eqref{eq:user-payoff}, and using the fact that underreporting does not occur at a PBE, we have
\begin{equation}
\begin{aligned}
 \mathcal E(\boldsymbol\pi^*,\boldsymbol\sigma^*)
 &=\mathbb E_{\substack{m\sim\boldsymbol q,\ s\sim\boldsymbol\pi^*(\cdot\mid m)\\
                      a(s)\sim\operatorname{Bern}(\sigma^*(s))}}
  \left[\max\big\{U^R(a(s),m)-f(m),0\big\}\right]\\
 &=\sum_{m\in S}q_m\sum_{s\in S}\pi^*(s\mid m)
  (1-\sigma^*(s))\,[f(s)-f(m)]_+\\
 &\le \sum_{m\in S}q_m\sum_{s\in S}\pi^*(s\mid m)
  \bigl(f(s)-f(m)\bigr)\\
 &= \sum_{m\in S}q_m\sum_{s\in S}\pi^*(s\mid m)f(s)
       -\sum_{m\in S}q_m f(m)\\
 &\le \operatorname{OPT}(\mathcal L_{c,k})-\sum_{m\in S}q_m f(m)\\
 &\le \frac{c\,\Delta f_{\max}}{k+\Delta f_{\max}}.
\end{aligned}
\label{eq:app_full_fraud_chain}
\end{equation}
The first inequality uses $1-\sigma^*(s)\le1$ and the absence of
underreporting; the second uses the LP feasibility established above.
The \emph{last} inequality follows immediately from the analysis in Appendix~\ref{apdx:corPf-equilibrium_suboptimality} due to the two inequalities~\eqref{eqn:equilibrium_suboptimality_ineq1}-\eqref{eqn:equilibrium_suboptimality_ineq2}. 

The misreporting probability bound follows directly from Appendix~\ref{apdx:corPf-misreportingBound} by isolating, for each true type $m$ and higher benefit report $s$, the term involving $q_m\pi^*(s\mid m)$ in constraint \eqref{eq:no-audit-con-bp}, noting that all other terms are nonnegative.
\end{proof}

The above proposition implies that classical signaling equilibria may involve positive auditing probabilities, whereas the LP solutions studied in Section~\ref{sec:misreportingSingle} induce no auditing under the Bayesian Stackelberg formulation. The LP continues to be informative because its upper bound holds for all classical signaling game equilibria and does not require selecting a particular signaling equilibrium.

\color{black}

\color{green}

\subsection{Extension of Static Audit Game to Setting with Stochastic Audit Costs}
\label{apdx:uniform-random-costs}

This section describes a stochastic cost variant of the static audit game. Let the audit cost be a random variable $C\sim \mathrm{Unif}[c_1,c_2]$, where $0 \leq c_1<c_2$ with $k \geq c_2$. The administrator observes the realized cost $C$, while users know only its distribution. Let $d:=c_2-c_1$ and define
{\setlength{\abovedisplayskip}{2pt}
\setlength{\belowdisplayskip}{2pt}
\begin{align}
    A(g):=\Pr(C<g)=
\begin{cases}
0, & g\le c_1,\\[1mm]
\dfrac{g-c_1}{d}, & c_1<g<c_2,\\[1mm]
1, & g\ge c_2.
\end{cases}
\end{align}}
Order the type set $\mathcal S=\{1,\ldots,J\}$ so that $f_1<\cdots<f_J$. Let $y_{ms}:=q_m\pi(s\mid m),$ so that $y_{ms}$ is the joint probability of true type $m$ and report $s$. As in the deterministic cost model, we show that under-reporting is dominated.
\begin{lemma}[No Underreporting with Stochastic Audit Costs]\label{lem:stochastic-no-underreporting}
Suppose $k\geq c_2$. Then any optimal solution to the sender's stochastic-cost Bayesian persuasion problem assigns zero probability to underreporting; that is, $y_{ms}=0$ whenever $s<m$.
\end{lemma}
\noindent The proof of Lemma~\ref{lem:stochastic-no-underreporting} is provided in Appendix~\ref{apdx:pf-stochastic-no-underreporting}. Thus we restrict attention to strategies satisfying $\pi(s\mid m)=0$ whenever $s<m$. We have 
{\setlength{\abovedisplayskip}{2pt}
\setlength{\belowdisplayskip}{2pt}
\begin{align}
\sum_{s\in\mathcal S}y_{ms}=q_m,\qquad y_{ms}\ge0,\qquad y_{ms}=0\quad\text{if }s<m.
\end{align}}
For each report $s$, define
{\setlength{\abovedisplayskip}{2pt}
\setlength{\belowdisplayskip}{2pt}
\begin{align}
P_s(y):=\sum_m y_{ms},\qquad
H_s(y):=\sum_{m<s}y_{ms}(k+f_s-f_m),\qquad
G_s(y):=\frac{H_s(y)}{P_s(y)},
\end{align}}
where the ratio is taken to be zero when $P_s(y)=0$. Here $G_s(y)$ is the administrator's posterior gross gain from auditing report $s$. Hence the administrator audits report $s$ if and only if $C<G_s(y)$, and users perceive audit probability $A(G_s(y))$. In particular, reports with $c_1<G_s(y)<c_2$ are audited with positive probability $(G_s(y)-c_1)/d$.

Any sender-preferred Bayesian-persuasion equilibrium must satisfy $H_s(y)\le c_2P_s(y),$ for all $s\in\mathcal S$ because if $G_s(y)\ge c_2$ and report $s$ receives positive misreporting mass, then report $s$ is audited with probability one and the relevant misreporting types strictly prefer to move that mass back to truthful reports. Therefore, for any equilibrium candidate, the expected audit penalty associated with report $s$ is $H_s(y)A(G_s(y))=\frac{1}{d}\left(\frac{H_s(y)^2}{P_s(y)}-c_1H_s(y)\right)_+.$ Introducing auxiliary variables $z_s$ gives the following second-order cone program:
{\setlength{\abovedisplayskip}{2pt}
\setlength{\belowdisplayskip}{2pt}
\begin{align}
\max_{y,z}\quad
&\sum_{m,s}y_{ms}f_s-\frac{1}{d}\sum_s z_s
\tag{SOCP}\label{eq:uniform-socp}\\
\text{s.t.}\quad
&\sum_{s\in\mathcal S}y_{ms}=q_m, && \forall m\in\mathcal S,\nonumber\\
&y_{ms}\ge0, && \forall m,s\in\mathcal S,\nonumber\\
&y_{ms}=0, && \forall (m,s)\text{ such that }s<m,\nonumber\\
&H_s(y)\le c_2P_s(y), && \forall s\in\mathcal S,\nonumber\\
&z_s\ge0, && \forall s\in\mathcal S,\nonumber\\
&\left\|
\begin{pmatrix}
2H_s(y)\\
z_s+c_1H_s(y)-P_s(y)
\end{pmatrix}
\right\|_2
\le z_s+c_1H_s(y)+P_s(y), && \forall s\in\mathcal S.\nonumber
\end{align}}
The cone constraint is equivalent to $H_s(y)^2\le \big(z_s+c_1H_s(y)\big)P_s(y),$ and hence represents the constraint $z_s\ge \frac{H_s(y)^2}{P_s(y)}-c_1H_s(y),$ where $z_s\ge0.$ Because the objective subtracts $(1/d)\sum_s z_s$, the program sets $z_s=\left(\frac{H_s(y)^2}{P_s(y)}-c_1H_s(y)\right)_+$ at optimum. Thus, \eqref{eq:uniform-socp} computes the sender-preferred Bayesian persuasion benchmark for the stochastic cost audit game. Since every \color{purple} Stackelberg \color{black} equilibrium induces a feasible joint distribution $y$, the SOCP value upper bounds the aggregate user payoff of any \color{purple} Stackelberg \color{black} equilibrium in this stochastic cost model.

It remains to verify when an optimizer of \eqref{eq:uniform-socp} is itself a \color{purple} Stackelberg \color{black} equilibrium. Let $(y^*,z^*)$ solve \eqref{eq:uniform-socp}, set $\pi^*(s\mid m)=y^*_{ms}/q_m$, and let the administrator use the best response $\sigma^*(s,C)=\mathbf 1\left\{C<\frac{H_s(y^*)}{P_s(y^*)}\right\},$ with the same zero-ratio convention.
Define $K_{ms}:=\mathbf 1\{s>m\}(k+f_s-f_m)$ and $U_m^*:=\sum_{s\ge m}\frac{y^*_{ms}}{q_m}\left(f_s-K_{ms}A(G_s(y^*))\right).$ For a unilateral row deviation by type $m$, let $x=(x_s)_{s\ge m}$ with $x_s\ge0$ and $\sum_{s\ge m}x_s=1$, and define $P_s^m(x):=P_s(y^*)-y^*_{ms}+q_mx_s$ and $H_s^m(x):=H_s(y^*)-y^*_{ms}K_{ms}+q_mx_sK_{ms},$ with $G_s^m(x):=H_s^m(x)/P_s^m(x)$, taking the ratio to be zero when the denominator is zero. Then $(\pi^*,\sigma^*)$ is a \color{purple} Stackelberg \color{black} equilibrium if and only if, for every $m\in\mathcal S$,
\begin{equation}
U_m^*\ge
\max_{x_s\ge0,\ \sum_{s\ge m}x_s=1}
\sum_{s\ge m}x_s\left(f_s-K_{ms}A(G_s^m(x))\right).
\tag{SG}
\label{eq:uniform-sg-check}
\end{equation}
Indeed, the right-hand side is exactly type $m$'s best payoff from changing only its own reporting row while the administrator best responds to the induced posteriors. Hence the inequalities in \eqref{eq:uniform-sg-check} are necessary and sufficient for the SOCP optimizer to induce a \color{purple} Stackelberg \color{black} equilibrium.

In the two-type case, this verification is automatic: the only nontrivial strategic decision is the low type's probability of reporting the high type, while the high type reports truthfully. Thus, for two types, the stochastic cost program directly characterizes both the sender-preferred Bayesian persuasion benchmark and the \color{purple} Stackelberg \color{black} equilibrium. With more than two types, the SOCP remains a tractable benchmark and coincides with the \color{purple} Stackelberg \color{black} equilibrium exactly when the verification condition \eqref{eq:uniform-sg-check} holds.
\color{black}

\subsection{Proof of Lemma~\ref{lem:stochastic-no-underreporting}}\label{apdx:pf-stochastic-no-underreporting}

For the purposes of this proof, we temporarily allow
underreporting. For any reporting strategy $y$, define
\[
P_s(y):=\sum_{m\in\mathcal S}y_{ms},\qquad
H_s(y):=\sum_{m\in\mathcal S\setminus\{s\}}
y_{ms}\bigl(k+[f_s-f_m]_+\bigr),
\]
and
\[
G_s(y):=\frac{H_s(y)}{P_s(y)},
\]
where the ratio is defined to be zero when $P_s(y)=0$.
The sender's expected utility is
\[
V(y):=
\sum_{m,s\in\mathcal S}y_{ms}f_s
-\sum_{s\in\mathcal S}H_s(y)A(G_s(y)).
\]

We prove the result by contradiction.

Suppose, for contradiction, that a reporting strategy $y$
assigns positive probability to underreporting, i.e.,
$y_{ms}>0$ for some $s<m$. Choose
$0<\varepsilon\leq y_{ms}$, and construct a new reporting
strategy $\widetilde y$ by transferring an amount
$\varepsilon$ of this probability from report $s$ to the
truthful report $m$:
\[
\widetilde y_{ms}=y_{ms}-\varepsilon,\qquad
\widetilde y_{mm}=y_{mm}+\varepsilon,
\]
with all other entries unchanged. This is a feasible
reporting strategy.

We first compare the administrator's expected audit
penalties before and after this transfer. Recall that the
expected audit penalty associated with report $r$ is
$H_r(y)A(G_r(y))$.

At report $s$, since $s<m$ and hence $f_s<f_m$, the
transfer removes an amount $k\varepsilon$ from $H_s$
and an amount $\varepsilon$ from $P_s$. Thus,
\[
P_s(\widetilde y)=P_s(y)-\varepsilon,\qquad
H_s(\widetilde y)=H_s(y)-k\varepsilon.
\]
If $P_s(\widetilde y)>0$, it follows that
\[
G_s(\widetilde y)-G_s(y)
=
\frac{\varepsilon\bigl(G_s(y)-k\bigr)}
     {P_s(y)-\varepsilon}.
\]
We consider two cases.

If $G_s(y)\leq c_2$, then $k\geq c_2\geq G_s(y)$,
so $G_s(\widetilde y)\leq G_s(y)$. Since $A$ is
nondecreasing and $H_s(\widetilde y)\leq H_s(y)$,
we obtain
\[
H_s(\widetilde y)A(G_s(\widetilde y))
\leq H_s(y)A(G_s(y)).
\]
If instead $G_s(y)>c_2$, then $A(G_s(y))=1$.
Using $A(\cdot)\leq 1$ and
$H_s(\widetilde y)\leq H_s(y)$, we again obtain
\[
H_s(\widetilde y)A(G_s(\widetilde y))
\leq H_s(y)A(G_s(y)).
\]
The same inequality holds trivially if
$P_s(\widetilde y)=0$, since in that case
$H_s(\widetilde y)=0$.

At report $m$, the transfer adds truthful reporting
probability but does not change the total incoming
misreporting mass. Consequently,
\[
P_m(\widetilde y)=P_m(y)+\varepsilon,\qquad
H_m(\widetilde y)=H_m(y).
\]
If $P_m(y)>0$, then
$G_m(\widetilde y)\leq G_m(y)$, and hence
\[
H_m(\widetilde y)A(G_m(\widetilde y))
\leq H_m(y)A(G_m(y)).
\]
If $P_m(y)=0$, then $H_m(y)=0$, and the same
inequality holds. All other reports are unaffected by
the transfer.

Thus, the total expected audit penalty weakly decreases.
On the other hand, the expected reported benefit increases
by
\[
\varepsilon(f_m-f_s)>0.
\]
It follows that
\[
V(\widetilde y)>V(y),
\]
contradicting the optimality of $y$. Therefore, every
optimal reporting strategy assigns zero probability to
underreporting, completing the proof.

\section{Additional Numerical Results}

In this section, we present additional numerical results beyond those presented in Section~\ref{sec:experiment-setup}. 
First, in Appendix~\ref{apdx:varyFinePlots}, we present additional results on the total cost of our mechanism when we fix the audit cost $\c$ and vary the audit fine $\k$. Next, in Appendix~\ref{apdx:misreportingPlots}, we show how the maximum misreporting probability varies as we adjust the audit fine $\k$ and audit cost $\c$ (see Corollary~\ref{cor:misrepotingProbBound}).

\subsection{Social Cost Comparisons under Varying Audit Fines} \label{apdx:varyFinePlots}

For our calibrated model parameters in Section~\ref{sec:experiment-setup}, Figure~\ref{fig:GAOFraudk} depicts the total costs of our audit mechanism and that of the setting without audits (termed ``No Audit'') as a function of the probability of the low type $q_{\min}$ for three different audit fines $ k \in \{ \$100, \$300, \$500 \}$ and an audit cost of $c = \$50$. Here, we depict the total costs of our audit mechanism under two coalition sizes: (i) when users do not collude, i.e., the coalition size is one, and (ii) when users form coalitions with a maximum coalition size of 150. Furthermore, we also depict a line corresponding to the GAO benchmark of $\$83,333$ of monthly fraudulent claims (termed ``Fraud GAO'') of the benefits program, based on the investigation of GAO. The results depicted in Figure~\ref{fig:GAOFraudk} show that for all the tested fines, the cost of our audit mechanism is lower than that of the setting without audits, and also the GAO benchmark of $\$83,333$ of monthly fraudulent claims for all values of $q_{\min}$. Furthermore, our results show that as the audit fine is increased, the total cost of our proposed audit mechanism is reduced, as users will misreport with a lower probability at a higher audit fine (see Corollary~\ref{cor:misrepotingProbBound}), and thus the total excess payments corresponding to our proposed audit mechanism will consequently be lower (see Corollary~\ref{cor:equilibrium_suboptimality}).

\vspace{-15pt}

\begin{figure}[tbh!]
  \centering
  \begin{subfigure}[b]{0.3\columnwidth}
      \include{Figs/GAO_Study_k_100_c_50}
  \end{subfigure} \hspace{20pt}
  \begin{subfigure}[b]{0.3\columnwidth}
      \include{Figs/GAO_Study_k_300_c_50}
  \end{subfigure} \hspace{5pt}
  \begin{subfigure}[b]{0.3\columnwidth}
      \include{Figs/GAO_Study_k_500_c_50}
  \end{subfigure}
     \vspace{-40pt}
    \caption{{\footnotesize \sf Comparison of the total costs of the no audit mechanism to that of our proposed audit mechanism for three values of the audit fines $\k = 100$ (left), $k = 300$ (center), and $k = 500$ (right) as the probability of the low type $q_{\min}$ varies. In particular, the tested mechanisms include (i) the no audit mechanism, (ii) our proposed audit mechanism with a coalition size of one, and (iii) our proposed audit mechanism with a coalition size of 150. Furthermore, each plot depicts a line corresponding to the monthly fraud uncovered by GAO for the 4000 federal employees in the National Capital Region.}} 
    \label{fig:GAOFraudk}
\end{figure}

\subsection{Misreporting Probability} \label{apdx:misreportingPlots} 

In this section, we depict the variation in the maximum misreporting probability of users for different ranges of the audit fine $\k$ and audit cost $\c$ for the model parameters presented in Section~\ref{sec:experiment-setup} for three different values of $q_{\min} \in \{ 0.25, 0.5, 0.75 \}$ in Figure~\ref{fig:misreportingProb}. In particular, we let the audit fine $\k$ vary in $\$100$ increments between $\$100$ to $\$1000$ and the audit cost $\c$ vary between $\$25$ to $\$150$ and observe from Figure~\ref{fig:misreportingProb} that as the audit fine increases and the audit cost reduces, the maximum misreporting probability also reduces (see Corollary~\ref{cor:misrepotingProbBound}). Furthermore, we observe that the misreporting probability of users increases as $q_{\min}$ decreases, which follows as the administrator is less likely to conduct audits when a larger proportion of the population are of the high type, which enables lower type users to misreport more when $q_{\min}$ is small.

\vspace{-10pt}
\begin{figure}[tbh!]
  \centering
  \begin{subfigure}[b]{0.3\columnwidth}
%
%
\begin{tikzpicture}

\begin{axis}[%
width=1.1in,
height=1.1in,
at={(0in,0in)},
scale only axis,
xmin=0,
xmax=150,
tick align=outside,
xlabel style={font=\color{white!15!black}, rotate = 38, font = \small},
xlabel={Audit Cost (c)},
ymin=0,
ymax=1000,
ylabel style={font=\color{white!15!black}, rotate = 340, font = \small},
ylabel={Fine (k)},
zmin=0,
zmax=1,
zlabel style={font=\color{white!15!black}, align = center, font = \small},
zlabel={Max Probability \\ of Misreporting},
title style={font=\bfseries},
title={$q_{\min} = 0.25$},
view={-51.6}{30.6008759124088},
axis background/.style={fill=white},
axis x line*=bottom,
axis y line*=left,
axis z line*=left,
xmajorgrids,
ymajorgrids,
zmajorgrids,
legend style={at={(1.03,1)}, anchor=north west, legend cell align=left, align=left, draw=white!15!black}
]

\addplot3[%
surf,
shader=flat corner, draw=black, z buffer=sort, colormap={mymap}{[1pt] rgb(0pt)=(0.2422,0.1504,0.6603); rgb(1pt)=(0.25039,0.164995,0.707614); rgb(2pt)=(0.257771,0.181781,0.751138); rgb(3pt)=(0.264729,0.197757,0.795214); rgb(4pt)=(0.270648,0.214676,0.836371); rgb(5pt)=(0.275114,0.234238,0.870986); rgb(6pt)=(0.2783,0.255871,0.899071); rgb(7pt)=(0.280333,0.278233,0.9221); rgb(8pt)=(0.281338,0.300595,0.941376); rgb(9pt)=(0.281014,0.322757,0.957886); rgb(10pt)=(0.279467,0.344671,0.971676); rgb(11pt)=(0.275971,0.366681,0.982905); rgb(12pt)=(0.269914,0.3892,0.9906); rgb(13pt)=(0.260243,0.412329,0.995157); rgb(14pt)=(0.244033,0.435833,0.998833); rgb(15pt)=(0.220643,0.460257,0.997286); rgb(16pt)=(0.196333,0.484719,0.989152); rgb(17pt)=(0.183405,0.507371,0.979795); rgb(18pt)=(0.178643,0.528857,0.968157); rgb(19pt)=(0.176438,0.549905,0.952019); rgb(20pt)=(0.168743,0.570262,0.935871); rgb(21pt)=(0.154,0.5902,0.9218); rgb(22pt)=(0.146029,0.609119,0.907857); rgb(23pt)=(0.138024,0.627629,0.89729); rgb(24pt)=(0.124814,0.645929,0.888343); rgb(25pt)=(0.111252,0.6635,0.876314); rgb(26pt)=(0.0952095,0.679829,0.859781); rgb(27pt)=(0.0688714,0.694771,0.839357); rgb(28pt)=(0.0296667,0.708167,0.816333); rgb(29pt)=(0.00357143,0.720267,0.7917); rgb(30pt)=(0.00665714,0.731214,0.766014); rgb(31pt)=(0.0433286,0.741095,0.73941); rgb(32pt)=(0.0963952,0.75,0.712038); rgb(33pt)=(0.140771,0.7584,0.684157); rgb(34pt)=(0.1717,0.766962,0.655443); rgb(35pt)=(0.193767,0.775767,0.6251); rgb(36pt)=(0.216086,0.7843,0.5923); rgb(37pt)=(0.246957,0.791795,0.556743); rgb(38pt)=(0.290614,0.79729,0.518829); rgb(39pt)=(0.340643,0.8008,0.478857); rgb(40pt)=(0.3909,0.802871,0.435448); rgb(41pt)=(0.445629,0.802419,0.390919); rgb(42pt)=(0.5044,0.7993,0.348); rgb(43pt)=(0.561562,0.794233,0.304481); rgb(44pt)=(0.617395,0.787619,0.261238); rgb(45pt)=(0.671986,0.779271,0.2227); rgb(46pt)=(0.7242,0.769843,0.191029); rgb(47pt)=(0.773833,0.759805,0.16461); rgb(48pt)=(0.820314,0.749814,0.153529); rgb(49pt)=(0.863433,0.7406,0.159633); rgb(50pt)=(0.903543,0.733029,0.177414); rgb(51pt)=(0.939257,0.728786,0.209957); rgb(52pt)=(0.972757,0.729771,0.239443); rgb(53pt)=(0.995648,0.743371,0.237148); rgb(54pt)=(0.996986,0.765857,0.219943); rgb(55pt)=(0.995205,0.789252,0.202762); rgb(56pt)=(0.9892,0.813567,0.188533); rgb(57pt)=(0.978629,0.838629,0.176557); rgb(58pt)=(0.967648,0.8639,0.16429); rgb(59pt)=(0.96101,0.889019,0.153676); rgb(60pt)=(0.959671,0.913457,0.142257); rgb(61pt)=(0.962795,0.937338,0.12651); rgb(62pt)=(0.969114,0.960629,0.106362); rgb(63pt)=(0.9769,0.9839,0.0805)}, mesh/rows=6]
table[row sep=crcr, point meta=\thisrow{c}] {%
x	y	z	c\\
25	100	0.576923076923077	0.576923076923077\\
25	200	0.326086956521739	0.326086956521739\\
25	300	0.227272727272727	0.227272727272727\\
25	400	0.174418604651163	0.174418604651163\\
25	500	0.141509433962264	0.141509433962264\\
25	600	0.119047619047619	0.119047619047619\\
25	700	0.102739726027397	0.102739726027397\\
25	800	0.0903614457831325	0.0903614457831325\\
25	900	0.0806451612903226	0.0806451612903226\\
25	1000	0.0728155339805825	0.0728155339805825\\
50	100	1	1\\
50	200	0.731707317073171	0.731707317073171\\
50	300	0.491803278688525	0.491803278688525\\
50	400	0.37037037037037	0.37037037037037\\
50	500	0.297029702970297	0.297029702970297\\
50	600	0.247933884297521	0.247933884297521\\
50	700	0.212765957446809	0.212765957446809\\
50	800	0.186335403726708	0.186335403726708\\
50	900	0.165745856353591	0.165745856353591\\
50	1000	0.149253731343284	0.149253731343284\\
75	100	1	1\\
75	200	1	1\\
75	300	0.803571428571429	0.803571428571429\\
75	400	0.592105263157895	0.592105263157895\\
75	500	0.46875	0.46875\\
75	600	0.387931034482759	0.387931034482759\\
75	700	0.330882352941176	0.330882352941176\\
75	800	0.288461538461538	0.288461538461538\\
75	900	0.255681818181818	0.255681818181818\\
75	1000	0.229591836734694	0.229591836734694\\
100	100	1	1\\
100	200	1	1\\
100	300	1	1\\
100	400	0.845070422535211	0.845070422535211\\
100	500	0.659340659340659	0.659340659340659\\
100	600	0.540540540540541	0.540540540540541\\
100	700	0.458015267175573	0.458015267175573\\
100	800	0.397350993377483	0.397350993377483\\
100	900	0.350877192982456	0.350877192982456\\
100	1000	0.31413612565445	0.31413612565445\\
125	100	1	1\\
125	200	1	1\\
125	300	1	1\\
125	400	1	1\\
125	500	0.872093023255814	0.872093023255814\\
125	600	0.707547169811321	0.707547169811321\\
125	700	0.595238095238095	0.595238095238095\\
125	800	0.513698630136986	0.513698630136986\\
125	900	0.451807228915663	0.451807228915663\\
125	1000	0.403225806451613	0.403225806451613\\
150	100	1	1\\
150	200	1	1\\
150	300	1	1\\
150	400	1	1\\
150	500	1	1\\
150	600	0.891089108910891	0.891089108910891\\
150	700	0.743801652892562	0.743801652892562\\
150	800	0.638297872340426	0.638297872340426\\
150	900	0.559006211180124	0.559006211180124\\
150	1000	0.497237569060773	0.497237569060773\\
};


\end{axis}

\begin{axis}[%
width=0in,
height=0in,
at={(0in,0in)},
scale only axis,
xmin=0,
xmax=1,
ymin=0,
ymax=1,
axis line style={draw=none},
ticks=none,
axis x line*=bottom,
axis y line*=left,
legend style={legend cell align=left, align=left, draw=white!15!black}
]
\end{axis}
\end{tikzpicture}%
  \end{subfigure} \hspace{28pt}
  \begin{subfigure}[b]{0.3\columnwidth}
%
%
\begin{tikzpicture}

\begin{axis}[%
width=1.1in,
height=1.1in,
at={(0in,0in)},
scale only axis,
xmin=0,
xmax=150,
tick align=outside,
xlabel style={font=\color{white!15!black}, rotate = 38, font = \small},
xlabel={Audit Cost (c)},
ymin=0,
ymax=1000,
ylabel style={font=\color{white!15!black}, rotate = 340, font = \small},
ylabel={Fine (k)},
zmin=0,
zmax=1,
zlabel style={font=\color{white!15!black}},
title style={font=\bfseries},
title={$q_{min} = 0.5$},
view={-53.1}{30.6008759124088},
axis background/.style={fill=white},
axis x line*=bottom,
axis y line*=left,
axis z line*=left,
xmajorgrids,
ymajorgrids,
zmajorgrids,
legend style={at={(1.03,1)}, anchor=north west, legend cell align=left, align=left, draw=white!15!black}
]

\addplot3[%
surf,
shader=flat corner, draw=black, z buffer=sort, colormap={mymap}{[1pt] rgb(0pt)=(0.2422,0.1504,0.6603); rgb(1pt)=(0.25039,0.164995,0.707614); rgb(2pt)=(0.257771,0.181781,0.751138); rgb(3pt)=(0.264729,0.197757,0.795214); rgb(4pt)=(0.270648,0.214676,0.836371); rgb(5pt)=(0.275114,0.234238,0.870986); rgb(6pt)=(0.2783,0.255871,0.899071); rgb(7pt)=(0.280333,0.278233,0.9221); rgb(8pt)=(0.281338,0.300595,0.941376); rgb(9pt)=(0.281014,0.322757,0.957886); rgb(10pt)=(0.279467,0.344671,0.971676); rgb(11pt)=(0.275971,0.366681,0.982905); rgb(12pt)=(0.269914,0.3892,0.9906); rgb(13pt)=(0.260243,0.412329,0.995157); rgb(14pt)=(0.244033,0.435833,0.998833); rgb(15pt)=(0.220643,0.460257,0.997286); rgb(16pt)=(0.196333,0.484719,0.989152); rgb(17pt)=(0.183405,0.507371,0.979795); rgb(18pt)=(0.178643,0.528857,0.968157); rgb(19pt)=(0.176438,0.549905,0.952019); rgb(20pt)=(0.168743,0.570262,0.935871); rgb(21pt)=(0.154,0.5902,0.9218); rgb(22pt)=(0.146029,0.609119,0.907857); rgb(23pt)=(0.138024,0.627629,0.89729); rgb(24pt)=(0.124814,0.645929,0.888343); rgb(25pt)=(0.111252,0.6635,0.876314); rgb(26pt)=(0.0952095,0.679829,0.859781); rgb(27pt)=(0.0688714,0.694771,0.839357); rgb(28pt)=(0.0296667,0.708167,0.816333); rgb(29pt)=(0.00357143,0.720267,0.7917); rgb(30pt)=(0.00665714,0.731214,0.766014); rgb(31pt)=(0.0433286,0.741095,0.73941); rgb(32pt)=(0.0963952,0.75,0.712038); rgb(33pt)=(0.140771,0.7584,0.684157); rgb(34pt)=(0.1717,0.766962,0.655443); rgb(35pt)=(0.193767,0.775767,0.6251); rgb(36pt)=(0.216086,0.7843,0.5923); rgb(37pt)=(0.246957,0.791795,0.556743); rgb(38pt)=(0.290614,0.79729,0.518829); rgb(39pt)=(0.340643,0.8008,0.478857); rgb(40pt)=(0.3909,0.802871,0.435448); rgb(41pt)=(0.445629,0.802419,0.390919); rgb(42pt)=(0.5044,0.7993,0.348); rgb(43pt)=(0.561562,0.794233,0.304481); rgb(44pt)=(0.617395,0.787619,0.261238); rgb(45pt)=(0.671986,0.779271,0.2227); rgb(46pt)=(0.7242,0.769843,0.191029); rgb(47pt)=(0.773833,0.759805,0.16461); rgb(48pt)=(0.820314,0.749814,0.153529); rgb(49pt)=(0.863433,0.7406,0.159633); rgb(50pt)=(0.903543,0.733029,0.177414); rgb(51pt)=(0.939257,0.728786,0.209957); rgb(52pt)=(0.972757,0.729771,0.239443); rgb(53pt)=(0.995648,0.743371,0.237148); rgb(54pt)=(0.996986,0.765857,0.219943); rgb(55pt)=(0.995205,0.789252,0.202762); rgb(56pt)=(0.9892,0.813567,0.188533); rgb(57pt)=(0.978629,0.838629,0.176557); rgb(58pt)=(0.967648,0.8639,0.16429); rgb(59pt)=(0.96101,0.889019,0.153676); rgb(60pt)=(0.959671,0.913457,0.142257); rgb(61pt)=(0.962795,0.937338,0.12651); rgb(62pt)=(0.969114,0.960629,0.106362); rgb(63pt)=(0.9769,0.9839,0.0805)}, mesh/rows=6]
table[row sep=crcr, point meta=\thisrow{c}] {%
x	y	z	c\\
25	100	0.192307692307692	0.192307692307692\\
25	200	0.108695652173913	0.108695652173913\\
25	300	0.0757575757575758	0.0757575757575758\\
25	400	0.0581395348837209	0.0581395348837209\\
25	500	0.0471698113207547	0.0471698113207547\\
25	600	0.0396825396825397	0.0396825396825397\\
25	700	0.0342465753424658	0.0342465753424658\\
25	800	0.0301204819277108	0.0301204819277108\\
25	900	0.0268817204301075	0.0268817204301075\\
25	1000	0.0242718446601942	0.0242718446601942\\
50	100	0.476190476190476	0.476190476190476\\
50	200	0.24390243902439	0.24390243902439\\
50	300	0.163934426229508	0.163934426229508\\
50	400	0.123456790123457	0.123456790123457\\
50	500	0.099009900990099	0.099009900990099\\
50	600	0.0826446280991736	0.0826446280991736\\
50	700	0.0709219858156028	0.0709219858156028\\
50	800	0.062111801242236	0.062111801242236\\
50	900	0.0552486187845304	0.0552486187845304\\
50	1000	0.0497512437810945	0.0497512437810945\\
75	100	0.9375	0.9375\\
75	200	0.416666666666667	0.416666666666667\\
75	300	0.267857142857143	0.267857142857143\\
75	400	0.197368421052632	0.197368421052632\\
75	500	0.15625	0.15625\\
75	600	0.129310344827586	0.129310344827586\\
75	700	0.110294117647059	0.110294117647059\\
75	800	0.0961538461538462	0.0961538461538462\\
75	900	0.0852272727272727	0.0852272727272727\\
75	1000	0.076530612244898	0.076530612244898\\
100	100	1	1\\
100	200	0.645161290322581	0.645161290322581\\
100	300	0.392156862745098	0.392156862745098\\
100	400	0.28169014084507	0.28169014084507\\
100	500	0.21978021978022	0.21978021978022\\
100	600	0.18018018018018	0.18018018018018\\
100	700	0.152671755725191	0.152671755725191\\
100	800	0.132450331125828	0.132450331125828\\
100	900	0.116959064327485	0.116959064327485\\
100	1000	0.104712041884817	0.104712041884817\\
125	100	1	1\\
125	200	0.961538461538462	0.961538461538462\\
125	300	0.543478260869565	0.543478260869565\\
125	400	0.378787878787879	0.378787878787879\\
125	500	0.290697674418605	0.290697674418605\\
125	600	0.235849056603774	0.235849056603774\\
125	700	0.198412698412698	0.198412698412698\\
125	800	0.171232876712329	0.171232876712329\\
125	900	0.150602409638554	0.150602409638554\\
125	1000	0.134408602150538	0.134408602150538\\
150	100	1	1\\
150	200	1	1\\
150	300	0.731707317073171	0.731707317073171\\
150	400	0.491803278688525	0.491803278688525\\
150	500	0.37037037037037	0.37037037037037\\
150	600	0.297029702970297	0.297029702970297\\
150	700	0.247933884297521	0.247933884297521\\
150	800	0.212765957446809	0.212765957446809\\
150	900	0.186335403726708	0.186335403726708\\
150	1000	0.165745856353591	0.165745856353591\\
};

\end{axis}

\begin{axis}[%
width=0in,
height=0in,
at={(0in,0in)},
scale only axis,
xmin=0,
xmax=1,
ymin=0,
ymax=1,
axis line style={draw=none},
ticks=none,
axis x line*=bottom,
axis y line*=left,
legend style={legend cell align=left, align=left, draw=white!15!black}
]
\end{axis}
\end{tikzpicture}%
  \end{subfigure} \hspace{5pt}
  \begin{subfigure}[b]{0.3\columnwidth}
%
%
\begin{tikzpicture}

\begin{axis}[%
width=1.1in,
height=1.1in,
at={(0in,0in)},
scale only axis,
xmin=0,
xmax=150,
tick align=outside,
xlabel style={font=\color{white!15!black}, rotate = 20, font = \small},
xlabel={Audit Cost (c)},
ymin=0,
ymax=1000,
ylabel style={font=\color{white!15!black}, rotate = 330, font = \small},
ylabel={Fine (k)},
zmin=0,
zmax=1,
zlabel style={font=\color{white!15!black}},
view={-37.5}{30},
title style={font=\bfseries},
title={$q_{min} = 0.75$},
axis background/.style={fill=white},
axis x line*=bottom,
axis y line*=left,
axis z line*=left,
xmajorgrids,
ymajorgrids,
zmajorgrids,
legend style={at={(1.03,1)}, anchor=north west, legend cell align=left, align=left, draw=white!15!black}
]

\addplot3[%
surf,
shader=flat corner, draw=black, z buffer=sort, colormap={mymap}{[1pt] rgb(0pt)=(0.2422,0.1504,0.6603); rgb(1pt)=(0.25039,0.164995,0.707614); rgb(2pt)=(0.257771,0.181781,0.751138); rgb(3pt)=(0.264729,0.197757,0.795214); rgb(4pt)=(0.270648,0.214676,0.836371); rgb(5pt)=(0.275114,0.234238,0.870986); rgb(6pt)=(0.2783,0.255871,0.899071); rgb(7pt)=(0.280333,0.278233,0.9221); rgb(8pt)=(0.281338,0.300595,0.941376); rgb(9pt)=(0.281014,0.322757,0.957886); rgb(10pt)=(0.279467,0.344671,0.971676); rgb(11pt)=(0.275971,0.366681,0.982905); rgb(12pt)=(0.269914,0.3892,0.9906); rgb(13pt)=(0.260243,0.412329,0.995157); rgb(14pt)=(0.244033,0.435833,0.998833); rgb(15pt)=(0.220643,0.460257,0.997286); rgb(16pt)=(0.196333,0.484719,0.989152); rgb(17pt)=(0.183405,0.507371,0.979795); rgb(18pt)=(0.178643,0.528857,0.968157); rgb(19pt)=(0.176438,0.549905,0.952019); rgb(20pt)=(0.168743,0.570262,0.935871); rgb(21pt)=(0.154,0.5902,0.9218); rgb(22pt)=(0.146029,0.609119,0.907857); rgb(23pt)=(0.138024,0.627629,0.89729); rgb(24pt)=(0.124814,0.645929,0.888343); rgb(25pt)=(0.111252,0.6635,0.876314); rgb(26pt)=(0.0952095,0.679829,0.859781); rgb(27pt)=(0.0688714,0.694771,0.839357); rgb(28pt)=(0.0296667,0.708167,0.816333); rgb(29pt)=(0.00357143,0.720267,0.7917); rgb(30pt)=(0.00665714,0.731214,0.766014); rgb(31pt)=(0.0433286,0.741095,0.73941); rgb(32pt)=(0.0963952,0.75,0.712038); rgb(33pt)=(0.140771,0.7584,0.684157); rgb(34pt)=(0.1717,0.766962,0.655443); rgb(35pt)=(0.193767,0.775767,0.6251); rgb(36pt)=(0.216086,0.7843,0.5923); rgb(37pt)=(0.246957,0.791795,0.556743); rgb(38pt)=(0.290614,0.79729,0.518829); rgb(39pt)=(0.340643,0.8008,0.478857); rgb(40pt)=(0.3909,0.802871,0.435448); rgb(41pt)=(0.445629,0.802419,0.390919); rgb(42pt)=(0.5044,0.7993,0.348); rgb(43pt)=(0.561562,0.794233,0.304481); rgb(44pt)=(0.617395,0.787619,0.261238); rgb(45pt)=(0.671986,0.779271,0.2227); rgb(46pt)=(0.7242,0.769843,0.191029); rgb(47pt)=(0.773833,0.759805,0.16461); rgb(48pt)=(0.820314,0.749814,0.153529); rgb(49pt)=(0.863433,0.7406,0.159633); rgb(50pt)=(0.903543,0.733029,0.177414); rgb(51pt)=(0.939257,0.728786,0.209957); rgb(52pt)=(0.972757,0.729771,0.239443); rgb(53pt)=(0.995648,0.743371,0.237148); rgb(54pt)=(0.996986,0.765857,0.219943); rgb(55pt)=(0.995205,0.789252,0.202762); rgb(56pt)=(0.9892,0.813567,0.188533); rgb(57pt)=(0.978629,0.838629,0.176557); rgb(58pt)=(0.967648,0.8639,0.16429); rgb(59pt)=(0.96101,0.889019,0.153676); rgb(60pt)=(0.959671,0.913457,0.142257); rgb(61pt)=(0.962795,0.937338,0.12651); rgb(62pt)=(0.969114,0.960629,0.106362); rgb(63pt)=(0.9769,0.9839,0.0805)}, mesh/rows=6]
table[row sep=crcr, point meta=\thisrow{c}] {%
x	y	z	c\\
25	100	0.0641025641025641	0.0641025641025641\\
25	200	0.036231884057971	0.036231884057971\\
25	300	0.0252525252525253	0.0252525252525253\\
25	400	0.0193798449612403	0.0193798449612403\\
25	500	0.0157232704402516	0.0157232704402516\\
25	600	0.0132275132275132	0.0132275132275132\\
25	700	0.0114155251141553	0.0114155251141553\\
25	800	0.0100401606425703	0.0100401606425703\\
25	900	0.00896057347670251	0.00896057347670251\\
25	1000	0.00809061488673139	0.00809061488673139\\
50	100	0.158730158730159	0.158730158730159\\
50	200	0.0813008130081301	0.0813008130081301\\
50	300	0.0546448087431694	0.0546448087431694\\
50	400	0.0411522633744856	0.0411522633744856\\
50	500	0.033003300330033	0.033003300330033\\
50	600	0.0275482093663912	0.0275482093663912\\
50	700	0.0236406619385343	0.0236406619385343\\
50	800	0.020703933747412	0.020703933747412\\
50	900	0.0184162062615101	0.0184162062615101\\
50	1000	0.0165837479270315	0.0165837479270315\\
75	100	0.3125	0.3125\\
75	200	0.138888888888889	0.138888888888889\\
75	300	0.0892857142857143	0.0892857142857143\\
75	400	0.0657894736842105	0.0657894736842105\\
75	500	0.0520833333333333	0.0520833333333333\\
75	600	0.0431034482758621	0.0431034482758621\\
75	700	0.0367647058823529	0.0367647058823529\\
75	800	0.032051282051282	0.032051282051282\\
75	900	0.0284090909090909	0.0284090909090909\\
75	1000	0.0255102040816327	0.0255102040816327\\
100	100	0.606060606060606	0.606060606060606\\
100	200	0.21505376344086	0.21505376344086\\
100	300	0.130718954248366	0.130718954248366\\
100	400	0.0938967136150235	0.0938967136150235\\
100	500	0.0732600732600733	0.0732600732600733\\
100	600	0.0600600600600601	0.0600600600600601\\
100	700	0.0508905852417303	0.0508905852417303\\
100	800	0.0441501103752759	0.0441501103752759\\
100	900	0.0389863547758285	0.0389863547758285\\
100	1000	0.0349040139616056	0.0349040139616056\\
125	100	1	1\\
125	200	0.320512820512821	0.320512820512821\\
125	300	0.181159420289855	0.181159420289855\\
125	400	0.126262626262626	0.126262626262626\\
125	500	0.0968992248062016	0.0968992248062016\\
125	600	0.0786163522012579	0.0786163522012579\\
125	700	0.0661375661375661	0.0661375661375661\\
125	800	0.0570776255707763	0.0570776255707763\\
125	900	0.0502008032128514	0.0502008032128514\\
125	1000	0.0448028673835125	0.0448028673835125\\
150	100	1	1\\
150	200	0.476190476190476	0.476190476190476\\
150	300	0.24390243902439	0.24390243902439\\
150	400	0.163934426229508	0.163934426229508\\
150	500	0.123456790123457	0.123456790123457\\
150	600	0.099009900990099	0.099009900990099\\
150	700	0.0826446280991736	0.0826446280991736\\
150	800	0.0709219858156028	0.0709219858156028\\
150	900	0.062111801242236	0.062111801242236\\
150	1000	0.0552486187845304	0.0552486187845304\\
};


\end{axis}

\begin{axis}[%
width=0in,
height=0in,
at={(0in,0in)},
scale only axis,
xmin=0,
xmax=1,
ymin=0,
ymax=1,
axis line style={draw=none},
ticks=none,
axis x line*=bottom,
axis y line*=left,
legend style={legend cell align=left, align=left, draw=white!15!black}
]
\end{axis}
\end{tikzpicture}%
  \end{subfigure}
     \vspace{-50pt}
    \caption{{\footnotesize \sf Variation in the maximum misreporting probability of users for a range of audit fines $\k$ and audit costs $\c$ for three different values of $q_{\min}$. }} 
    \label{fig:misreportingProb}
\end{figure}

\end{document}